\pdfoutput=1

\documentclass[letterpaper,11pt]{article}

\usepackage{amsmath}
\usepackage{amsfonts}
\usepackage{amssymb}
\usepackage{amsthm}
\usepackage{fullpage}
\usepackage{graphicx}

\usepackage{xspace}

\newtheorem{theorem}{Theorem}[section]
\newtheorem{lemma}[theorem]{Lemma}
\newtheorem{definition}[theorem]{Definition}
\newtheorem{corollary}[theorem]{Corollary}
\newtheorem{condition}[theorem]{Condition}

\newcommand{\N}{\mathbb{N}}
\newcommand{\R}{\mathbb{R}}
\newcommand{\E}{\mathbb{E}}
\newcommand{\Ss}{\mathbb{S}}
\newcommand{\BO}{\mathcal{O}}
\newcommand{\sr}[2]{\stackrel{\eqref{#1}}{#2}}

\newcommand{\acc}{\emph{accept}\xspace}
\newcommand{\slp}{\emph{sleep}\xspace}
\newcommand{\rdy}{\emph{ready}\xspace}
\newcommand{\prop}{\emph{propose}\xspace}
\newcommand{\rec}{\emph{recover}\xspace}
\newcommand{\init}{\emph{init}\xspace}
\newcommand{\join}{\emph{join}\xspace}
\newcommand{\none}{\emph{none}\xspace}
\newcommand{\res}{\emph{resync}\xspace}
\newcommand{\supp}{\emph{supp}\xspace}
\newcommand{\dorm}{\emph{dormant}\xspace}
\newcommand{\srw}{\ensuremath{\emph{sleep}\rightarrow\emph{waking}}\xspace}
\newcommand{\srr}{\ensuremath{\emph{supp}\rightarrow\emph{resync}}\xspace}
\newcommand{\act}{\emph{active}\xspace}
\newcommand{\pass}{\emph{passive}\xspace}
\newcommand{\wake}{\emph{waking}\xspace}
\newcommand{\sus}{\emph{suspect}\xspace}
\newcommand{\darts}{\mbox{\sc{darts}}\xspace}
\newcommand{\pulse}{\mbox{\sc{pulse}}\xspace}

\DeclareMathOperator{\Mem}{Mem}
\DeclareMathOperator{\Time}{Time}

\newcommand{\namedref}[2]{\hyperref[#2]{#1~\ref*{#2}}}
\newcommand{\sectionref}[1]{\namedref{Section}{#1}}

\newcommand{\theoremref}[1]{\namedref{Theorem}{#1}}
\newcommand{\defref}[1]{\namedref{Definition}{#1}}
\newcommand{\figureref}[1]{\namedref{Figure}{#1}}
\newcommand{\figref}[1]{\namedref{Figure}{#1}}

\newcommand{\lemmaref}[1]{\namedref{Lemma}{#1}}

\newcommand{\corollaryref}[1]{\namedref{Corollary}{#1}}

\newcommand{\conditionref}[1]{\namedref{Condition}{#1}}
\newcommand{\equalityref}[1]{\hyperref[#1]{Equality~\eqref{#1}}}
\newcommand{\inequalityref}[1]{\hyperref[#1]{Inequality~\eqref{#1}}}

\newcommand{\emn}[1]{{\em #1}\/}
\newcommand{\LRa}{\Leftrightarrow}

\newcommand{\tick}[1]{{\tt Tick\,(#1)}}

\newcommand{\bfno}[1]{{\noindent\bf #1}\/}

\newcommand{\theterm}{(R_1+(\vartheta+2)T_1+T_2/\vartheta+(8\vartheta+9)d)}

\begin{document}
\setcounter{tocdepth}{3}

\title{\texorpdfstring{{\Large \textbf{\LARGE F}ault-tolerant \textbf{\LARGE
A}lgorithms for \textbf{\LARGE T}ick-Generation
in \textbf{\LARGE A}synchronous \textbf{\LARGE L}ogic:}\\
Robust Pulse Generation}{
Fault-tolerant Algorithms for Tick-Generation
in Asynchronous Logic: Robust Pulse Generation}}

\author{Danny Dolev
, Matthias F\"ugger
, Christoph Lenzen
, and Ulrich Schmid
}

\date{}

\maketitle

\thispagestyle{empty}

\begin{abstract}

Today's hardware technology presents a new challenge in designing robust
systems. Deep submicron VLSI technology introduced transient and permanent
faults that were never considered in low-level system designs in the
past. Still, robustness of that part of the system is crucial and needs to be
guaranteed for any successful product. Distributed systems, on the other
hand, have been dealing with similar issues for decades. However, neither the
basic abstractions nor the complexity of contemporary fault-tolerant
distributed algorithms match the peculiarities of hardware implementations.

This paper is intended to be part of an attempt striving to
overcome this gap between theory and practice for the clock synchronization
problem. Solving this task sufficiently well will allow to build a very robust
high-precision clocking system for hardware designs like systems-on-chips
in critical applications. As our first building block, we describe 
and prove correct a novel Byzantine fault-tolerant self-stabilizing pulse 
synchronization protocol, which
can be implemented using standard asynchronous digital logic. Despite the strict
limitations introduced by hardware designs, it offers optimal resilience and
smaller complexity than all existing protocols.
\end{abstract}

\newpage

\section{Introduction \& Related Work}
\label{sec:Intro}

With today's deep submicron technology running at GHz clock
     speeds~\cite{ITRS07}, disseminating the high-speed clock
     throughout a \emph{very large scale integrated} (VLSI) circuit,
     with negligible skew, is difficult and
     costly~\cite{BK09:iccd,BZMLCLD02,Fri01,MDM04,Resetal01}.
Systems-on-chip are hence increasingly designed \emph{globally
     asynchronous locally synchronous} (GALS) \cite{Cha84}, where
     different parts of the chip use different local clock signals.
Two main types of clocking schemes for GALS systems exist, namely, (i)
     those where the local clock signals are unrelated, and (ii)
     multi-synchronous ones that provide a certain degree of synchrony
     between local clock signals \cite{SG03,TGL07}.

GALS systems clocked by type (i) permanently bear the risk of
     \emph{metastable upsets} when conveying information from one
     clock domain to another.
To explain the issue, consider a physical implementation of a bistable
     storage element, like a register cell, which can be accessed by
     read and write operations concurrently.
It can be shown that two operations (like two writes with different
     values) occurring very closely to each other can cause the
     storage cell to attain neither of its two stable states for an
     unbounded time~\cite{Mar81}, and thereby, during an unbounded
     time afterwards, successive reads may return none of the stable
     states.
Although the probability of a single upset is very small, one has to
     take into account that every bit of transmitted information
     across clock domains is a candidate for an upset.
Elaborate synchronizers~\cite{DB99,KBY02,PM95} are the only means for
     achieving an acceptably low probability for metastable upsets
     here.

This problem can be circumvented in clocking schemes of type (ii):
     Common synchrony properties offered by multi-synchronous clocking
     systems are: 
\begin{itemize}
\item \emph{bounded skew}, i.e., bounded maximum time between the
     occurence of any two matching clock transitions of any two local
     clock signals.
Thereby, in classic clock synchronization, two clock transitions
     are matching iff they are both the $k^\text{th}$, $k\geq 1$, clock
     transition of a local clock.

\item \emph{bounded accuracy}, i.e., bounded minimum and maximum time
     between the occurence of any two successive clock transitions of
     any local clock signal.
\end{itemize}

Type (ii) clocking schemes  are particularly beneficial from a
     designer's point of view, since  they combine the convenient local
     synchrony of a GALS system with a global time base across the
     whole chip.
It has been shown in~\cite{PHS09:SSS} that these properties indeed
     facilitate metastability-free high-speed communication across
     clock domains.

The decreasing structure sizes of deep submicron technology also resulted in an
increased likelihood of chip components failing during operation: Reduced
voltage swing and smaller critical charges make circuits more susceptible to
ionized particle hits, crosstalk, and electromagnetic
interference~\cite{Con03,GEBC06}. \emph{Fault-tolerance} hence becomes an
increasingly pressing issue in chip design. Unfortunately, faulty components may
behave non-benign in many ways. They may perform signal transitions at arbitrary
times and even convey inconsistent information to their successor components if
their outgoing communication channels are affected by a failure. This forces
to model faulty components as unrestricted, i.e., Byzantine,
if a high fault coverage is to be guaranteed.

The \darts fault-tolerant clock generation
     approach~\cite{FDS10:edcc,FSFK06:edcc} developed by some of the
     authors of this paper is a Byzantine fault-tolerant
     multi-synchronous clocking scheme.
\darts comprises a set of modules, each of which generates a local
     clock signal for a single clock domain.
The \darts modules (nodes) are synchronized to each other to within a
     few clock cycles.
This is achieved by exchanging binary clock signals only, via single
     wires.
The basic idea behind \darts is to employ a simple fault-tolerant
     distributed algorithm~\cite{WS09:DC}---based on
     Srikanth~\&~Toueg's consistent broadcasting
     primitive~\cite{ST87}---implemented in asynchronous digital
     logic.
An important property of the \darts clocking scheme is that it
     guarantees that no metastable upsets occur during fault-free
     executions.
For executions with faults, metastable upsets cannot be ruled out:
     Since Byzantine faulty components are allowed to issue unrelated
     read and write accesses by definition, the same arguments as for
     clocking schemes of type (i) apply.
However, in~\cite{FFS09:ASYNC09}, it was shown that by proper chip
     design the probability of a Byzantine component leading to a
     metastable upset of \darts can be made arbitrarily small.

Although both theoretical analysis and experimental evaluation revealed many
attractive additional features of \darts, like guaranteed startup, automatic 
adaption to current
operating conditions, etc., there is room for improvement. The most obvious
drawback of \darts is its inability to support late joining and restarting of
nodes, and, more generally, its lack of self-stabilization properties. If, for
some reasons, more than a third of the \darts nodes ever become faulty, the
system cannot be guaranteed to resume normal operation even if all failures
cease. Even worse, simple transient faults such as radiation- or 
crosstalk-induced additional (or omitted) clock ticks accumulate 
over time to arbitrarily large skews in an otherwise benign execution.

Byzantine-tolerant self-stabilization, on the other hand, is  the major strength
of a number of
protocols~\cite{BDH08:podc,DD06,DolWelSSBYZCS04,HDD06:SSS,Mal06:SSS}
primarily devised for distributed systems.
Of particular interest in the above context is the work on self-stabilizing
\emph{pulse synchronization}, where the purpose is to generate well-separated
anonymous pulses that are synchronized at all correct nodes. This
facilitates self-stabilizing clock synchronization, as agreement on a time
window permits to simulate a synchronous protocol in a bounded-delay system.
Beyond optimal (i.e., $\lceil n/3\rceil-1$, c.f.~\cite{PSL80}) resilience, an
attractive feature of these protocols is a small stabilization
time~\cite{BDH08:podc,DD06,HDD06:SSS,Mal06:SSS}, which is crucial for
applications with stringent availability requirements. In particular,
\cite{BDH08:podc} synchronizes clocks in expected constant time in a synchronous
system. Given any pulse synchronization protocol stabilizing in a bounded-delay
system in expected time $T$, this implies an expected $(T+\BO(1))$-stabilizing
clock synchronization protocol.

Nonetheless, it remains open whether a (with respect to the number of nodes $n$)
sublinear convergence time can be achieved: While the classical consensus lower
bound of $f+1$ rounds for synchronous, deterministic algorithms in a system with
$f< n/3$ faults~\cite{FL82} proves that \emph{exact} agreement on a clock value
requires at least $f+1\in \Omega(n)$ deterministic rounds, one has to face the
fact that only approximate agreement on the current time is achievable in a
bounded-delay system anyway. However, no non-trivial lower bounds on approximate
deterministic synchronization or the exact problem with randomization are known
by now.

Note that existing synchronization algorithms, in particular those
     that do not rely on pulse synchronization, have deficiencies
     rendering them unsuitable in our context.
For example, they have exponential convergence
     time~\cite{DolWelSSBYZCS04}, require the relative drift of the
     nodes' local clocks to be very
     small~\cite{DDP03:SSS,Mal06:SSS},\footnote{Note that it is too
     costly and space consuming to equip each node with a quartz
     oscillator.
Simple digital oscillators, like inverters with feedback, in turn
     exhibit drifts of at least several percent, which heavily vary
     with operating conditions.} provide larger skew only
     \cite{Mal06:SSS} or make use of linear-sized
     messages~\cite{DD06}.
Furthermore, standard models used by the distributed systems community
     do not account for metastability, resulting in the same to be
     true for the existing solutions.

It is hence natural to explore ways of combining and extending the above lines
of research. The present paper is the first step towards this goal.

\textbf{Detailed contributions.}
We describe and prove correct the novel FATAL pulse synchronization 
protocol, which facilitates a direct implementation in standard asynchronous digital 
logic. It self-stabilizes
within $\BO(n)$ time with probability $1-2^{n-f}$,\footnote{Note that the
algorithm from~\cite{BDH08:podc} achieving an expected constant stabilization
time in a synchronous model needs to run for $\Omega(n)$ rounds to ensure the
same probability of stabilization.} in the presence of up to $\lceil
n/3\rceil-1$ Byzantine faulty nodes, and is metastability-free by construction
after stabilization in failure-free runs. While executing the protocol,
non-faulty nodes broadcast a constant number of bits in constant time. In terms
of distributed message complexity, this implies that stabilization is achieved
after broadcasting $\BO(n)$ messages of size $\BO(1)$, improving by factor
$\Omega(n)$ on the number of bits transmitted by previous
algorithms.\footnote{We remark that~\cite{Mal06:SSS} achieves the same
complexity, but considers a much simpler model. In particular, \emph{all}
communication is restricted to broadcasts, i.e., all nodes observe the same
behaviour of a given other node, even if it is faulty.} The protocol can sustain
large relative clock drifts of more than $10\%$, which is crucial if the local
clock sources are simple ring oscillators (uncompensated ring oscillators suffer
from clock drifts of up to $9\%$ \cite{SAA06}). If the number of faults is not
overwhelming, i.e., a majority of at least $n-f$ nodes continues to execute the
protocol in an orderly fashion, recovering nodes and late joiners
(re)synchronize in constant time. This property is highly desirable in practical
systems, in particular in combination with Byzantine fault-tolerance: Even if
nodes randomly experience transient faults on a regular basis, quick recovery
ensures that the mean time until failure of the system as a whole is
substantially increased. All this is achieved against a powerful adversary that,
at time $t$, knows the whole history of the system up to time $t+\varepsilon$
(where $\varepsilon>0$ is infinitesimally small) and does not need to choose the
set of faulty nodes in advance. Apart from bounded drifts and communication
delays, our solution solely requires that receivers can unambiguously identify
the sender of a message, which is a property that arises naturally in hardware
designs.

We also describe how the pulse synchronization protocol can be
     implemented using asynchronous digital logic.
Moreover, we sketch how the pulse synchronization protocol will be
     integrated with \darts clocks to build a high-precision
     self-stabilizing clocking system for multi-synchronous GALS.
The basic idea of our integration is to let the pulse synchronization
     protocol non-intrusively monitor the operation of \darts clocks
     and to recover \darts clocks that run abnormally.
Like the original \darts, the joint system is metastability-free in
     failure-free runs after stabilization.
During stabilization, the fact that nodes merely undergo a constant
     number of state transitions in constant time ensures a very small
     probability of metastable upsets.

\section{Model}\label{sec:model}

Our formal framework will be tied to the peculiarities of hardware designs,
which consist of modules that \emph{continuously}\footnote{In sharp contrast to
classic distributed computing models, there is no computationally complex
discrete zero-time state-transition here.} compute their output signals based on
their input signals. Following \cite{Fue10:diss,FS10:TR}, we define (the trace
of) a \emn{signal} to be a timed event trace over a finite alphabet $\Ss$ of
possible signal states: Formally, signal $\sigma \subseteq \Ss \times \R_0^+$.
All times and time intervals refer to a global \emph{reference time} taken from
$\R_0^+$, that is, signals describe the system's behaviour from time~0 on. The
elements of $\sigma$ are called \emn{events}, and for each event $(s,t)$ we call
$s$ the \emn{state of event} $(s,t)$ and $t$ the \emn{time of event} $(s,t)$. In
general, a signal $\sigma$ is required to fulfill the following conditions: (i)
for each time interval $[t^-,t^+]\subseteq \R_0^+$ of finite length, the number
of events in $\sigma$ with times within $[t^-,t^+]$ is finite, (ii) from $(s,t)
\in \sigma$ and $(s',t)\in \sigma$ follows that $s=s'$, and (iii) there exists
an event at time~$0$ in $\sigma$.

Note that our definition allows for events $(s,t)$ and $(s,t')\in \sigma$, where
$t < t'$, without having an event $(s',t'')\in \sigma$ with $s'\neq s$ and $t <
t'' < t'$. In this case, we call event $(s,t')$ \emph{idempotent}. Two signals
$\sigma$ and $\sigma'$ are \emph{equivalent}, iff they differ in idempotent
events only. We identify all signals of an equivalence class, as they describe
the same physical signal. Each equivalence class $[\sigma]$ of signals contains
a unique signal $\sigma_0$ having no idempotent events. We say that \emn{signal
$\sigma$ switches to} $s$ at time~$t$ iff event $(s,t)\in \sigma_0$.

The \emn{state of signal}~$\sigma$ at time $t \in \R_0^+$, denoted by
$\sigma(t)$, is given by the state of the event with the maximum time not
greater than $t$.\footnote{To facilitate intuition, we here slightly abuse
notation, as this way $\sigma$ denotes both a function of time and the signal
(trace), which is a subset of $\Ss \times \R_0^+$. Whenever referring to
$\sigma$, we will talk of the signal, not the state function.} Because of (i),
(ii) and (iii), $\sigma(t)$ is well defined for each time~$t\in \R_0^+$. Note
that $\sigma$'s state function in fact depends on $[\sigma]$ only, i.e., we may
add or remove idempotent events at will without changing the state function.

\paragraph{Distributed System}
On the topmost level of abstraction, we see the system as a set of
$V=\{1,\ldots,n\}$ physically remote \emph{nodes} that communicate by means of
\emph{channels}. In the context of a VLSI circuit, ``physically remote''
actually refers to quite small distances (centimeters or even less). However, at
gigahertz frequencies, a local state transition will not be observed remotely
within a time that is negligible compared to clock speeds. We stress this point,
since it is crucial that different clocks (and their attached logic) are not too
close to each other, as otherwise they might fail due to the same event such as
a particle hit. This would render it pointless to devise a system that is
resilient to a certain fraction of the nodes failing.

Each node~$i$ comprises a number of \emn{input ports}, namely $S_{i,j}$ for each
node $j$, an \emn{output port} $S_i$, and a set of \emn{local ports}, introduced
later on. An \emn{execution} of the distributed system assigns to each port of
each node a signal. For convenience of notation, for any port $p$, we refer to
the signal assigned to port $p$ simply by signal~$p$. We say that \emn{node $i$
is in state $s$} at time $t$ iff $S_i(t)=s$. We further say that \emn{node $i$
switches to state $s$} at time $t$ iff signal $S_i$ switches to $s$ at time $t$.

Nodes exchange their states via the channels between them: for each pair of
nodes $i,j$, output port $S_i$ is connected to input port $S_{j,i}$ by a FIFO
channel from $i$ to $j$. Note that this includes a channel from $i$ to $i$
itself. Intuitively, $S_i$ being connected to $S_{j,i}$ by a (non-faulty)
channel means that $S_{j,i}(\cdot)$ should mimic $S_i(\cdot)$, however, with a
slight delay accounting for the time it takes the signal to propagate. In
contrast to an asynchronous system, this delay is bounded by the \emn{maximum
delay} $d > 0$.\footnote{With respect to $\BO$-notation, we normalize $d\in
\BO(1)$, as all time bounds simply depend linearly on $d$.}

Formally we define: The \emn{channel} from node $i$ to $j$ is said
to be \emn{correct} during $[t^-,t^+]$ iff there exists a function $\tau_{i,j}:
\R_0^+ \to \R_0^+$, called the channel's \emn{delay function}, such that: (i)
$\tau_{i,j}$ is continuous and strictly increasing, (ii) $\forall t\in
[t^-,t^+]: 0 \leq \tau_{i,j}(t)-t < d$, and (iii) for each $t \in [t^-,t^+]$,
$(s,\tau_{i,j}(t)) \in S_{j,i} \LRa (s,t) \in S_{i}$. We
say that node $i$ \emph{observes node $j$ in state $s$} at time $t$ if
$S_{i,j}(t)=s$.

\paragraph{Clocks and Timeouts}

Nodes are never aware of the current reference time and we also do not require
the reference time to resemble Newtonian ``real'' time. Rather we allow for
physical clocks that run arbitrarily fast or slow, as long as their speeds are
close to each other in comparison. One may hence think of the reference time as
progressing at the speed of the currently slowest correct clock. In this
framework, nodes essentially make use of bounded clocks with bounded drift.

Formally, clock rates are within $[1,\vartheta]$ (with respect to reference
time), where $\vartheta>1$ is constant and $\vartheta-1$ is the \emph{(maximum)
clock drift}. A \emph{clock} $C$ is a continuous, strictly increasing function
$C:\R^+_0\to \R^+_0$ mapping reference time to some local time. Clock $C$ is
said to be \emn{correct} during $[t^-,t^+]\subseteq \R^+_0$ iff we have for any
$t,t'\in [t^-,t^+]$, $t<t'$, that $t'-t\leq C(t')-C(t)\leq \vartheta (t'-t)$.
Each node comprises a set of clocks assigned to it, which allow the node to
estimate the progress of reference time.

Instead of directly accessing the value of their clocks, nodes have
access to so-called \emn{timeout ports} of watchdog timers. A \emph{timeout} is a triple
$(T,s,C)$, where $T\in \R^+$ is a duration, $s\in \Ss$ is a state, and $C$ is a
clock, say of node~$i$. Each timeout $(T,s,C)$ has a corresponding timeout port
$\Time_{T,s,C}$, being part of node $i$'s local ports. Signal $\Time_{T,s,C}$ is
Boolean, that is, its possible states are from the set $\{0,1\}$. We say that
timeout $(T,s,C)$ is \emn{correct} during $[t^-,t^+]\subseteq \R^+_0$ iff
clock~$C$ is correct during $[t^-,t^+]$ and the following holds:

\begin{enumerate}
  \item For each time $t_s\in [t^-,t^+]$ when node~$i$ switches to state~$s$,
  there is a time $t\in[t_s,\tau_{i,i}(t_s)]$ such that $(T,s,C)$ is \emph{reset},
  i.e., $(0,t)\in \Time_{T,s,C}$. This is a one-to-one correspondence, i.e.,
  $(T,s,C)$ is not reset at any other times.
  
  \item For a time $t\in [t^-,t^+]$, denote by $t_0$ the supremum of all times
  from $[t^-,t]$ when $(T,s,C)$ is reset. Then it holds that $(1,t)\in
  \Time_{T,s,C}$ iff $C(t)-C(t_0) = T$. Again, this is a one-to-one correspondence.
\end{enumerate}

We say that timeout $(T,s,C)$ \emph{expires} at time~$t$ iff
$\Time_{T,s,C}$ switches to~$1$ at time~$t$, and it \emph{is
expired} at time $t$ iff $\Time_{T,s,C}(t)=1$.
For notational convenience, we will omit the clock $C$ and simply write $(T,s)$
for both the timeout and its signal.

A \emn{randomized timeout} is a triple $({\cal D},s,C)$, where ${\cal D}$ is a
bounded random distribution on $\R^+_0$, $s\in \Ss$ is a state, and $C$ is a
clock. Its corresponding timeout port $\Time_{{\cal D},s,C}$ behaves very similar
to the one of an ordinary timeout, except that whenever it is reset, the local
time that passes until it expires next---provided that it is not reset again
before that happens---follows the distribution $\cal D$. Formally, $({\cal
D},s,C)$ is correct during $[t^-,t^+]\subseteq \R^+_0$, if $C$ is correct during
$[t^-,t^+]$ and the following holds:

\begin{enumerate}
  \item For each time $t_s\in [t^-,t^+]$ when node~$i$ switches to state~$s$,
  there is a time $t\in[t_s,\tau_{i,i}(t_s)]$ such that $({\cal D},s,C)$ is
  \emph{reset}, i.e., $(0,t)\in \Time_{{\cal D},s,C}$. This is a one-to-one
  correspondence, i.e., $({\cal D},s,C)$ is not reset at any other times.
  \item For a time $t\in [t^-,t^+]$, denote by $t_0$ the supremum of all times
  from $[t^-,t]$ when $({\cal D},s,C)$ is reset. Let $\mu:\R_0^+\to \R_0^+$
  denote the density of $\cal D$. Then $(1,t)\in \Time_{{\cal D},s,C}$ ``with
  probability $\mu(C(t)-C(t_0))$'' and we require that the probability of $(1,t)\in
  \Time_{{\cal D},s,C}$---conditional to $t_0$ and $C$ on $[t_0,t]$ being
  given---is independent of the system's state at times smaller than $t$. More
  precisely, if superscript $\cal E$ identifies variables in execution $\cal E$
  and $t_0'$ is the infimum of all times from $(t_0,t^+]$ when node $i$ switches
  to state $s$, then we demand for any $[\tau^-,\tau^+]\subseteq [t_0,t_0']$ that
  \begin{equation*}
P\left[\exists t'\in [\tau^-,\tau^+]:(1,t')\in \Time_{{\cal D},s,C}\,\Big|\,
t_0^{\cal E}=t_0 \wedge C\big|_{[t_0,t']}^{\cal
E}=C\big|_{[t_0,t']}\right]=\int_{\tau^-}^{\tau^+}\mu(C(\tau)-C(t_0))~d\tau,
\end{equation*}
independently of ${\cal E}\big|_{[0,\tau^-)}$.
\end{enumerate}

We will apply the same notational conventions to randomized timeouts as we do
for regular timeouts.

Note that, strictly speaking, this definition does not induce a random variable
describing the time $t'\in [t_0,t_0')$ satisfying that $(1,t')\in \Time_{{\cal
D},s,C}$. However, for the state of the timeout port, we get the
meaningful statement that for any $t'\in [t_0,t_0')$,
\begin{equation*}
P[\Time_{{\cal D},s,C}\mbox{ switches to
$1$ during }[t_0,t']]=\int_{t_0}^{t'} \mu(C(t')-C(t_0))~d\tau.
\end{equation*}
The reason for phrasing the definition in the above more cumbersome way is that
we want to guarantee that an adversary knowing the full present state of the
system and memorizing its whole history cannot reliably predict when the timeout
will expire.\footnote{This is a non-trivial property. For instance nodes could
just determine, by drawing from the desired random distribution at time $t_0$,
at which local clock value the timeout shall expire next. This would, however,
essentially give away early when the timeout will expire, greatly reducing the
power of randomization!}

We remark that these definitions allow for different timeouts to be driven by
the same clock, implying that an adversary may derive some information on the
state of a randomized timeout before it expires from the node's behaviour, even
if it cannot directly access the values of the clock driving the timeout. This
is crucial for implementability, as it might be very difficult to guarantee
that the behaviour of a dedicated clock that drives a randomized timeout
is indeed independent of the execution of the algorithm.

\paragraph{Memory Flags}

Besides timeout and randomized timeout ports, another kind of node~$i$'s local
ports are \emn{memory flags}. For each state $s\in \Ss$ and each node $j \in V$,
$\Mem_{i,j,s}$ is a local port of node~$i$. It is used to memorize whether
node~$i$ has observed node~$j$ in state $s$ since the last reset of the flag. We
say that node $i$ \emph{memorizes node $j$ in state $s$} at time $t$ if
$\Mem_{i,j,s}(t)=1$. Formally, we require that signal $\Mem_{i,j,s}$ switches
to~$1$ at time~$t$ iff node~$i$ observes node~$j$ in state~$s$ at time~$t$ and
$\Mem_{i,j,s}$ is not already in state~$1$. The times $t$ when $\Mem_{i,j,s}$ is
\emn{reset}, i.e., $(0,t)\in\Mem_{i,j,s}$, are specified by node~$i$'s state
machine, which is introduced next.

\paragraph{State Machine}

It remains to specify how nodes switch states and when they reset memory flags.
We do this by means of state machines that may attain states from the finite
alphabet $\Ss$. A node's state machine is specified by (i) the set $\Ss$, (ii) a
function $tr$, called the \emn{transition function}, from ${\cal T}\subseteq
\Ss^2$ to the set of Boolean predicates on the alphabet consisting of
expressions ``$p = s$'' (used for expressing guards), where $p$ is from the
node's input and local ports and $s$ is from the set of possible states of
signal~$p$, and (iii) a function $re$, called the \emn{reset function}, from
$\cal T$ to the power set of the node's memory flags.

Intuitively, the transition function specifies the conditions (guards) under
which a node switches states, and the reset function determines
which memory flags to reset upon the state change.
Formally, let $P$ be a predicate on node~$i$'s input and local ports.
We define $P$ \emn{holds at time $t$} by structural induction: If $P$
is equal to $p = s$, where $p$ is one of node $i$'s input and
local ports and $s$ is one of the states signal $p$ can obtain,
then $P$ \emn{holds at time $t$} iff $p(t)=s$.
Otherwise, if $P$ is of the form $\neg P_1$, $P_1 \wedge P_2$, or $P_1
\vee P_2$, we define $P$ \emn{holds at time $t$} in the
straightforward manner.

We say node~$i$ \emn{follows its state machine during} $[t^-,t^+]$ iff
the following holds: Assume node~$i$ observes itself in
state~$s\in \Ss$ at time~$t\in [t^-,t^+]$, i.e., $S_{i,i}(t)=s$.
Then, for each $(s,s')\in \cal T$, both:
\begin{enumerate}
  \item Node~$i$ switches to state~$s'$ at time~$t$ iff $tr(s,s')$ holds
  at time~$t$ and $i$ is not already in state~$s'$.\footnote{In
  case more than one guard $tr(s,s')$ can be true at the same time,
  we assume that an arbitrary tie-breaking ordering exists among
  the transition guards that specifies to which state to switch.}
  
  \item Node $i$ resets memory flag $m$ at some time in the interval
  $[t,\tau_{i,i}(t)]$ iff $m\in re(s,s')$ and $i$ switches from state $s$ to
  state~$s'$ at time~$t$. This correspondence is one-to-one.
\end{enumerate}

A node is defined to be \emn{non-faulty} during $[t^-,t^+]$ iff during
$[t^-,t^+]$ all its timeouts and randomized timeouts are
correct and it follows its state machine. If it employs multiple state machines
(see below), it needs to follow all of them.

In contrast, a faulty node may change states arbitrarily.
Note that while a faulty node may be forced to send consistent output
state signals to all other nodes if its channels remain correct,
there is no way to guarantee that this still holds true if
channels are faulty.\footnote{A single physical fault may cause
this behaviour, as at some point a node's output port must be
connected to remote nodes' input ports.
Even if one places bifurcations at different physical locations
striving to mitigate this effect, if the voltage at the output
port drops below specifications, the values of corresponding
input channels may deviate in unpredictable ways.}

\paragraph{Metastability}

In our discrete system model, the effect of metastability is captured
by the lacking capability of state machines to instantaneously take on
new states: Node $i$ decides on state transitions based on the delayed
status of port $S_{i,i}$ instead of its ``true'' current state $S_i$.
This non-zero delay from $S_i$ to $S_{i,i}$ bears the potential for
metastability, as a successful state transition can only be guaranteed
if after a transition guard from some state $s$ to some state $s'$
becomes true, all other transition guards from $s$ to $s'' \neq s'$
remain false during this delay at least.

This is exemplified in the following scenario:
Assume node $i$ is in state $s$ at some time $t$.
However, since it switched to $s$ only very recently, it still
observes itself in state $s' \neq s$ at time $t$ via $S_{i,i}$.
Given that there is a transition $(s',s'')$ in $\cal T$, $s''\neq s$,
whose condition is fulfilled at time $t$, it will switch to state
$s''$ at time $t$ (although state $s$ has not even stabilized
yet).
That is, due to the discrepancy between $S_{i,i}$ and $S_i$, node $i$
switches from state $s$ to state $s''$ at time~$t$ even if
$(s,s'')$ is not in ${\cal T}$ at all.\footnote{Note that while
the ``internal'' delay $\tau_{i,i}(t)-t$ can be made quite small,
it cannot be reduced to zero if the model is meant to reflect
physical implementations.} 
In a physical chip design, this premature change of state might even
result in inconsistent operations on the local memory, up to the
point where it cannot be properly described in terms of $\Ss$,
and thus in terms of our discrete model, anymore.
Even worse, the state of $i$ is part of the local memory and the
node's state signal may attain an undefined value that is
propagated to other nodes and their memory.
While avoiding the latter is the task of the input ports of a
non-faulty node, our goal is to prevent this erroneous behaviour
in situations where input ports attain legitimate values only.

Therefore, we define node $i$ to be \emph{metastability-free}, if the
situation described above does not occur.
\begin{definition}[Metastability-Freedom] Node $i\in V$ is called
\emn{metastability-free during $[t^-,t^+]$}, iff for each time $t\in [t^-,t^+]$
when $i$ switches to some state $s\in \Ss$, it holds that $\tau_{i,i}(t) < t'$,
where $t'$ is the infimum of all times in $(t,t^+]$ when $i$ switches to some
state $s' \in \Ss$. \end{definition}

\paragraph{Multiple State Machines}
In some situations the previous definitions are too stringent, as there might be
different ``components'' of a node's state machine that act concurrently and
independently, mostly relying on signals from disjoint input ports or orthogonal
components of a signal. We model this by permitting that nodes run several state
machines in parallel. All these state machines share the input and local ports
of the respective node and are required to have disjoint state spaces. If node
$i$ runs state machines $M_1,\ldots,M_k$, node $i$'s output signal is the
product of the output signals of the individual machines. Formally we define:
Each of the state machines $M_j$, $1 \leq j \leq k$, has an additional own
output port $s_j$. The state of node $i$'s output port $S_i$ at any time $t$ is
given by $S_i(t):=(s_1(t),\ldots,s_k(t))$, where the signals of ports
$s_1,\ldots,s_k$ are definied analogously to the signals of the output ports of
state machines in the single state machine case, each. Note that by this
definition, the only (local) means for node $i$'s state machines to interact
with each other is by reading the delayed state signal $S_{i,i}$.

We say that \emn{node $i$'s state machine $M_j$ is in state $s$ at time $t$} iff
$s_j(t) = s$, where $S_i(t) = (s_1(t),\ldots,s_k(t))$, and that \emn{node $i$'s
state machine $M_j$ switches to state $s$ at time $t$} iff signal $s_j$ switches
to $s$ at time $t$. Since the state spaces of the machines $M_j$ are disjoint,
we will omit the phrase ``state machine $M_j$'' from the notation, i.e., we
write ``node $i$ is in state $s$'' or ``node $i$ switched to state $s$'',
respectively.

Recall that the various state machines of node $i$ are as loosely coupled as
remote nodes, namely via the delayed status signal on channel $S_{i,i}$ only.
Therefore, it makes sense to consider them independently also when it comes to
metastability.

\begin{definition}[Metastability-Freedom (Multiple State Machines)]
State machine $M$ of node $i\in V$ is called \emn{metastability-free
during $[t^-,t^+]$}, iff for each time $t\in [t^-,t^+]$ when $M$ switches to
some state $s\in \Ss$, it holds that $\tau_{i,i}(t) < t'$, where $t'$ is the
infimum of all times in $(t,t^+]$ when $M$ switches to some state
$s' \in \Ss$.
\end{definition}

Note that by this definition the different state machines may switch
states concurrently without suffering from metastability.\footnote{However,
care has to be taken when implementing the inter-node communication of the
state components in a metastability-free manner,
cf.~\sectionref{sec:implementation}.} It is even possible that some state
machine suffers metastability, while another is not affected by this at
all.\footnote{This is crucial for the algorithm we are going to present.
For stabilization purposes, nodes comprise a state machine that is prone to
metastability. However, the state machine generating pulses (i.e., having
the state \acc, cf.~\defref{def:pulse}) does not take its output signal
into account once stabilization is achieved. Thus, the algorithm is
metastability-free after stabilization in the sense that we guarantee a
metastability-free signal indicating when pulses occur.}

\paragraph{Problem Statement}

The purpose of the pulse synchronization protocol is that nodes generate
synchronized, well-separated pulses by switching to a distinguished state \acc.
Self-stabilization requires that they start to do so within bounded time, for any
possible initial state. However, as our protocol makes use of randomization,
there are executions where this does not happen at all; instead, we will show
that the protocol stabilizes with probability one in finite time. To give a
precise meaning to this statement, we need to define appropriate probability
spaces.

\begin{definition}[Adversarial Spaces] Denote for $i\in V$ by ${\cal
C}_i=\{C_{i,k}\,|\,k\in \{1,\ldots,c_i\}\}$ the set of clocks of node $i$. An
\emph{adversarial space} is a probabilistic space that is defined by subsets of
nodes and channels $W\subseteq V$ and $E\subseteq V\times V$, a time interval
$[t^-,t^+]$, a protocol ${\cal P}$ (nodes' ports, state machines, etc.) as
previously defined, sets of clock and delay functions ${\cal C}=\bigcup_{i\in
V}{\cal C}_i$ and $\Theta=\{\tau_{i,j}:\R^+_0\to \R^+_0\,|\,(i,j)\in V^2\}$, an
initial state ${\cal E}_0$ of all ports, and an \emph{adversarial function
${\cal A}$}. Here ${\cal A}$ is a function that maps a partial execution ${\cal
E}|_{[0,t]}$ until time $t$ (i.e., all ports' values until time $t$), $W$, $E$,
$[t^-,t^+]$, ${\cal P}$, ${\cal C}$, and $\Theta$ to the states of all faulty
ports during the time interval $(t,t']$, where $t'$ is the infimum of all times
greater than $t$ when a non-faulty node or channel switches states.

The adversarial space ${\cal AS}(W,E,[t^-,t^+],{\cal P},{\cal C},\Theta,{\cal
E}_0,{\cal A})$ is now defined on the set of all executions ${\cal E}$
satisfying that $(i)$ the initial state of all ports is given by ${\cal
E}|_{[0,0]}={\cal E}_0$, $(ii)$ for all $i\in V$ and $k\in \{1,\ldots,c_i\}:$
$C_{i,k}^{\cal E}=C_{i,k}$, $(iii)$ for all $(i,j)\in V^2$, $\tau_{i,j}^{\cal
E}=\tau_{i,j}$, $(iv)$ nodes in $W$ are non-faulty during
$[t^-,t^+]$ with respect to the protocol ${\cal P}$, $(v)$ all channels in $E$
are correct during $[t^-,t^+]$, and $(vi)$ given ${\cal E}|_{[0,t]}$ for any
time $t$, ${\cal E}|_{(t,t']}$ is given by ${\cal A}$, where $t'$ is the infimum
of times greater than $t$ when a non-faulty node switches states. Thus, except
for when randomized timeouts expire, ${\cal E}$ is fully predetermined by the
parameters of ${\cal AS}$.\footnote{This follows by induction starting from the
initial configuration ${\cal E}_0$. Using ${\cal A}$, we can always extend
${\cal E}$ to the next time when a correct node switches states, and when
correct nodes switch states is fully determined by the parameters of ${\cal AS}$
except for when randomized timeouts expire. Note that the induction reaches any
finite time within a finite number of steps, as signals switch states finitely
often in finite time.} The probability measure on ${\cal AS}$ is induced by the
random distributions of the randomized timeouts specified by ${\cal P}$.
\end{definition}

To avoid confusion, observe that if the clock functions and delays do not follow
the model constraints during $[t^-,t^+]$, the respective adversarial space is
empty and thus of no concern. This cumbersome definition provides the means to
formalize a notion of stabilization that accounts for worst-case drifts and
delays and an adversary that knows the full state of the system up to the
current time.

We are now in the position to formally state the pulse synchronization problem
in our framework. Intuitively, the goal is that after transient faults cease,
nodes should with probability one eventually start to issue well-separated,
synchronized pulses by switching to a dedicated state \acc. Thus, as the initial
state of the system is arbitrary, specifying an algorithm\footnote{We use the
terms ``algorithm'' and ``protocol'' interchangably throughout this work.} is
equivalent to defining the state machines that run at each node, one of which
has a state \acc.

\begin{definition}[Self-Stabilizing Pulse Synchronization]\label{def:pulse}
Given a set of nodes $W \subseteq V$ and a set $E\subseteq V\times V$ of
channels, we say that protocol ${\cal P}$ is a \emn{$(W,E)$-stabilizing pulse
synchronization protocol with skew $\Sigma$ and accuracy bounds $T^-,T^+$ that
stabilizes within time $T$ with probability $p$} iff the following holds. Choose
any time interval $[t^-,t^+]\supseteq [t^-,t^- +T+\Sigma]$ and any adversarial
space ${\cal AS}(W,E,[t^-,t^+],{\cal P},\cdot,\cdot,\cdot,\cdot)$ (i.e., ${\cal
C}$, $\Theta$, ${\cal E}_0$, and ${\cal A}$ are arbitrary). Then executions from
${\cal AS}$ satisfy with probability at least $p$ that there exists a time $t_s
\in [t^-,t^- +T]$ so that, denoting by $t_i(k)$ the time when node~$i$ switches
to a distinguished state \acc\ for the $k^\text{th}$ time after $t_s$
($t_i(k)=\infty$ if no such time exists), $(i)$ $t_i(1)\in (t_s,t_s+\Sigma)$,
$(ii)$ $|t_i(k)-t_j(k)| \le \Sigma$ if $\max\{t_i(k),t_j(k)\} \le t^+$, and
$(iii)$ $T^-\le |t_i(k+1)-t_i(k)| \le T^+$ if $t_i(k)+T^+\le t^+$.
\end{definition}

Note that the fact that ${\cal A}$ is a deterministic function and, more
generally, that we consider each space ${\cal AS}$ individually, is no
restriction: As ${\cal P}$ succeeds for any adversarial space with probability
at least $p$ in achieving stabilization, the same holds true for randomized
adversarial strategies ${\cal A}$ and worst-case drifts and delays.

\section{The FATAL Pulse Synchronization Protocol}

In this section, we present our self-stabilizing pulse generation algorithm.
In order to be suitable for implementation in hardware, it needs to utilize
very simple rules only. It is stated in terms of a state machine as introduced
in the previous section.

Since the ultimate goal of the pulse generation algorithm is to stabilize a
system of \darts clocks, we introduce an additional port $\text{\darts}_i$, for
each node $i$, which is driven by node $i$'s \darts instance. As for other state
signals, its output raises flag $\Mem_{i,\darts}$, to which for simplicity we
refer to as $\darts_i$ as well. Note that the \darts signals are of no concern
to the liveliness or stabilization of the pulse algorithm itself; rather, it is
a control signal from the \darts component that helps in adjusting the frequency
of pulses to the speed of the \darts clocks once the system as a whole
(including the \darts component) is stable. The pulse algorithm will stabilize
independently of the \darts signal, and the \darts component will stabilize once
the pulse component did so. Therefore we can partition the algorithm's analysis
into two parts. When proving the correctness of the algorithm in
\sectionref{sec:analysis}, we assume that for each node~$i$, $\text{\darts}_i$
is arbitrary. In \sectionref{sec:coupling}, we will outline how the pulse
algorithm and \darts interact.

\subsection{Basic Cycle}
\label{sec:basic_cycle}

\begin{figure}[t!]
\centering
\includegraphics[width=.5\textwidth]{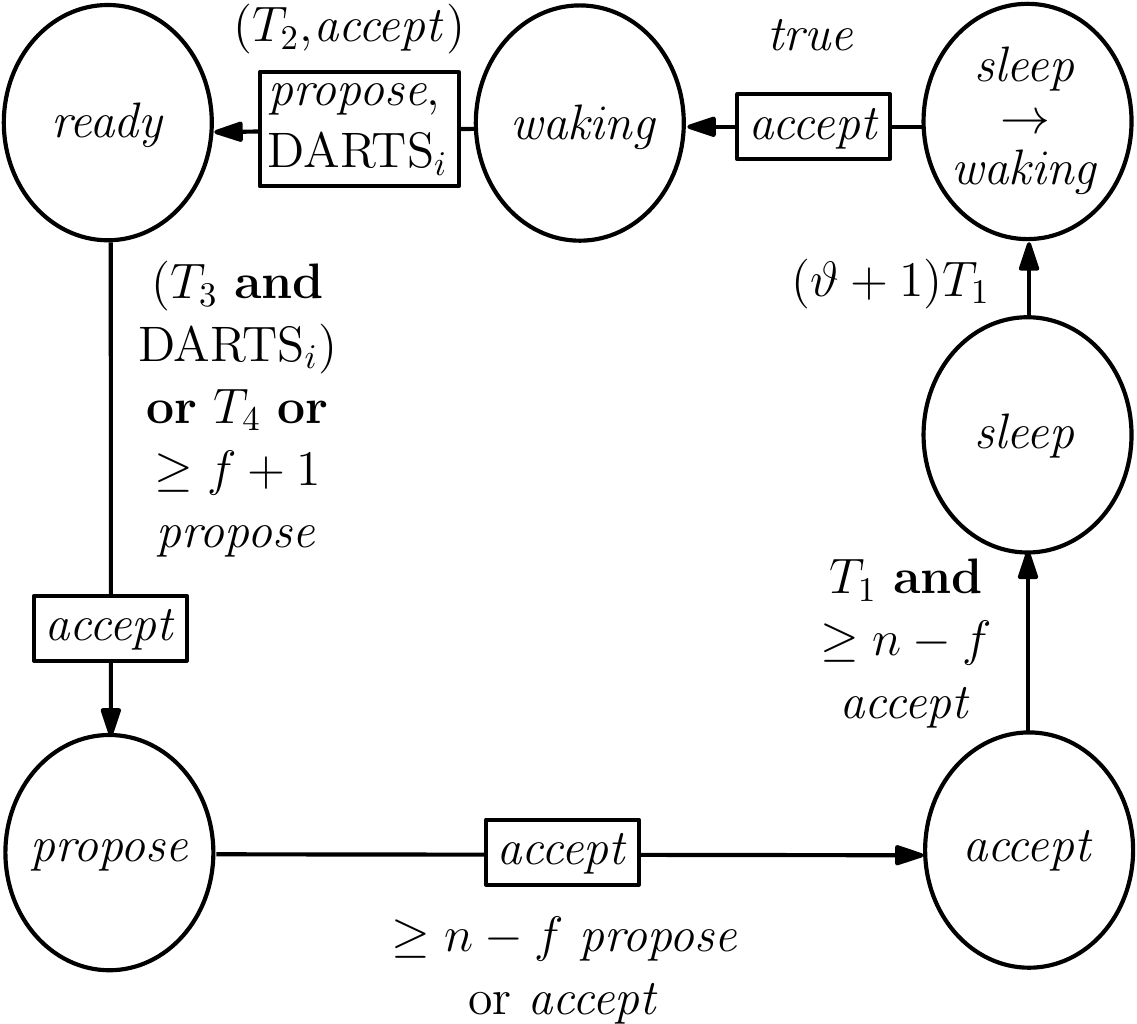}
\caption{Basic cycle of node~$i$ once the algorithm has stabilized.}
\label{fig:main_simple}
\end{figure}

The full algorithm makes use of a rather involved interplay between conditions
on timeouts, states, and thresholds to converge to a safe state despite a
limited number of faulty components. As our approach is thus difficult to
present in bulk, we break it down into pieces. Moreover, to facilitate giving
intuition about the key ideas of the algorithm, in this section we assume that
there are $f<n/3$ faulty nodes, and the remaining $n-f$ nodes are non-faulty
within $[0,\infty)$ (where of course the time $0$ is unknown to the nodes). We
further assume that channels between non-faulty nodes (including loopback
channels) are correct within $[0,\infty)$. We start by presenting the basic
cycle that is repeated every pulse once a safe configuration is reached (see
\figureref{fig:main_simple}).

We employ graphical representations of the state machine of each node~$i \in V$.
States are represented by circles containing their names, while transition
$(s,s') \in \cal T$ is depicted as an arrow from $s$ to $s'$. The guard
$tr(s,s')$ is written as a label next to the arrow, and the reset function's
value $re(s,s')$ is depicted in a rectangular box on the arrow. To keep labels
more simple we make use of some abbreviations. We write $T$ instead of $(T,s)$
if $s$ is the state which node~$i$ leaves if the condition involving $(T,s)$ is
satisfied. Threshold conditions like ``\,$\ge f+1$ $s$\,'', where $s \in \Ss$,
abbreviate Boolean predicates that reach over all of node~$i$'s memory flags
$\Mem_{i,j,s}$, where $j \in V$, and are defined in a straightforward manner. If
in such an expression we connect two states by ``or'', e.g., ``\,$\ge n-f$ $s$
or $s'$\,'' for $s,s'\in \Ss$, the summation considers flags of both types $s$
and $s'$. Thus, such an expression is equivalent to $\sum_{j\in
V}\max\{\Mem_{i,j,s},\Mem_{i,j,s'}\}\geq f+1$. For any state $s \in \Ss$, the
condition $S_{i,j} = s$, (respectively, $\neg(S_{i,j} = s)$) is written in short
as ``$j$ in~$s$'' (respectively, ``$j$ not in~$s$''). If $j=i$, we simply write
``(not) in~$s$''. We write ``true'' instead of a condition that is always true
(like e.g.\ ``(in $s$) \textbf{or} (not in $s$)'' for an arbitrary state $s \in
\Ss$). Finally, $re(\cdot,\cdot)$ always requires to reset all memory flags of
certain types, hence we write e.g.\ \prop\ if all flags $\Mem_{i,j,\prop}$ are to
be reset.

We now briefly introduce the basic flow of the algorithm once it stabilizes,
i.e., once all $n-f$ non-faulty nodes are well-synchronized. Recall that the
remaining up to $f<n/3$ faulty nodes may produce arbitrary signals on their
outgoing channels. A pulse is locally triggered by switching to state \acc.
Thus, assume that at some time all non-faulty nodes switch to state \acc\ within
a time window of $2d$, i.e., a valid pulse is generated. Supposing that $T_1\geq
3\vartheta d$, these nodes will observe, and thus memorize, each other and
themselves in state \acc\ before $T_1$ expires. This makes timeout $T_1$ the
critical condition for switching to state \slp. From state \slp, they will
switch to states \srw, \wake, and finally \rdy, where the timeout $(T_2,\acc)$
is determining the time this takes, as it is considerably larger than
$\vartheta(\vartheta+2)T_1$. The intermediate states serve the purpose of
achieving stabilization, hence we leave them out for the moment. Note that upon
switching to state \rdy, nodes reset their \prop\ flags and $\darts_i$. Thus,
they essentially ignore these signals between the most recent time they switched
to \prop\ before switching to \acc\ and the subsequent time when they switch to
\rdy. This ensures that nodes do not take into account outdated information for
the decision when to switch to state \prop. Hence, it is guaranteed that the
first node switching from state \rdy\ to state \prop\ again does so because
$T_4$ expired or because $T_3$ expired and its \darts memory flag is true. Due
to the constraint $\min\{T_3,T_4\}\geq \vartheta (T_2+4d)$, we are sure that all
non-faulty nodes observe themselves in state \rdy\ before the first one switches
to \prop. Hence, no node deletes information about nodes that switch to \prop\
again after the previous pulse. The first non-faulty node that switches to state
\acc\ again cannot do so before it memorizes at least $n-f$ nodes in state
\prop, as the \acc\ flags are reset upon switching to state \prop. Therefore, at
this time at least $n-2f\geq f+1$ non-faulty nodes are in state \prop. Hence,
the rule that nodes switch to \prop\ if they memorize $f+1$ nodes in states
\prop\ will take effect, i.e., the remaining non-faulty nodes in state \rdy\
switch to \prop\ after less than $d$ time. Another $d$ time later all non-faulty
nodes in state \prop\ will have become aware of this and switch to state \acc\
as well, as the threshold of $n-f$ nodes in states \prop\ or \acc\ is reached.
Thus the cycle is complete and the reasoning can be repeated inductively.

Clearly, for this line of argumentation to be valid, the algorithm could be
simpler than stated in \figureref{fig:main_simple}. We already mentioned that
the motivation of having three intermediate states between \acc\ and \rdy\ is to
facilitate stabilization. Similarly, there is no need to make use of the \acc\
flags in the basic cycle at all; in fact, it adversely affects the constraints
the timeouts need to satisfy for the above reasoning to be valid. However, the
\acc\ flags are much better suited for diagnostic purposes than the \prop\
flags, since nodes are expected to switch to \acc\ in a small time window and
remain in state \acc\ for a small period of time only (for all our results, it
is sufficient if $T_1=4\vartheta d$). Moreover, two different timeout conditions
for switching from \rdy\ to \prop\ are unnecessary for correct operation of the
pulse synchronization routine. As discussed before, they are introduced in order
to allow for a seamless coupling to the \darts system. We elaborate on this in
\sectionref{sec:coupling}.

\subsection{Main Algorithm}

\begin{figure}[t!]
\centering
\includegraphics[width=.7\textwidth]{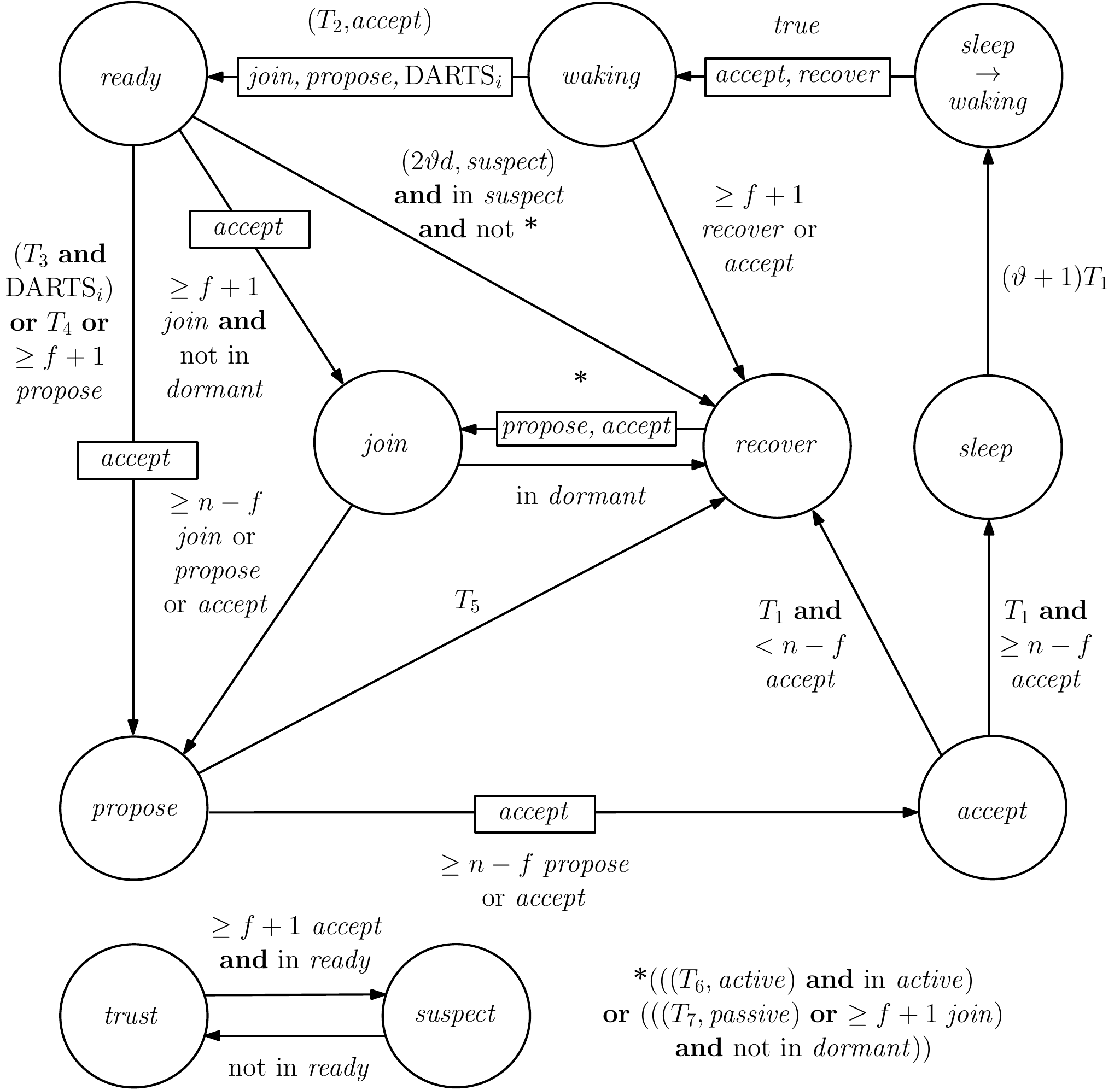}
\caption{Overview of the core routine of node $i$'s self-stabilizing pulse
algorithm.}
\label{fig:main}
\end{figure}

We proceed by describing the main routine of the pulse algorithm in
full. Alongside the main routine, several other state machines run
concurrently and provide additional information to be used during recovery.

The main routine is graphically presented in \figureref{fig:main}, together with
a very simple second component whose sole purpose is to simplify the otherwise
overloaded description of the main routine. Except for the states \rec\ and
\join\ and additional resets of memory flags, the main routine is identical to
the basic cycle. The purpose of the two additional states is the following:
Nodes switch to state \rec\ once they detect that something is wrong, that is,
non-faulty nodes do not execute the basic cycle as outlined in
\sectionref{sec:basic_cycle}. This way, non-faulty nodes will not continue to
confuse others by sending for example state signals \prop\ or \acc\ despite
clearly being out-of-sync. There are various consistency checks that nodes
perform during each execution of the basic cycle. The first one is that in order
to switch from state \acc\ to state \slp, non-faulty nodes need to memorize at
least $n-f$ nodes in state \acc. If this does not happen within $T_1$ time after
switching to state \acc, by the arguments given in \sectionref{sec:basic_cycle},
they could not have entered state \acc\ within $2d$ of each other. Therefore,
something must be wrong and it is feasible to switch to state \rec. Next,
whenever a non-faulty node is in state \wake, there should be no non-faulty
nodes in states \acc\ or \rec. Considering that the node resets its \acc\ and
\rec\ flags upon switching to \wake, it should not memorize $f+1$ or more nodes
in states \acc\ or \rec\ at a time when it observes itself in state \wake. If it
does, however, it again switches to state \rec. Similarly, when in state \rdy,
nodes expect others not to be in state \acc\ for more than a short period of
time, as a non-faulty node switching to \acc\ should imply that every non-faulty
node switches to \prop\ and then to \acc\ shortly thereafter. This is expressed
by the second state machine comprising two states only. If a node is in state
\rdy\ and memorizes $f+1$ nodes in state \acc, it switches to \sus.
Subsequently, if it remains in state \rdy\ until a timeout of $2\vartheta d$
expires, it will switch to state \rec. Last but not least, during a synchronized
execution of the basic cycle, no non-faulty node may be in state \prop\ for more
than a certain amount of time before switching to state \acc. Therefore, nodes
will switch from \prop\ to \rec\ when timeout $T_5$ expires.

Nodes can join the basic cycle again via the second new state, called
\join. Since the Byzantine nodes may ``play nice'' towards $n-2f$ or more
nodes still executing the basic cycle, making them believe that system
operation continues as usual, it must be possible to join the basic cycle
again without having a majority of nodes in state \rec. On the other hand,
it is crucial that this happens in a sufficiently well-synchronized manner,
as otherwise nodes could drop out again because the various checks of
consistency detect an erroneous execution of the basic cycle.

In part, this issue is solved by an additional agreement step. In order to enter
the basic cycle again, nodes need to memorize $n-f$ nodes in states \join\ (the
respective nodes detected an inconsistency), \prop\ (these nodes continued to
execute the basic cycle), or \acc\ (there are executions where nodes reset
their \prop\ flags because of switching to \join\ when other nodes already
switched to \acc). Since there are thresholds of $f+1$ nodes memorized in
state \join\ both for leaving state \rec\ and switching from \rdy\ to \join, all
nodes will follow the first one switching from \join\ to \prop\ quickly, just as
with the switch from \prop\ to \acc\ in an ordinary execution of the basic
cycle. However, it is decisive that all nodes are in states that permit to
participate in this agreement step in order to guarantee success of this
approach.

As a result, still a certain degree of synchronization needs to be established
beforehand, both among nodes that still execute the basic cycle and those that
do not. For instance, if at the point in time when a majority of nodes and
channels become non-faulty, some nodes already memorize nodes in \join\ that
are not, they may switch to state \join\ and subsequently \prop\ prematurely,
causing others to have inconsistent memory flags as well. Again, Byzantine
faults may sustain this amiss configuration of the system indefinitely.

So why did we put so much effort in ``shifting'' the focus to this part of the
algorithm? The key advantage is that nodes outside the basic cycle may take into
account less reliable information for stabilization purposes. They may take the
risk of metastable upsets (as we know it is impossible to avoid these during the
stabilization process, anyway) and make use of randomization. 

In fact, to make the above scheme work, it is sufficient that all
     non-faulty nodes agree on a so called \emph{resynchronization
     point} (formally defined later on), that is, a point in time at
     which nodes reset the memory flags for states \join\ and \srw\ as
     well as certain timeouts, while guaranteeing that no node is in
     these states close to the respective reset times.
Except for state \srw, all of these timeouts, memory flags, etc.\ are
     not part of the basic cycle at all, thus nodes may enforce
     consistent values for them when they agree on such a
     resynchronization point.

Conveniently, the use of randomization also ensures that it is quite
unlikely that nodes are in state \srw\ close to a resynchronization point, as
the consistency check of having to memorize $n-f$ nodes in state \acc\ in order
to switch to state \slp\ guarantees that the time windows during which
non-faulty nodes may switch to \slp\ make up a small fraction of all times only.

Consequently, the remaining components of the algorithm deal with agreeing on
resynchronization points and utilizing this information in an appropriate way to
ensure stabilization of the main routine. We describe this connection to the
main routine first. It is done by another, quite simple state machine, which
runs in parallel alongside the core routine. It is depicted in
\figureref{fig:extended}.

\begin{figure}[t!]
\centering
\includegraphics[width=.4\textwidth]{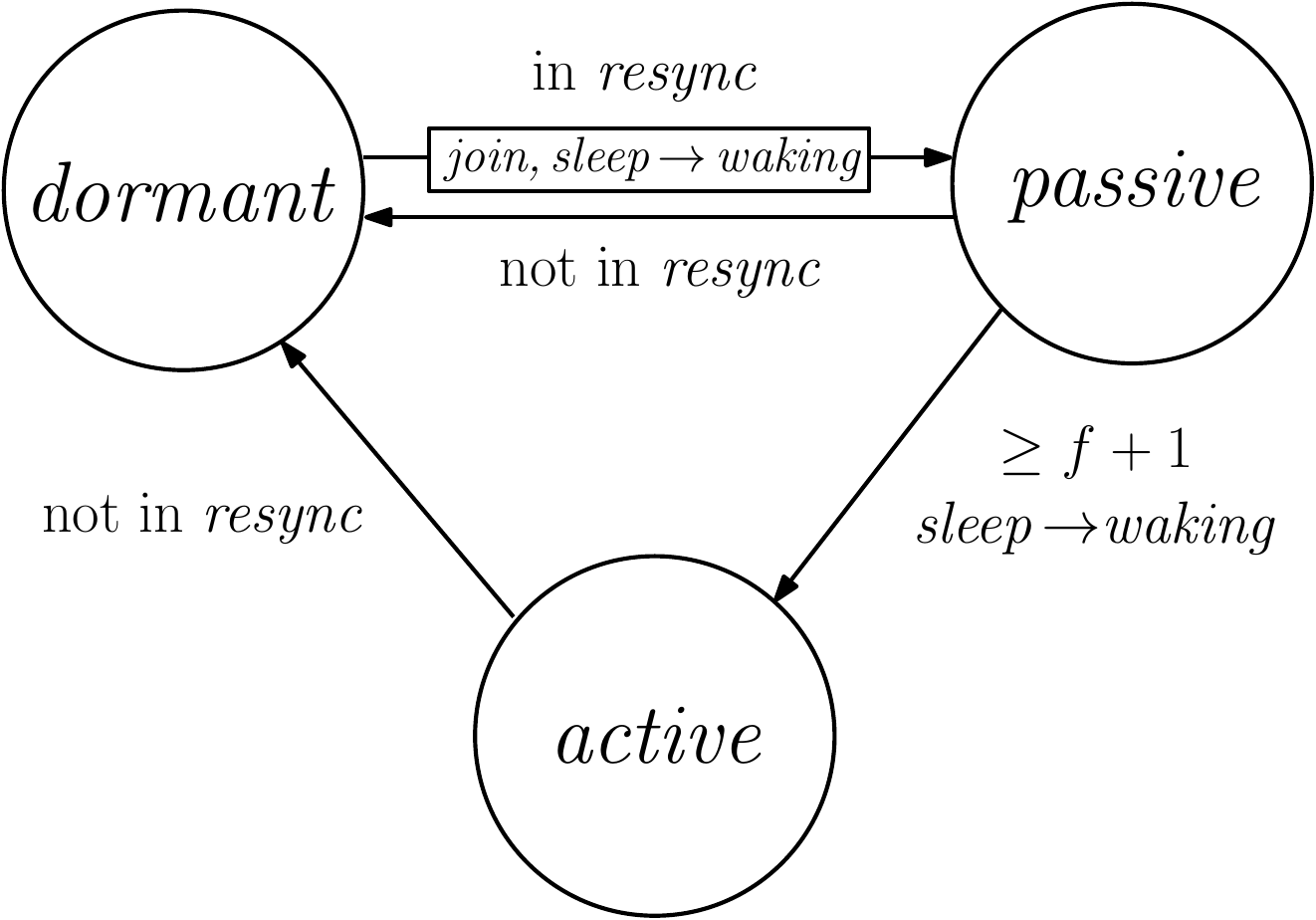}
\caption{Extension of node~$i$'s core routine.}
\label{fig:extended}
\end{figure}

Its purpose is to reset memory flags in a consistent way and to determine when a
node is permitted to switch to \join. In general, a resynchronization point
(locally observed by switching to state \res, which is introduced later)
triggers the reset of the \join\ and \srw\ flags. If there are still nodes
executing the basic cycle, a node may become aware of it by observing $f+1$
nodes in state \srw\ at some time. In this case it switches from the state
\pass, which it entered at the point in time when it locally observed the
resynchronization point, to the state \act, which enables an earlier transition
to state \join. This is expressed by the rather involved transition rule
$tr(\rec,\join)$: $T_6$ is much smaller than $T_7$, but $T_6$ is of no
concern until the node switches to state \act\ and resets $T_6$.\footnote{The
condition ``not in \dorm'' here ensures that the transition is not performed
because the node has been in state \res\ a long time ago, but there was no
recent switching to \res.}

It remains to explain how nodes agree on resynchronization points.

\subsection{Resynchronization Algorithm}

The resynchronization routine is specified in \figureref{fig:resync} as well. It
is a lower layer that the core routine uses for stabilization purposes only. It
provides some synchronization that is very similar to that of a pulse, except
that such ``weak pulses'' occur at random times, and may be generated
inconsistently after the algorithm as a whole has stabilized.
Since the main routine operates independently of the resynchronization routine
once the system has stabilized, we can afford the weaker guarantees of the
routine: If it succeeds in generating a ``good'' resynchronization point merely
once, the main routine will stabilize deterministically.

\begin{definition}[Resynchronization Points]
Given $W\subseteq V$, time $t$ is a \emph{$W$-resynchronization point} iff each
node in $W$ switches to state \srr\ in the time interval $(t,t+2d)$.
\end{definition}

\begin{definition}[Good Resynchronization Points]
A $W$-resynchronization point is called \emph{good} if no node from $W$
switches to state \slp\ during $(t-(\vartheta+3)T_1,t)$ and no node is in
state \join\ during $[t-T_1-d,t+4d)$.
\end{definition}
\begin{figure}[t!]
\centering
\includegraphics[width=.8\textwidth]{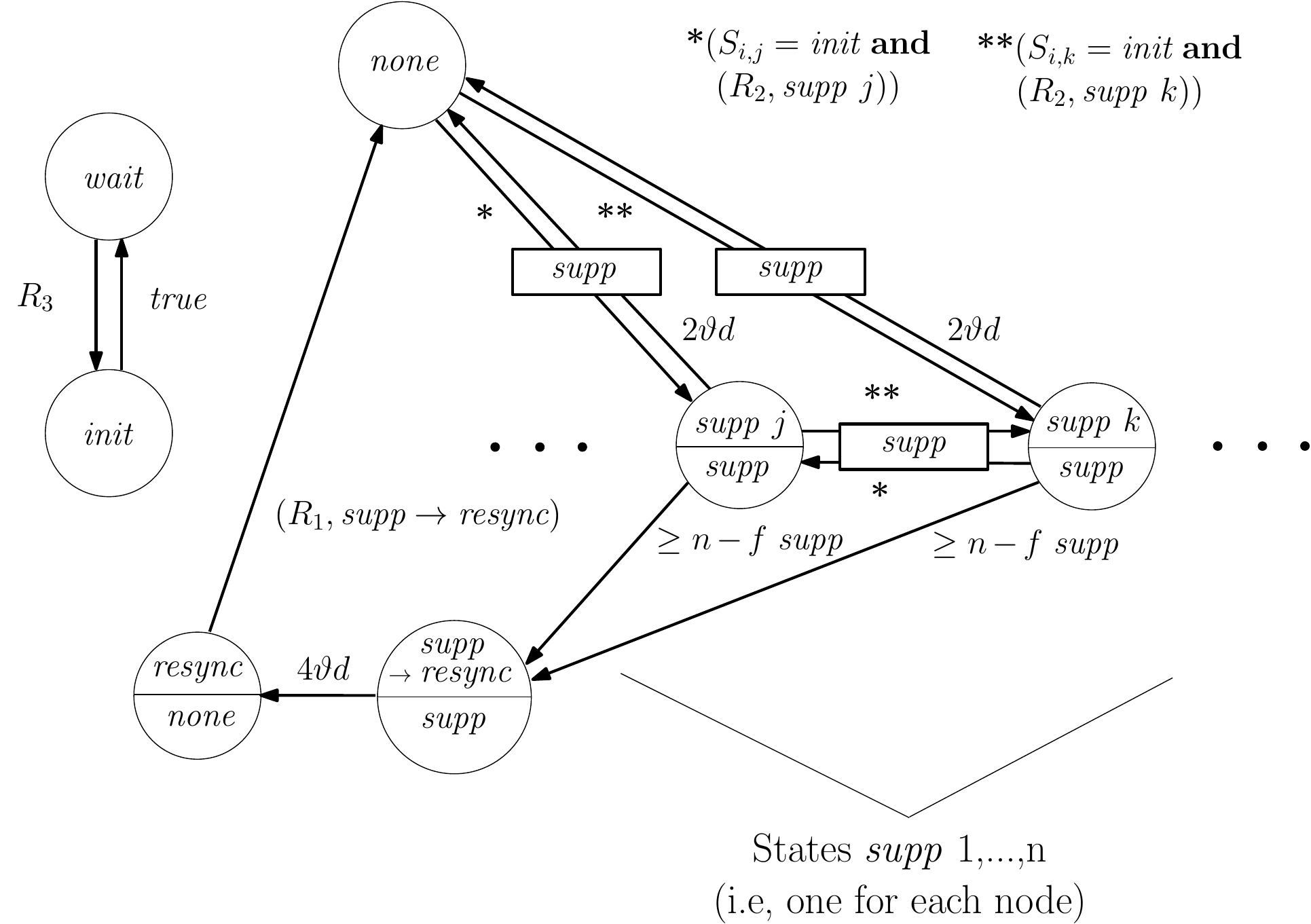}
\caption{Resynchronization algorithm, comprising two state machines
executed in parallel at node~$i$.}\label{fig:resync}
\end{figure}

In order to clarify that despite having a linear number of states
     ($\supp_1,\ldots,\supp_n$), this part of the algorithm can be
     implemented using $2$-bit communication channels between state
     machines only, we generalize our description of state machines as
     follows.
If a state is depicted as a circle separated into an upper and a lower
     part, the upper part denotes the local state, while the lower
     part indicates the signal state to which it is mapped.
A node's memory flags then store the respective signal states only,
     i.e., remote nodes do not distinguish between states that share
     the same signal.
Clearly, such a machine can be simulated by a machine as introduced in
     the model section.
The advantage is that such a mapping can be used to reduce the number
     of transmitted state bits; for the resynchronization routine
     given in \figureref{fig:resync}, we merely need two bits
     (\init/\emph{wait} and \none/\supp) instead of $\lceil \log
     (n+3)\rceil+1$ bits.

The basic idea behind the resynchronization algorithm is the following: Every
now and then, nodes will try to initiate agreement on a resynchronization point.
This is the purpose of the small state machine on the left in
\figureref{fig:resync}. Recalling that the transition condition ``true'' simply
means that the node switches to state \emph{wait} again as soon as it observes
itself in state \init, it is easy to see that it does nothing else than creating
an \init\ signal as soon as $R_3$ expires and resetting $R_3$ again as quickly
as possible. As the time when a node switches to \init\ is determined by the
randomized timeout $R_3$ distributed over a large interval (cf.\ Equality
\eqref{eq:R_3}) only, it is impossible to predict when it will expire, even
with full knowledge of the execution up to the current point in time. Note that
the complete independence of this part of node $i$'s state from the remaining
protocol implies that faulty nodes are not able to influence the respective
times by any means.

Consider now the state machine displayed on the right of \figureref{fig:resync}.
To understand how the routine is intended to work, assume that at the time $t$
when a non-faulty node $i$ switches to state \init, all non-faulty nodes are not
in any of the states \srr, \res, or \supp~$i$, and at all non-faulty nodes the
timeout $(R_2,\supp~i)$ has expired. Then, no matter what the signals from
faulty nodes or on faulty channels are, all non-faulty nodes will be in one of
the states \supp~$j$, $j\in V$, or \srr\ at time $t+d$. Hence, they will observe
each other (and themselves) in one of these states at some time smaller than
$t+2d$. These statements follow from the various timeout conditions of at least
$2\vartheta d$ and the fact that observing node $i$ in state \init\ will make
nodes switch to state \supp~$i$ if in \none\ or \supp~$j$, $j\neq i$. Hence, all
of them will switch to state \srr\ during $(t,t+2d)$, i.e., $t$ is a
resynchronization point. Since $t$ follows a random distribution that is
independent of the remaining algorithm and, as mentioned earlier, most of the
times nodes cannot switch to state \slp\ and it is easy to deal with the
condition on \join\ states, there is a large probability that $t$ is a good
resynchronization point. Note that timeout $R_1$ makes sure that no non-faulty
node will switch to \srr\ again anytime soon, leaving sufficient time for the
main routine to stabilize.

The scenario we just described relies on the fact that at time $t$ no node is in
state \srr\ or state \res. We will choose $R_2\gg R_1$, implying that $R_2+3d$
time after a node switched to state \init\ all nodes have ``forgotten'' about
this, i.e., $(R_2,\supp~i)$ is expired and they switched back to state \none\
(unless other \init\ signals interfered). Thus, in the absence of Byzantine
faults, the above requirement is easily achieved with a large probability by
choosing $R_3$ as a uniform distribution over some interval
$[R_2+3d,R_2+\Theta(n R_1)]$: Other nodes will switch to \init\ $\BO(n)$ times
during this interval, each time ``blocking'' other nodes for at most $\BO(R_1)$
time. If the random choice picks any other point in time during this interval, a
resynchronization point occurs. Even if the clock speed of the clock driving
$R_3$ is manipulated in a worst-case manner (affecting the density of the
probability distribution with respect to real time by a factor of at most
$\vartheta$), we can just increase the size of the interval to account for this.

However, what happens if only \emph{some} of the nodes receive an \init\ signal
due to faulty channels or nodes? If the same holds for some of the subsequent
\supp\ signals, it might happen that only a fraction of the nodes reaches the
threshold for switching to state \srr, resulting in an inconsistent reset of
flags and timeouts across the system. Until the respective nodes switch to
state \none\ again, they will not support a resynchronization point again, i.e.,
about $R_1$ time is ``lost''. This issue is the reason for the agreement step
and the timeouts $(R_2,\supp~j)$. In order for any node to switch to state \srr,
there must be at least $n-2f\geq f+1$ non-faulty nodes supporting this. Hence,
all of these nodes recently switched to a state \supp~$j$ for some $j\in V$,
resetting $(R_2,\supp~j)$. Until these timeouts expire, $f+1\in \Omega(n)$
non-faulty nodes will ignore \init\ signals on the respective channels. Since
there are $\BO(n^2)$ channels, it is possible to choose $R_2\in \BO(n R_1)$
such that this may happen at most $\BO(n)$ times in $\BO(n)$ time. Playing with
constants, we can pick $R_3\in \BO(n)$ maintaining that still a constant
fraction of the times are ``good'' in the sense that $R_3$ expiring at a
non-faulty node will result in a good resynchronization point.

\subsection{Timeout Constraints}

\conditionref{cond:timeout_bounds} summarizes the constraints we
require on the timeouts for the core routine and the
resynchronization algorithm to act and interact as intended.

\begin{condition}[Timeout Constraints]\label{cond:timeout_bounds}
Define
\begin{equation}
\lambda:=\sqrt{\frac{25\vartheta-9}{25\vartheta}}\in\left(\frac{4}{5},1\right)
\label{eq:def_lambda},
\end{equation}
$\Delta_g:=(\vartheta+3)T_1$, $\Delta_s:=T_2/\vartheta-2T_1-d$,
$\delta_s:=2T_1+3d$, and $\tilde{\delta}_s:=(\vartheta+2-1/\vartheta)
T_1+4d$. The timeouts need to satisfy the constraints
\begin{eqnarray}
T_1 &\geq & \vartheta 4d\label{eq:T_1}\\
T_2 &\geq & \vartheta\max\left\{T_1+\Delta_g-(4\vartheta^2+16\vartheta+5)d,
\left(3\vartheta+1-\frac{1}{\vartheta}\right)T_1+T_5\right\}\label{eq:T_2}\\
T_3 &\geq & \max\left\{(\vartheta-1)T_2+\vartheta(2T_1+(2\vartheta+4)d),
(2\vartheta^2+3\vartheta-1)T_1-T_2+\vartheta (T_6+5d)\right\}\label{eq:T_3}\\
T_4 &\geq & T_3\label{eq:T_4}\\
T_5 &\geq & \max\left\{\vartheta (T_4+7d)-T_3+(\vartheta-1)T_2,
(\vartheta^2+\vartheta-2)T_1+\vartheta(T_2+T_4+9d)-T_6\right\}\label{eq:T_5}\\
T_6 &\geq & \vartheta
\left(\tilde{\delta}_s-\left(1-\frac{1}{\vartheta}\right)T_1+T_2+2d\right)
>\vartheta\Delta_s\label{eq:T_6}\\
T_7 &\geq &\vartheta(T_2+T_4+T_5+\Delta_s+\tilde{\delta}_s-\Delta_g+d)+T_6-4d
\label{eq:T_7}\\
R_1 &\geq & \vartheta \max\left\{T_7+(4\vartheta+8)d,
\left(2\vartheta+4-\frac{3}{\vartheta}\right)T_1+2T_4
+T_5-\Delta_s-\Delta_g+17d\right\}\label{eq:R_1}\\
R_2 &\geq & \frac{2\vartheta(R_1+(\vartheta+2)T_1+T_2/\vartheta
+(8\vartheta+9)d)(n-f)}{1-\lambda}
\label{eq:R_2}\\
R_3 &= &\mbox{uniformly distributed random variable on }
\left[\vartheta (R_2+3d), \vartheta(R_2+3d)
+8(1-\lambda)R_2\right]\label{eq:R_3}\\
\lambda&\leq &\frac{\Delta_s-\Delta_g-\delta_s}{\Delta_s}.\label{eq:lambda}
\end{eqnarray}
\end{condition}

We need to show for which values of $\vartheta$ this system can be solved.
Furthermore, we would like to allow for the largest possible drift of DARTS
clocks, which necessitates to maximize the ratio
$(T_2+T_4)/(\vartheta(T_2+T_3+4d))$, that is, the minimal gap between pulses
provided that the states of the \darts signals are zero divided by the maximal
time it takes nodes to observe themselves in state \rdy\ with $T_3$ expired
after a pulse (as then they will respond to $\darts_i$ switching to one).
\begin{lemma}\label{lemma:constraints}
Define $\vartheta_{\max}\approx 1.247$ as the positive solution of
$2\vartheta+1=\vartheta^3+\vartheta^2$. Given that $\vartheta<\vartheta_{\max}$,
\conditionref{cond:timeout_bounds} can be satisfied with $T_1,\ldots,T_7,R_1\in
\BO(1)$ and $R_2\in \BO(n)$. The ratio 
\begin{equation*}
\frac{(T_2+T_4)/\vartheta}{T_2+T_3+4d}
\end{equation*}
can be made larger than any constant smaller than
\begin{equation*}
\frac{\vartheta^3+2\vartheta+1}{2\vartheta^4+\vartheta^3}.
\end{equation*}
\end{lemma}
\begin{proof}
First, we identify several redundant inequalities in the system. We have that
\begin{eqnarray*}
\left(2\vartheta+2-\frac{1}{\vartheta}\right)T_1+T_5
&\sr{eq:T_5}{>} & 3\vartheta T_1+T_2+T_4-T_6\\
&\stackrel{(\ref{eq:T_3},\ref{eq:T_4})}{>}& 7\vartheta T_1\\
&>& T_1+\Delta_g-(4\vartheta^2+16\vartheta+5)d,
\end{eqnarray*}
i.e., the left term in the maximum in \inequalityref{eq:T_2} is redundant. The
same holds true for the left terms in the maxima in \inequalityref{eq:T_3} and
\inequalityref{eq:T_5}, since
\begin{eqnarray*}
(2\vartheta^2+3\vartheta-1)T_1-T_2+\vartheta (T_6+5d)
&\sr{eq:T_6}{>}& 3\vartheta T_1+(\vartheta-1)T_2+4d\\
&\sr{eq:T_1}{>}&(\vartheta-1)T_2+\vartheta(2T_1+(2\vartheta+4)d)
\end{eqnarray*}
and
\begin{eqnarray*}
\vartheta (T_4+7d)-T_3+(\vartheta-1)T_2&\sr{eq:T_3}{<}&
\vartheta (T_2+T_4-T_6+7d)\\
&<& (\vartheta^2+\vartheta-2)T_1+\vartheta(T_2+T_4+9d)-T_6.
\end{eqnarray*}
Finally, we can eliminate the right term in the maximum in
\inequalityref{eq:R_1} from the system, as
\begin{eqnarray*}
T_7+(4\vartheta+8)d&\sr{eq:T_7}{>}&
T_2+T_4+T_5+T_6+2\tilde{\delta}_s-\Delta_g+13d\\
&\sr{eq:T_2}{>}&\left(2\vartheta+4-\frac{3}{\vartheta}\right)
T_1 +T_4+2T_5+T_6-\Delta_g+17d\\
&\sr{eq:T_5}{>}&\left(2\vartheta+4-\frac{3}{\vartheta}\right)
T_1 +T_2+2T_4+T_5-\Delta_g+17d.
\end{eqnarray*}

Next, it is not difficult to see that the right hand sides of all inequalities are
strictly increasing in $T_1$ (except for \inequalityref{eq:lambda}, whose right
hand side decreases with $T_1$), implying that w.l.o.g.\ we may set
$T_1:=4\vartheta d$. Similarly, we demand that \inequalityref{eq:T_7},
\inequalityref{eq:R_1}, and \inequalityref{eq:R_2} are satisfied with equality,
i.e.,
\begin{eqnarray*}
T_7 &=&\vartheta (T_2+T_4+T_5)+T_6-(4\vartheta^2+4)d\\
R_1 &= & \vartheta T_7+(4\vartheta^2+8\vartheta)d\\
R_2 &= & \frac{2\vartheta(R_1+T_2/\vartheta
+(4\vartheta^2+16\vartheta+9)d)(n-f)}{1-\lambda}\\
R_3 &= &\mbox{uniformly distributed random variable on }
\left[\vartheta (R_2+3d), \vartheta(R_2+3d)
+8(1-\lambda)R_2\right].
\end{eqnarray*}

We set $T_4:=\alpha T_3$ for a parameter
\begin{equation*}
\alpha\in \left[1,\frac{2\vartheta+1}{\vartheta^3+\vartheta^2}\right),
\end{equation*}
implying that \inequalityref{eq:T_4} holds by definition. The remaining simpler
system is as follows.
\begin{eqnarray}
T_2 &\geq & (8\vartheta^3+8\vartheta^2-4\vartheta)d
+\vartheta T_5\label{eq:T_2_simple}\\
T_3 &\geq & (8\vartheta^3+12\vartheta^2+\vartheta)d
-T_2+\vartheta T_6\label{eq:T_3_simple}\\
T_5 &\geq & (4\vartheta^3+4\vartheta^2+\vartheta)d
+\vartheta(T_2+\alpha T_3)-T_6\label{eq:T_5_simple}\\
T_6 &\geq & (4\vartheta^2+6\vartheta-4)d+T_2\label{eq:T_6_simple}\\
\sqrt{\frac{25\vartheta-9}{25\vartheta}}
&\leq &\frac{T_2/\vartheta-(4\vartheta^2+28\vartheta+4)d}
{T_2/\vartheta-(8\vartheta+1)d}.\nonumber
\end{eqnarray}
Note the above equalities do not affect this system and can be resolved
iteratively once the other variables are fixed. We observe that the right hand
side of \inequalityref{eq:T_2_simple} is increasing in $T_5$, the right hand
side of \inequalityref{eq:T_5_simple} is increasing in $T_3$, and neither $T_3$
nor $T_5$ are present in any further inequalities. Hence, we rule that
\inequalityref{eq:T_3_simple} and \inequalityref{eq:T_5_simple} shall be
satisfied with equality, i.e.,
\begin{eqnarray*}
T_3 &= & (8\vartheta^3+12\vartheta^2+\vartheta)d
-T_2+\vartheta T_6\\
T_5 &= & (\alpha(8\vartheta^4+12\vartheta^3+\vartheta^2)
+(4\vartheta^3+4\vartheta^2+\vartheta))d
-(\vartheta\alpha-1)T_2+(\vartheta^2\alpha-1)T_6
\end{eqnarray*}
and arrive at the subsystem
\begin{eqnarray}
T_2 &\geq & \frac{(\alpha(8\vartheta^5+12\vartheta^4+\vartheta^3)
+(4\vartheta^4+12\vartheta^3+9\vartheta^2-4\vartheta))d
+(\vartheta^3\alpha-\vartheta)T_6}{1+\vartheta-\vartheta^2\alpha}
\label{eq:T_2_simpler}\\
T_6 &\geq & (4\vartheta^2+6\vartheta-4)d+T_2\nonumber\\
T_2 &\geq & \frac{(4\vartheta^3+20\vartheta^2+3\vartheta)d}
{1-\sqrt{(25\vartheta-9)/(25\vartheta)}},\nonumber
\end{eqnarray}
where we used that $1+\vartheta-\vartheta^2\alpha>0$. Now we can see that
\inequalityref{eq:T_2_simpler} is also increasing in $T_6$, set
\begin{equation*}
T_6:=(4\vartheta^2+6\vartheta-4)d+T_2,
\end{equation*}
and obtain
\begin{eqnarray}
T_2 &\geq & \frac{(\alpha(12\vartheta^5+18\vartheta^4-3\vartheta^3)
+(4\vartheta^4+8\vartheta^3+3\vartheta^2))d}
{1+2\vartheta-(\vartheta^3+\vartheta^2)\alpha}
\label{eq:T_2_simplest}\\
T_2 &\geq &
\frac{25(1+\sqrt{(25\vartheta-9)/(25\vartheta)})
(4\vartheta^4+20\vartheta^3+3\vartheta^2)d}{9},\label{eq:lambda_simple}
\end{eqnarray}
exploiting that $1+2\vartheta-(\vartheta^3+\vartheta^2)\alpha>0$.

Since $\alpha$ and thus $\vartheta$ are constantly bounded (and we treat $d$ as
constant as well), we have a feasible solution for $T_2\in \BO(1)$ (considering
asymptotic with respect to $n$). Resolving the equalities we derived for the
other variables, we see that $T_1,\ldots,T_7,R_1\in \BO(1)$ and $R_2\in \BO(n)$
as claimed.

It remains to determine the maximal ratio
$(T_2+T_4)/(\vartheta(T_2+T_3+4d))=(T_2+\alpha T_3)/(\vartheta(T_2+T_3+4d))$ we
can ensure. Obviously, for any value of $\alpha$, fixing either $T_2$ or $T_3$
implies that we want to minimize $T_2$ or maximize $T_3$, respectively. Have a
look at Inequalities \eqref{eq:T_2_simple}--\eqref{eq:T_6_simple} again. The
solution we constructed minimized $T_3$ and subsequently $T_2$, parametrized by
feasible values of $\alpha$. Increase now $T_3$ by $x\in \R^+$ in
\inequalityref{eq:T_3_simple}. Consequently, we may increase $T_6$ in
\inequalityref{eq:T_6_simple} by $x/\vartheta$ compared to our previous
solution (where we minimized all inequalities). Hence, we need to increase $T_5$
by $(\vartheta \alpha -1/\vartheta)x$ according to
\inequalityref{eq:T_5_simple}, and finally $T_2$ by $\vartheta(\vartheta \alpha
-1/\vartheta)x$. Thus, for any feasible $\alpha$ and any $\varepsilon>0$, we can
achieve that $T_2\leq (\vartheta^2 \alpha -1+\varepsilon)T_3$ if we just choose
$x$ large enough. We conclude that we can get arbitrarily close to the ratio
\begin{equation*}
\frac{(\alpha+(\vartheta^2 \alpha -1))T_3}{\vartheta(1+(\vartheta^2 \alpha
-1))T_3} =\frac{\vartheta^2\alpha+\alpha-1}{\vartheta^3\alpha}.
\end{equation*}
Inserting the supremum of admissible values for $\alpha$, this expression
becomes
\begin{equation*}
\frac{(2\vartheta+1)(\vartheta^2+1)-(\vartheta^3+\vartheta^2)}
{\vartheta^3(2\vartheta+1)}=
\frac{\vartheta^3+2\vartheta+1}{2\vartheta^4+\vartheta^3}.
\end{equation*}
This shows the last claim of the lemma, concluding the proof.
\end{proof}

\section{Analysis}\label{sec:analysis}

In this section we derive skew bounds $\Sigma$, as well as accuracy bounds
$T^-,T^+$, such that the presented protocol is a $(W,E)$-stabilizing pulse
synchronization protocol, for proper choices of the set of nodes $W$ and the set
of channels $E$, with skew $\Sigma$ and accuracy bounds $T^-,T^+$ that
stabilizes within time $T(k) \in \BO(kn)$  with probability $1-1/2^{k(n-f)}$,
for any $k\in \N$.

To start our analysis, we need to define the basic requirements for
stabilization. Essentially, we need that a majority of nodes is non-faulty and
the channels between them are correct. However, the first part of the
stabilization process is simply that nodes ``forget'' about past events that are
captured by their timeouts. Therefore, we demand that these nodes indeed have
been non-faulty for a time period that is sufficiently large to ensure that all
timeouts have been reset at least once after the considered set of nodes became
non-faulty.
\begin{definition}[Coherent States]
The subset of nodes $W\subseteq V$ is called \emph{coherent} during the time
interval $[t^-,t^+]$, iff during $[t^--(\vartheta(R_2+3d)
+8(1-\lambda)R_2)-d,t^+]$ all nodes
$i\in W$ are non-faulty, and all channels $S_{i,j}$, $i,j\in W$, are correct.
\end{definition}
We will show that if a coherent set of at least $n-f$ nodes fires a pulse, i.e.,
switches to \acc\ in a tight synchrony, this set will generate pulses
deterministically and with controlled frequency, as long the set remains
coherent. This motivates the following definitions.
\begin{definition}[Stabilization Points]
We call $t$ a \emph{$W$-stabilization point (quasi-stabi\-lization point)} iff
all nodes $i\in W$ switch to \acc\ during $[t,t+2d)$ $([t,t+3d))$.
\end{definition}

\textbf{Throughout this section, we assume the set of coherent nodes $W$ with
$|W|\geq n-f$ to be fixed and consider all nodes in and channels originating from
$V\setminus W$ as (potentially) faulty.} As all our statements refer to nodes in
$W$, we will typically omit the word ``non-faulty'' when referring to the
behaviour or states of nodes in $W$, and ``all nodes'' is short for ``all nodes
in $W$''. Note, however, that we will still clearly distinguish between channels
originating at faulty and non-faulty nodes, respectively, to nodes in $W$.

As a first step, we observe that at times when $W$ is coherent, indeed all nodes
reset their timeouts, basing the respective state transition on proper
perception of nodes in $W$.
\begin{lemma}\label{lemma:counters}
If the system is coherent during the time interval $[t^-,t^+]$, any (randomized)
timeout $(T,s)$ of any node $i\in W$ expiring at a time $t\in[t^-,t^+]$ has been
reset at least once since time $t^--(\vartheta(R_2+3d)
+8(1-\lambda)R_2)$. If $t'$ denotes the time
when such a reset occurred, for any $j\in W$ it holds that
$S_{i,j}(t')=S_j(\tau_{j,i}^{-1}(t'))$, i.e., at time $t'$, $i$ observes $j$ in a
state $j$ attained when it was non-faulty.
\end{lemma}
\begin{proof}
According to \conditionref{cond:timeout_bounds}, the largest possible value of
any (randomized) timeout is $\vartheta(R_2+3d)
+8(1-\lambda)R_2$. Hence, any timeout that is in state~$1$ at
a time smaller than $t^--(\vartheta(R_2+3d)
+8(1-\lambda)R_2)$ expires before time $t_1$ or is reset at least
once. As by the definition of coherency all nodes in $W$ are non-faulty and all
channels between such nodes are correct during $[t^--(\vartheta(R_2+3d)
+8(1-\lambda)R_2)-d,t^+]$, this
implies the statement of the lemma.
\end{proof}

Phrased informally, any corruption of timeout and channel states eventually
ceases, as correct timeouts expire and correct links remember no events that lie
$d$ or more time in the past. Proper cleaning of the memory flags is more
complicated and will be explained further down the road. \textbf{Throughout this
section, we will assume for the sake of simplicity that the system is coherent
at all times} and use this lemma implicitly, e.g.\ we will always assume that
nodes from $W$ will observe all other nodes from $W$ in states that they indeed
had less than $d$ time ago, expiring of randomized timeouts at non-faulty nodes
cannot be predicted accurately, etc. We will discuss more general settings in
\sectionref{sec:generalizations}.

We proceed by showing that once all nodes in $W$ switch to \acc\ in a short
period of time, i.e., a $W$-quasi-stabilization point is reached, the algorithm
guarantees that synchronized pulses are generated deterministically with a
frequency that is bounded both from above and below.

\begin{theorem}\label{theorem:stability}
Suppose $t$ is a $W$-quasi-stabilization point. Then
\begin{itemize}
  \item [(i)] all nodes in $W$ switch to \acc\ exactly once within
  $[t,t+3d)$, and do not leave \acc\ until $t+4d$, and
  \item [(ii)] there will be a $W$-stabilization point
  $t'\in(t+(T_2+T_3)/\vartheta,t+T_2+T_4+5d)$ satisfying that no
  node in $W$ switches to \acc\ in the time interval $[t+3d,t')$ and that
  \item [(iii)] each node $i$'s, $i\in W$, core state machine
  (\figref{fig:main_simple}) is metastability-free during $[t+4d,t'+4d)$.
\end{itemize}
\end{theorem}
\begin{proof}
Proof of (i): Due to \inequalityref{eq:T_1}, a node does not leave the state
\acc\ earlier than $T_1/\vartheta \geq 4d$ time after switching to it.
Thus, no node can switch to \acc\ twice during $[t,t+3d)$. By definition of a
quasi-stabilization point, every node does switch to \acc\ in the interval
$[t,t+3d)\subset [t,t+T_1/\vartheta)$. This proves Statement~(i).

\medskip

Proof of (ii): For each $i\in W$, let $t_i \in [t,t+3d)$ be the time when $i$
switches to \acc. By (i) $t_i$ is well-defined. Further let $t'_i$ be the
infimum of times in $(t_i,\infty)$ when $i$ switches to \rec, \join, or
\prop.\footnote{Note that we follow the convention that $\inf\emptyset = \infty$
if the infimum is taken with respect to a (from above) unbounded subset of
$\R^+_0$.} In the following, denote by $i\in W$ a node with minimal $t'_i$.

We will show that all nodes switch to \prop\ via states \slp, \srw, \wake, and
\rdy\ in the presented order. By (i) nodes do not leave \acc\ before $t+4d$.
Thus at time $t+4d$, each node in $W$ is in state \acc\ and observes each other
node in $W$ in \acc. Hence, each node in $W$ memorizes each other node in $W$ in
\acc\ at time $t+4d$. For each node $j\in W$, let $t_{j,s}$ be the time node
$j$'s timeout $T_1$ expires first after $t_j$. Then $t_{j,s} \in
(t_j+T_1/\vartheta,t_j+T_1+d)$.\footnote{The upper bound comprises an additive
term of $d$ since $T_1$ is reset at some time from $(t_j,t_j+d)$.} Since  $|W|
\geq n-f$, each node $j$ switches to state \slp\ at time $t_{j,s}$. Hence, by
time $t+T_1+4d$, no node will be observed in state \acc\ anymore (until the time
when it switches to \acc\ again).

When a node $j\in W$ switches to state \wake\ at the minimal time $t_w$ larger
than $t_j$, it does not do so earlier than at time $t+T_1/\vartheta +
(1+1/\vartheta)T_1 = t+(1+2/\vartheta)T_1 > t+T_1+5d$. This implies that all
nodes in $W$ have already left \acc\ at least $d$ time ago, since they switched
to it at their respective times $t_j<t+T_1+4d$. Moreover, they cannot switch to
\acc\ again until $t_i'$ as it is minimal and nodes need to switch to \prop\
before switching to \acc. Hence, nodes in $W$ are not
observed in state \acc\ during $(t+T_1+5d,t_i']$, in particular not by node
$j$. Furthermore, nodes in $W$ are not observed in state \rec\ during
$(t_w-d,t_i']$. As it resets its \acc\ and \rec\ flags upon switching to \wake,
$j$ will hence neither switch from \wake\ to \rec\ nor from \emph{trust} to
\sus\ during $(t_w,t_i']$, and thus also not from \rdy\ to \rec.

Now consider node $i$. By the previous observation, it will not switch from
\wake\ to \rec, but to \rdy, following the basic cycle. Consequently, it must
wait for timeout $T_2$ to expire, i.e., cannot switch to \rdy\ earlier than at
time $t+T_2/\vartheta$. As nodes in $W$ clear their \join\ flags upon
switching to state \rdy, by definition of $t_i'$ node $i$ cannot switch from
\rdy\ to \join, but has to switch to \prop. Again, by definition of $i$, it
cannot do so before timeouts $T_3$ or $T_4$ expire, i.e., before time
\begin{equation}
t+\frac{T_2}{\vartheta} + \frac{\min\{T_3,T_4\}}{\vartheta} \sr{eq:T_4}{=}
t+\frac{T_2+T_3}{\vartheta} \sr{eq:T_3}{>}t+T_2+5d.\label{eq:low}
\end{equation}

All other nodes in $W$ will switch to \wake, and for the first time after $t_j$,
observe themselves in state \wake\ at a time within
$(t+T_1+4d,t+T_1(2+\vartheta)+7d)$. Recall that unless they memorize at least
$f+1$ nodes in \acc\ or \rec\ while being in state \wake, they will all switch
to state \rdy\ by time
\begin{equation}
\max\{t+T_2+4d,t+(\vartheta+2)T_1+7d\}\sr{eq:T_2}{=}t+T_2+4d.\label{eq:up}
\end{equation}
As we just showed that $t_i'>t+T_2+5d$, this implies that at time $t+T_2+5d$ all
nodes are observed in state \rdy, and none of them leaves before time $t_i'$.

Now choose $t'$ to be the infimum of times from
$(t+(T_2+T_3)/\vartheta,t+T_2+T_4+4d]$ when a node in $W$ switches to state
\acc.\footnote{Note that since we take the infimum on
$(t+(T_2+T_3)/\vartheta,t+T_2+T_4+4d]$, we have that $t'\leq t+T_2+T_4+4d$.}
Because of \inequalityref{eq:low}, $t'$ is the first time any node $j\in W$ may
switch to \acc\ again after its respective time $t_j$. We will next show that no
node $j\in W$ can switch to \rec\ within $[t_j,t'+2d]$. Since at time $t_i'$
node $j$ does not memorize other nodes from $W$ in state \acc, it will also not
do so during $[t_i',t']$. Hence, it cannot switch from \rdy\ to \rec\ during
$[t_i',t'+2d]$ since it cannot be in state \sus\ during $[t_i',t']$. By
\inequalityref{eq:low}, $j$ cannot switch to \prop\ within
$[t_j,t+(T_2+T_3)/\vartheta)$, and thus its timeout $T_5$ cannot expire until
time
\begin{equation}
t+\frac{T_2+T_3+T_5}{\vartheta}\sr{eq:T_5}{\geq}t+T_2+T_4+7d
\geq t'+3d,\label{eq:acc_on_time}
\end{equation}
making it impossible for $j$ to switch from \prop\ to \rec\ at a time
within $[t_j,t'+3d]$. What is more, a node from $W$ that switches to \acc\ must
stay there for at least $T_1/\vartheta>3d$ time. Thus, by definition of $t'$, no
node $j\in W$ can switch from \acc\ to \rec\ at a time within $[t_j,t'+3d]$.
Hence, no node $j\in W$ can switch to state \rec\ after $t_j$, but earlier than
time $t'+2d$. As nodes reset their \join\ flags upon switching to state \rdy, it
follows that no node in $W$ can switch to other states than \prop\ or \acc\
during $[t+T_2+4d,t'+2d]$. In particular, no node in $W$ resets its \prop\ flags
during $[t+T_2+5d,t'+2d]\supset [t_i',t'+2d]$.

If at time $t'$ a node in $W$ switches to state \acc, $n-2f\geq f+1$ of its
\prop\ flags corresponding to nodes in $W$ are true, i.e., in state $1$. As
the node reset its \prop\ flags at the most recent time when it switched to
\rdy\ and no nodes from $W$ have been observed in \prop\ between this time and
$t_i'$, it holds that $f+1$ nodes in $W$ switched to state \prop\ during
$[t_i',t')$. Since we established that no node resets its \prop\ flags during
$[t_i',t'+2d]$, it follows that all nodes are in state \prop\ by time $t'+d$.
Consequently, all nodes in $W$ will observe all nodes in $W$ in state \prop\
before time $t'+2d$ and switch to \acc, i.e., $t'\in
(t+(T_2+T_3)/\vartheta,t+T_2+T_4+4d)$ is a stabilization point. Statement (ii)
follows.

On the other hand, if at time $t'$ no node in $W$ switches to state \acc, it
follows that $t'=t+T_2+T_4+4d$. As all nodes observe themselves in state \rdy\
by time $t+T_2+5d$, they switch to \prop\ before time $t+T_2+T_4+5d=t'+d$
because $T_4$ expired. By the same reasoning as in the previous case, they
switch to \acc\ before time $t'+2d$, i.e., Statement (ii) holds as well.

\medskip

Proof of (iii): We have shown that within $[t_j,t'+2d]$, any node $j\in W$
switches to states along the basic cycle only. Moreover, such nodes switch to
\acc at some time in $[t',t'+2d]$. Since $T_1\geq 4\vartheta d$, this implies
that no node observing itself in \acc after time $t'$ will leave this state
before time $t'+4d$. To show the correctness of Statement~(iii), it is thus
sufficient to prove that, whenever $j$ switches from state $s$ of the basic
cycle to $s'$ of the basic cycle during time $[t_j+d,t'+2d]\supset
[t+4d,t'+2d]$, the transition from $s$ to \join\ or \rec\ is disabled from the
time it switches to $s'$ until it observes itself in this state. We consider
transitions $tr(\acc,\rec)$, $tr(\wake,\rec)$, $tr(\rdy,\rec)$,
$tr(\rdy,\join)$, and $tr(\prop,\rec)$ one after the other:

\begin{enumerate}
  \item $tr(\acc,\rec)$: We showed that node $j$'s $tr(\acc,\slp)$ is
  satisfied before time $t+4d\leq t+T_1/\vartheta$, i.e., before
  $tr(\acc,\rec)$ can hold, and no node resets its \acc\ flags less
  than $d$ time after switching to state \slp.
  When $j$ switches to state \acc\ again at or after time $t'$, $T_1$ will
  not expire earlier than time $t'+4d$.
  
  \item $tr(\wake,\rec)$: As part of the reasoning in (ii), we derived
  that $tr(\wake,\rec)$ does not hold at nodes from $W$ observing
  themselves in state \wake.
  
  \item $tr(\rdy,\rec)$ and $tr(\rdy,\join)$: Similarly, we proved that
  at no node in $W$, condition $tr(\rdy,\rec)$ or $tr(\rdy,\join)$ can hold
  during $(t+(T_2+T_3)/\vartheta,t'+2d)$, and nodes in $W$ are in
  state \rdy\ during $(t+(T_2+T_3)/\vartheta,t'+d)$ only.
  
  \item $tr(\prop,\rec)$: Finally, the additional slack of $d$ in
  \inequalityref{eq:acc_on_time} ensures that $T_5$ does not expire
  at any node in $W$ switching to state \acc\ during $(t',t'+2d)$
  earlier than time $t'+3d$.
\end{enumerate}
Since $[t_j,t'+4d) \supset [t+3d,t'+4d)$, Statement~(iii) follows.
\end{proof}

Inductive application of \theoremref{theorem:stability} shows that by
construction of our algorithm, nodes in $W$ provably do not suffer from
metastability upsets once a $W$-quasi-stabilization point is reached, as long as
all nodes in $W$ remain non-faulty and the channels connecting them correct.
Unfortunately, it can be shown that it is impossible to ensure this property
during the stabilization period, thus rendering a formal treatment infeasible.
This is not a peculiarity of our system model, but a threat to any model that
allows for the possibility of metastable upsets as encountered in physical chip
designs. However, it was shown that, by proper chip design, the probability of
metastable upsets can be made arbitrarily small~\cite{FFS09:ASYNC09}. \textbf{In
the remainder of this work, we will therefore assume that all non-faulty nodes
are metastability-free in all executions.}

The next lemma reveals a very basic property of the main algorithm that is
satisfied if no nodes may switch to state \join\ in a given period of time. It
states that in order for any non-faulty node to switch to state \slp, there need
to be $f+1$ non-faulty nodes supporting this by switching to state \acc.
Subsequently, these nodes cannot do so again for a certain time window. In
particular, this implies that during the respective time window no node may
switch to \slp.
\begin{lemma}\label{lemma:sleep_one}
Assume that at time $t_s$, some node from $W$ switches to \slp\ and no
node from $W$ is in state \join\ during $[t_s-T_1-d,t^+\,]$.
Then there is a subset $A\subseteq W$ of at least $n-2f$ nodes such that
\begin{itemize}
  \item [(i)] each node from $A$ has been in state \acc\ at some time in the
  interval $(t_s-T_1-d,t_s)$ and
  \item [(ii)] no node from $A$ is in state \prop\ or switches to state \acc\
  during the time interval
  \begin{equation*}
\left(t_s,\min\left\{t_s+\Delta_s,t^+\right\}\right).
\end{equation*}
\end{itemize}
\end{lemma}
\begin{proof}
In order to switch to \slp\ at time $t_s$, a node must have observed
$n-2f$ non-faulty nodes in state \acc\ at times from
$(t_s-T_1,t_s]$, since it resets its \acc\ flags at the time
$t_a\geq t_s-T_1$ (that is minimal with this property) when it switched to
state \acc. Each of these nodes must have been in state \acc\ at some time
from $(t_s-T_1-d,t_s)$, showing the existence of a set $A\subseteq W$
satisfying Statement~(i).

\medskip

We will next prove Statement~(ii). Consider a node $i\in A$.
In order to switch to \prop\ or again to \acc, $i$ must switch to
\join\ first or wait for $T_2$ to expire after switching to state
\acc\ some time after $t_s-2T_1-d$.
However, by assumption the first option is impossible until time
$t^+$, since no nodes are in state \join\ during
$[t_s-T_1-d,t^+]$.
Therefore, $j$ will not be in state \prop\ or switch to state \acc\
again until $t_s-2T_1+T_2/\vartheta-d=t_s+\Delta_s$ or $t^+$,
respectively, whatever is smaller.
This proves Statement~(ii).
\end{proof}

Granted that nodes are not in state \join, this implies that the time windows
during which nodes may switch to \slp\ and \srw, respectively, are
well-separated.
\begin{corollary}\label{coro:window}
Assume that during $[t^--T_1-d,t^+]$ no node from $W$ is in state \join, where
$t^+-t^- \leq \Delta_s$.
Then
\begin{itemize}
  \item [(i)] any time interval $[t_a,t_b]\subseteq [t^-,t^+]$ of
  minimum length containing all switches of nodes in $W$ from
  \acc\ to \slp\ during $[t^-,t^+]$ has length at most $2T_1+3d$,
  and
  
  \item [(ii)] granted that no node from $W$ switches to state \slp\
  during $(t^--(\vartheta+1)T_1-d,t^-)$, any time interval
  $[t_a,t_b]\subseteq [t^-,t^++(1+1/\vartheta)T_1]$ of minimum
  length containing all times in $[t^-,t^++(1+1/\vartheta)T_1]$
  when a node from $W$ switches to \srw\ has length at most
  $\tilde{\delta}_s$.
\end{itemize}
\end{corollary}
\begin{proof}
Consider Statement~(i) first.
If there is no node from $W$ that switches from \acc\ to \slp\ during
$[t^-,t^+]$, the statement is trivially satisfied.

Otherwise, choose any such interval $[t_a,t_b]$.
Since $[t_a,t_b]\neq \emptyset$ is minimal, both at time $t_a$ and
$t_b$ some nodes from $W$ switch to \slp.
Assume by means of contradiction that $t_b-t_a> 2T_1+3d$.
Due to the constraints on $t^-$ and $t^+$, we have that $t_b\leq
t_a+\Delta_s$.
Moreover, during $[t_a-T_1-d,t_b] \subseteq [t^--T_1-d,t^+]$ no node
from $W$ is in state \join.
Thus, we can apply \lemmaref{lemma:sleep_one} to $t_a$ and see that at
least $n-2f\geq f+1$ nodes from $W$ do not switch to \acc\ in the
time interval
\begin{equation*}
\left(t_a,t_a+\Delta_s\right)
\supset(t_b-(2T_1+3d),t_b].
\end{equation*}
As nodes from $W$ leave state \acc\ as soon as $T_1$ expires, these
nodes are not in state \acc\ during $[t_b-(T_1+2d),t_b]$,
implying that they are not observed in this state during
$[t_b-(T_1+d),t_b]$.
It follows that no node in $W$ can observe more than $n-f-1$ different
nodes in state \acc\ during $[t_b-(T_1+d),t_b]$.
As nodes from $W$ clear their \acc\ flags upon switching to \acc\ and
leave state \acc\ after less than $T_1+d$ time, we conclude that
no node from $W$ switches to state \slp\ at time $t_b$.
This is a contradiction, implying that the assumption that $t_b-t_a>
2T_1+3d$ must be wrong and therefore Statement~(i) must be true.

\medskip

To obtain Statement~(ii), observe first that any node from $W$
switching to state \slp\ at some time $t\leq
t^--(\vartheta+1)T_1-d$ switches to state \srw\ before time
$t^-$.
Subsequently, it needs to switch to state \slp\ again in order to be
in state \srw\ at or later than time $t^-$. On the other hand, every
node that switches to \slp\ after time $t^+$ will not switch to \srw\
again before time $t^++(1+1/\vartheta)T_1$. Hence, any node switching to
state \srw\ during the considered interval must switch to \slp\ during
$[t^-,t^+]$.
Applying Statement~(i) to $[t^-,t^+]$ yields that nodes from $W$ can
only switch to \slp\ within a time interval of length at most
$2T_1+3d$.
Considering the fastest and slowest possible transitions from \slp\ to
\srw\ we obtain that nodes from $W$ can switch to \srw\ within a
time interval of length at most
$2T_1+3d+(\vartheta+1)T_1+d-(1+1/\vartheta)T_1 =
\tilde{\delta}_s$.
Statement~(ii) follows.
\end{proof}

We are now ready to advance to proving that good resynchronization points are
likely to occur within bounded time, no matter what the strategy of the
Byzantine faulty nodes and channels is. To this end, we first establish that in
any execution, at most of the times a node switching to state \init\ will result
in a good resynchronization point. This is formalized by the following
definition.
\begin{definition}[Good Times]
Given an execution $\cal E$ of the system, denote by ${\cal E}'$ any
execution satisfying that ${\cal E}|_{[0,t)}'={\cal E}|_{[0,t)}$,
where at time $t$ a node $i\in W$ switches to state \init\ in
${\cal E}'$.
Time $t$ is \emph{good in $\cal E$ with respect to $W$} provided that
for any such ${\cal E}'$ it holds that $t$ is a
good $W$-resynchronization point in ${\cal E}'$.
\end{definition}
The previous statement thus boils down to showing that in any execution, the
majority of the times is good.
\begin{lemma}\label{lemma:good}
Given any execution $\cal E$ and any time interval $[t^-,t^+]$, the volume of
good times in $\cal E$ during $[t^-,t^+]$ is at least
\begin{equation*}
\lambda^2(t^+-t^-)-\frac{11(1-\lambda)R_2}{10\vartheta}.
\end{equation*}
\end{lemma}
\begin{proof}
Assume w.l.o.g.\ that $|W|=n-f$ (otherwise consider a subset of
size $n-f$) and abbreviate
\begin{eqnarray*}
N&:= &\left(\frac{\vartheta(t^+-t^-)}{R_2}+\frac{11}{10}\right)(n-f)\\
&\geq &
\left\lceil\frac{\vartheta(t^+-t^-)+R_2/10}{R_2}\right\rceil(n-f)\\
&\stackrel{(\ref{eq:R_2})}{\geq}&
\left\lceil\frac{\vartheta(t^+-t^-)+\vartheta\theterm/(5(1-\lambda))}
{R_2}\right\rceil(n-f)\\
&\stackrel{(\ref{eq:def_lambda})}{\geq}&
\left\lceil\frac{\vartheta(t^+-t^-+\theterm)}
{R_2}\right\rceil(n-f)\\
&\sr{eq:T_2}{\geq}&
\left\lceil\frac{\vartheta(t^+-t^-+R_1+T_1+4d+\Delta_g)}
{R_2}\right\rceil(n-f).
\end{eqnarray*} 

The proof is in two steps: First we construct a measurable subset of
$[t^-,t^+]$ that comprises good times only.
In a second step a lower bound on the volume of this set is derived.

\medskip

\bfno{Constructing the set:} Consider an arbitrary time $t \in [t^-,t^+]$, and
assume a node $i\in W$ switches to state \init\ at time $t$. When it does so,
its timeout $R_3$ expires. By \lemmaref{lemma:counters} all timeouts of node $i$
that expire at times within $[t^-,t^+]$, have been reset at least once until
time $t^-$. Let $t_{E3}$ be the maximum time not later than $t$ when $R_3$ was
reset. Due to the distribution of $R_3$ we know that
\begin{equation*}
t_{E3} \sr{eq:R_3}{\leq} t- (R_2+3d).
\end{equation*}
Thus, node $i$ is not in state \init\ during time $[t-(R_2+2d),t)$,
and no node $j\in W$ observes $i$ in state \init\ during time $[t-(R_2+d),t)$.
Thereby any node $j$'s, $j\in W$, timeout $(R_2,\supp~i)$
corresponding to node $i$ is expired at time~$t$.

\medskip

We claim that the condition that no node from $W$ is in or observed in one of
the states \res\ or \srr\ at time $t$ is sufficient for $t$ being a
$W$-resynchronization point. To see this, assume that the condition is
satisfied. Thus all nodes $j\in W$ are in states \none\ or $\supp~k$ for some
$k\in \{1,\ldots,n\}$ at time $t$. By the algorithm, they all will switch to
state $\supp~i$ or state \srr\ during $(t,t+d)$.
It might happen that they subsequently switch to another state \supp~$k'$ for
some $k'\in V$, but all of them will be in one of the states with signal \supp\
during $(t+d,t+2d]$. Consequently, all nodes will observe at least $n-f$ nodes
in state \supp\ during $(t',t+2d)$ for some time $t'<t+2d$. Hence, those nodes
in $W$ that were still in state \supp~$i$ (or \supp~$k'$ for some $k'$) at time
$t+d$ switch to state \srr\ before time $t+2d$, i.e., $t$ is a
$W$-resynchronization point.

\medskip

We proceed by analyzing under which conditions $t$ is a good
$W$-resynchronization point.
Recall that in order for $t$ to be good, it has to hold that no node from
$W$ switches to state \slp\ during $(t-\Delta_g,t)$ or is in
state \join\ during $(t-T_1-d,t+4d)$.

\medskip

We begin by characterizing subsets of good times within $(t_r,t_r')
\subset [t^-,t^+]$, where $t_r$ and $t_r'$ are times such that
during $(t_r,t_r')$ no node from $W$ switches to state \srr.
Due to timeout
\begin{equation*}
R_1 \stackrel{(\ref{eq:R_1})}{\geq} (4\vartheta+2)d,
\end{equation*}
we know that during $(t_r+R_1+2d,t_r')$, no node from $W$ will be
in, or be observed in, states \srr\ or \res.
Thus, if a node from $W$ switches to \init\ at a time within
$(t_r+R_1+2d,t_r')$, it is a $W$-resynchronization point.
Further, all nodes in $W$ will be in state \dorm\ during
$(t_r+R_1+2d,t_r'+4d)$.
Thus all nodes in $W$ will be observed to be in state \dorm\ during
$(t_r+R_1+3d,t_r'+4d)$, implying that they are not in state
\join\ during $(t_r+R_1+3d,t_r'+4d)$.
In particular, any time $t\in (t_r+R_1+T_1+4d,t_r')$ satisfies that no
node in $W$ is in state \join\ during $(t-T_1-d,t+4d)$.

Further define $t_a$ to be the infimum of times from
$(t_r+R_1+T_1+4d,t_r']$ when a node from $W$ switches to state
\slp.
By \corollaryref{coro:window}, no node from $W$ switches to state
\slp\ during $(t_a+\delta_s,\min\{t_a+\Delta_s,t_r'\})$.
Hence, if $t_a < \infty$, all times in both $(t_r+R_1+T_1+4d+\Delta_g,t_a)$
and $(t_a+\delta_s+\Delta_g,\min\{t_a+\Delta_s,t_r'\})$ are good.

In case $t_a < t_r'-\Delta_s$ we can repeat the reasoning, defining that
$t_a'$ is the infimum of times from $[t_a+\Delta_s,t_r']$ when a
node switches to state \slp.
By analogous arguments as before we see that all times in the sets
$[t_a+\Delta_s,t_a')$ and
$(t_a'+\delta_s+\Delta_g,\min\{t_a'+\Delta_s,t_r'\})$ are good.

By induction on the times $t_a,t_a',\ldots, t_a^{k}$ (halting once
$t_a^{k} \geq t_r'-\Delta_s$), we infer that the total volume of times from
$(t_r,t_r')$ as well as from $(t_r+R_1+T_1+4d+\Delta_g,t_r')$ that
is good is at least
\begin{eqnarray}
&&\left\lfloor\frac{t_r'-(t_r+R_1+T_1+4d+\Delta_g)}{\Delta_s}
\right\rfloor (\Delta_s-\Delta_g-\delta_s) >\notag\\
&&\frac{t_r'-(t_r+R_1+T_1+4d+\Delta_g+\Delta_s)}{\Delta_s}
(\Delta_s-\Delta_g-\delta_s)\label{eq:good}\;.
\end{eqnarray}

In other words, up to a constant loss in each interval $(t_r,t_r')$, a
constant fraction of the times are good.

\medskip

\bfno{Volume of the set:} In order to infer a lower bound on the volume of good
times during $[t^-,t^+]$, we subtract from $[t^-,t^+]$ all intervals
$[t_r,t_r+R_1+T_1+4d+\Delta_g]$, where a node from $W$ switches
to \srr\ at a time $t_r$ within
$[t^--(R_1+T_1+4d+\Delta_g),t^+]$. Formally define
\begin{equation*}
\bar{G} = \bigcup_{\substack{t_r\in [t^--(R_1+T_1+4d+\Delta_g),t^+]\\
\exists i\in W:\,i\text{ switches to }\srr\ \text{at }t_r}}
[t_r,t_r+R_1+T_1+4d+\Delta_g].
\end{equation*}
What remains is the set $[t^-,t^+] \setminus \bar{G}$, that has as
subset the union of intervals $(t_r+R_1+T_1+4d+\Delta_g,t_r')
\subseteq [t^-,t^+]$, where $t_r$ and $t_r'$ are times at which a
node from $W$ switches to \srr\, and no node from $W$ switches to
\srr\ within $(t_r,t_r')$.
Note that for each such interval we already know it contains a certain amount
of good times because of \inequalityref{eq:good}.
In order to lower bound the good times in $[t^-,t^+]$, it is thus feasible to lower
bound the volume and number of connected components (i.e., maximal intervals) of
any subset of $[t^-,t^+] \setminus \bar{G}$.

\medskip

Observe that any node in $W$ does not switch
to state \init\ more than
\begin{equation}
\left\lceil\frac{t^+-t^-+R_1+T_1+4d+\Delta_g}{R_3}\right\rceil \stackrel{(\ref{eq:R_3})}{\leq}
\left\lceil\frac{t^+-t^-+R_1+T_1+4d+\Delta_g}{R_2+d}\right\rceil \leq
\frac{N}{n-f}\label{ieq:frac1}
\end{equation}
times during $[t^--(R_1+T_1+4d+\Delta_g),t^+]$.

Now consider the case that a node in $W$ switches to state \srr\ at a
time $t$ satisfying that no node in $W$ switched to state \init\
during $(t-(8\vartheta+6)d,t)$.
This necessitates that this node observes $n-f$ of its channels in
state \supp\ during $(t-(2\vartheta+1)d,t)$, at least $n-2f\geq
f+1$ of which originate from nodes in $W$.
As no node from $W$ switched to \init\ during $(t-(8\vartheta+6)d,t)$,
every node that has not observed a node $i\in V\setminus W$ in
state \init\ at a time from $(t-(8\vartheta+4)d,t)$ when
$(R_2,\supp~i)$ is expired must be in a state whose signal is
\none\ during $(t-(2\vartheta+2)d,t)$ due to timeouts.
Therefore its outgoing channels are not in state \supp\ during
$(t-(2\vartheta+1)d,t)$.
By means of contradiction, it thus follows that for each node $j$ of
the at least $f+1$ nodes (which are all from $W$), there exists a
node $i\in V\setminus W$ such that node $j$ resets timeout
$(R_2,\supp~i)$ during the time interval $(t-(8\vartheta+4)d,t)$.

The same reasoning applies to any time $t'\not \in (t-(8\vartheta+6)d,t)$
satisfying that some node in $W$ switches to state \srr\ at time $t'$ and no
node in $W$ switched to state \init\ during $(t'-(8\vartheta+6)d,t')$. Note that
the set of the respective at least $f+1$ events (corresponding to the at least
$f+1$ nodes from $W$) where timeouts $(R_2,\supp~i)$ with $i\in V\setminus W$
are reset and the set of the events corresponding to $t$ are disjoint. However,
the total number of events where such a timeout can be reset during
$[t^--(R_1+T_1+4d+\Delta_g),t^+]$ is upper bounded~by
\begin{equation}
|V\setminus W||W|\left\lceil\frac{t^+-t^-+R_1+T_1+4d+\Delta_g}{R_2/\vartheta}\label{ieq:frac2}
\right\rceil<(f+1)N,
\end{equation}
i.e., the total number of channels from nodes not in $W$ ($|V\setminus
W|$ many) to nodes in $W$ multiplied by the number of times the
associated timeout can expire at the receiving node in $W$ during
$[t^--(R_1+T_1+4d+\Delta_g),t^+]$.

\medskip

With the help of inequalities~\eqref{ieq:frac1} and~\eqref{ieq:frac2},
we can show that $\bar{G}$ can be covered by less than $2N$
intervals of size $(R_1+T_1+4d+\Delta_g)+(8\vartheta+6)d$ each. By
\inequalityref{ieq:frac1}, there are no more than $N$ times
$t\in [t^--(R_1+T_1+4d+\Delta_g),t^+]$ when a non-faulty node switches to
\init\ and thus may cause others to switch to state \srr\ at times in
$[t,t+(8\vartheta+6)d]$. Similarly, \inequalityref{ieq:frac2} shows that
the channels from $V\setminus W$ to $W$ may cause at most $N-1$ such times
$t\in [t^--(R_1+T_1+4d+\Delta_g),t^+]$, since any such time
requires the existence of at least $f+1$ events where timeouts
$(R_2,\supp~i)$, $i\in V\setminus W$, are reset at nodes in $W$, and the
respective events are disjoint. Thus, all times $t_r\in
[t^--(R_1+T_1+4d+\Delta_g),t^+]$ when some node $i\in W$ switches to \srr\
are covered by at most $2N-1$ intervals of length $(8\vartheta+6)d$.

This results in a cover $\bar{G}' \supseteq \bar{G}$ consisting of at
most $2N-1$ intervals that satisfies that
\begin{eqnarray*}
\operatorname{vol}\left( \bar{G} \right) \leq \operatorname{vol}\left(
\bar{G}' \right) < 2N (R_1+T_1+\Delta_g+(8\vartheta+10)d).\label{ieq:Gbar}
\end{eqnarray*}

Summing over the at most $2N$ intervals that remain in $[t^-,t^+]
\setminus \bar{G}'$ and using \inequalityref{eq:good}, we
conclude that the volume of good times during $[t^-,t^+]$ is at
least
\begin{eqnarray*}
&&\frac{t^+-t^--2N(R_1+T_1+(8\vartheta+10)d+\Delta_g+\Delta_s)}
{\Delta_s}(\Delta_s-\Delta_g-\delta_s)\\
&=&\frac{t^+-t^--2N\theterm}{\Delta_s}(\Delta_s-\Delta_g-\delta_s)\\
&\stackrel{(\ref{eq:lambda})}{\geq}&
\lambda \left(t^+-t^--2\left(\frac{\vartheta(t^+-t^-)}{R_2}
+\frac{11}{10}\right)(n-f)\theterm\right)\\
&=&
\lambda
\left(1-\frac{2\vartheta\theterm(n-f)}{R_2}\right) (t^+-t^-)\\
& &
-\frac{11\lambda\theterm (n-f)}{5}\\
&\sr{eq:R_2}{\geq}&\lambda^2(t^+-t^-)
-\frac{11(1-\lambda)R_2}{10\vartheta},
\end{eqnarray*}
as claimed. The lemma follows.
\end{proof}

We are now in the position to prove our second main theorem, which states that
a good resynchronization point occurs within $\BO(R_2)$ time with overwhelming
probability.

\begin{theorem}\label{theorem:resync}
Denote by $\hat{E}_3:=\vartheta (R_2+3d)+8(1-\lambda)R_2+d$ the maximal value
the distribution $R_3$ can attain plus the at most $d$ time until $R_3$ is reset
whenever it expires.
For any $k\in \N$ and any time $t$, with probability at least
$1-(1/2)^{k(n-f)}$ there will be a good $W$-resynchronization
point during $[t,t+(k+1)\hat{E}_3]$.
\end{theorem}
\begin{proof}
Assume w.l.o.g.\ that $|W|=n-f$ (otherwise consider a subset of size
$n-f$).
Fix some node $i\in W$ and denote by $t_0$ the infimum of times from
$[t,t+(k+1)\hat{E}_3]$ when node $i$ switches to \init.
We have that $t_0< t+\hat{E}_3$.
By induction, it follows that node~$i$ will switch to state \init\ at
least another $k$ times during $[t,t+(k+1)\hat{E}_3]$ at the
times $t_1<t_2<\ldots<t_k$.
We claim that each such time $t_j$, $j\in \{1,..,k\}$, has an
independently by $1/2$ lower bounded probability of being good
and therefore being a good $W$-resynchronization point.

We prove this by induction on $j$: As induction hypothesis, suppose
for some $j\in \{1,\ldots,k-1\}$, we showed the statement for
$j'\in \{1,\ldots,j-1\}$ and the execution of the system is fixed
until time $t_{j-1}$, i.e., ${\cal E}|_{[0,t_{j-1}]}$ is given.
Now consider the set of executions that are extensions of ${\cal
E}|_{[0,t_{j-1}]}$ and have the same clock functions as $\cal E$.
For each such execution ${\cal E}'$ it holds that ${\cal
E}'|_{[0,t_{j-1}]}={\cal E}|_{[0,t_{j-1}]}$, and all nodes'
clocks make progress in ${\cal E}'$ as in ${\cal E}$.
Clearly each such ${\cal E}'$ has its own time $t_j <
t+(j+1)\hat{E}_3$ when $R_3$ expires next after $t_{j-1}$ at node
$i$, and $i$ switches to \init.
We next characterize the distribution of the times $t_{j}$.

As the rate of the clock driving node $i$'s $R_3$ is between $1$ and $\vartheta$,
$t_{j}>t_{j-1}$ is within an interval, call it $[t^-,t^+]$,
of size at most
\begin{equation*}
t^+-t^-\leq 8(1-\lambda)R_2,
\end{equation*}
regardless of the progress that $i$'s clock $C$ makes in any execution~${\cal
E}'$.

Certainly we can apply \lemmaref{lemma:good} also to each of the
${\cal E}'$, showing that the volume of times from $[t^-,t^+]$
that are \emph{not} good in ${\cal E}'$ is at most
\begin{equation*}
(1-\lambda^2)(t^+-t^-)+\frac{11(1-\lambda)R_2}{10\vartheta}.
\end{equation*}

Since clock $C$ can make progress not faster than at rate
$\vartheta$ and the probability density of $R_3$ is
constantly $1/(8(1-\lambda)R_2)$ (with respect to
the clock function $C$), we obtain that the probability of $t_{j}$ not
being a good time is upper bounded~by
\begin{equation*}
\frac{(1-\lambda^2)(t^+-t^-)+11(1-\lambda)R_2/(10\vartheta)}
{8(1-\lambda)R_2/\vartheta}\leq
\vartheta(1-\lambda^2)+\frac{11}{80}\sr{eq:def_lambda}{<}
\vartheta\frac{9}{25\vartheta}+\frac{7}{50}=\frac{1}{2}.
\end{equation*}
Here we use that the time when $R_3$ expires is independent of ${\cal
E}'|_{[0,t_{j-1}]}$.

We complete our reasoning as follows.
Given ${\cal E}|_{[0,t_{j-1}]}$, we permit an adversary to choose
${\cal E}'$, including random bits of all nodes and full
knowledge of the future, with the exception that we deny it
control or knowledge of the time $t_j$ when $R_3$ expires at node
$i$, i.e., ${\cal E}'$ is an imaginary execution in which $R_3$
does not expire at $i$ at any time greater than $t_{j-1}$.
Note that for the good $W$-resynchronisation points we considered, the
choice of ${\cal E}'$ does not affect the probability that
$t_1,\ldots,t_{j-1}$ are good $W$-resynchronization points: The
conditions referring to times greater than a
$W$-resynchronisation point $t$, i.e., that all nodes in $W$
switch to state \srr\ during $(t,t+2d)$ and no node in $W$ shall
be in state \join\ during $(t-T_1-d,t+4d)$, are already fully
determined by the history of the system until time $t$.
As we fixed ${\cal E}'$, the behaviour of the clock driving $R_3$ is
fixed as well.
Next, we determine the time $t_j$ when $R_3$ expires according to its
distribution, given the behaviour of node $i$'s clock.
The above reasoning shows that time $t_j$ is good in ${\cal E}'$ with
probability at least $1/2$, independently of ${\cal
E}'|_{[0,t_{j-1}]}={\cal E}|_{[0,t_{j-1}]}$.
We define that ${\cal E}|_{[0,t_j)}={\cal E}'|_{[0,t_j)}$ and in
${\cal E}$ node $i$ switches to state \init\ (because $R_3$
expired).
As --- conditional to the clock driving $R_3$ and $t_{j-1}$ being
specified --- $t_j$ is independent of ${\cal E}|_{[0,t_j)}$, ${\cal
E}$ is indistinguishable from ${\cal E}'$ until time $t_j$.
Because $t_j$ is good with probability at least $1/2$ independently of
${\cal E}|_{[0,t_{j-1}]}'={\cal E}|_{[0,t_{j-1}]}$, so it is in
${\cal E}$.
Hence, in ${\cal E}$ $t_j$ is a good $W$-resynchronization point with
probability $1/2$, independently of ${\cal E}|_{[0,t_{j-1}]}$.
Since ${\cal E}'$ was chosen in an adversarial manner, this completes
the induction step.

In summary, we showed that for \emph{any} node in $W$ and \emph{any}
execution (in which we do not manipulate the times when $R_3$
expires at the respective node), starting from the second time
during $[t,t+(k+1)\hat{E}_3]$ when $R_3$ expires at the
respective node, there is a probability of at least $1/2$ that
the respective time is a good $W$-resynchronization point.
Since we assumed that $|W|=n-f$ and there are at least $k$ such times
for each node in $W$, this implies that having no good
$W$-resynchronization point during $[t,t+(k+1)\hat{E}_3]$ is as
least as unlikely as $k(n-f)$ unbiased and independent coin flips
all showing tail, i.e., $(1/2)^{k(n-f)}$.
This concludes the proof.
\end{proof}

Having established that eventually a good $W$-resynchronization point
$t_g$ will occur, we turn to proving the convergence of the main
routine. We start with a few helper statements wrapping up that a good
resynchronization point guarantees proper reset of flags and timeouts
involved in the stabilization process of the main routine.

\begin{lemma}\label{lemma:clean}
Suppose $t_g$ is a good $W$-resynchronization point. Then
\begin{itemize}
  \item [(i)] each node $i\in W$ switches to \pass\ at a time
  $t_i\in (t_g+4d,t_g+(4\vartheta+3)d)$ and observes itself in state \dorm\
  during $[t_g+4d,\tau_{i,i}(t_i))$,
  \item [(ii)] $\Mem_{i,j,\join}|_{[\tau_{i,i}(t_i),t_\join]}\equiv 0$ for all
  $i,j\in W$, where $t_\join\geq t_g+4d$ is the infimum of all times greater than
  $t_g-T_1-d$ when a node from $W$ switches to \join,
  \item [(iii)] $\Mem_{i,j,\srw}|_{[\tau_{i,i}(t_i),t_s]}\equiv 0$ for
  all $i,j\in W$, where $t_s\geq t_g+(1+1/\vartheta)T_1$ is the infimum of all
  times greater or equal to $t_g$ when a node from $W$ switches to \srw,
  \item [(iv)] no node from $W$ resets its \srw\ flags during
  $[t_g+(1+1/\vartheta)T_1,t_g+R_1/\vartheta]$, and
  \item [(v)] no node from $W$ resets its \join\ flags due to switching to \pass\ during
  $[t_g+(1+1/\vartheta)T_1,t_g+R_1/\vartheta]$.
\end{itemize}
\end{lemma}
\begin{proof}
All nodes in $W$ switch to state \srr\ during $(t_g,t_g+2d)$ and switch to state
\res\ when their timeout of $\vartheta 4d$ expires, which does not happen until
time $t_g+4d$. Once this timeout expired, they switch to state \pass\ as soon as
they observe themselves in state \res, i.e., by time $t_g+(4\vartheta+3d)$.
Hence, every node $i\in W$ does not observe itself in state \res\
within $[t_g+3d,\tau_{i,i}(t_i))$, and therefore is in state \dorm\ during
$[t_g+3d,\tau_{i,i}(t_i)]$. This implies that it observes itself in state \dorm\
during $[t_g+4d,\tau_{i,i}(t_i))$, completing the proof of Statement~(i).

\medskip

Moreover, from the definition of a good $W$-resynchronization point we
have that no nodes from $W$ are in state \join\ at times in
$[t_g-T_1-d,t_\join)$.
Statement~(ii) follows, as every node from $W$ resets its \join\ flags upon
switching to state \pass\ at time $t_i$.

\medskip

Regarding Statement~(iii), observe first that no nodes from $W$ are in
state \srw\ during $(t_g-d,t_g+(1+1/\vartheta)T_1)$ for the
following reason: By definition of a good $W$-resynchronization
point no node from $W$ switches to \slp\ during
$(t_g-\Delta_g,t_g) \supseteq (t_g-(\vartheta+1)T_1-3d,t_g)$.
Any node in $W$ that is in states \slp\ or \srw\ at time
$t_g-(\vartheta+1)T_1-3d$ switches to state \wake\ before time
$t_g-d$ due to timeouts.
Finally, any node in $W$ switching to \slp\ at or after time $t_g$
will not switch to state \srw\ before time
$t_g+(1+1/\vartheta)T_1$. The observation follows.

Since nodes in $W$ reset their \srw\ flags at some time from
\begin{equation*}
[t_i,\tau_{i,i}(t_i)]\subset (t_g+3d,t_g+(4\vartheta+4)d)\sr{eq:T_1}{\subseteq}
(t_g+3d,t_g+(1+1/\vartheta)T_1),
\end{equation*}
Statement~(iii) follows.

\medskip

Statements~(iv) and~(v) follow from the fact that all nodes in $W$
switch to state \pass\ until time
\begin{equation*}
t_g+(3+4\vartheta)d\sr{eq:T_1}{\leq}t_g+\left(1+\frac{1}{\vartheta}\right)T_1-d,
\end{equation*}
while timeout $(R_1,\srr)$ must expire first in order to switch to
\dorm\ and subsequently \pass\ again.
\end{proof}

Before we proceed, in the next lemma we make the basic yet crucial observation
that after a good $W$-resynchronization point $t_g$, no node from $W$ will
switch to state \join\ until either time $t_g+T_7/\vartheta+4d$ or
$T_6/\vartheta$ time after the first non-faulty node switched to \srw\ again
after $t_g$.
By proper choice of $T_6$ and $T_7>T_6$, this will guarantee that nodes from $W$
do not switch to \join\ prematurely during the final steps of the stabilization
process.

\begin{lemma}\label{lemma:switch}
Suppose $t_g$ is a good $W$-resynchronization point.
Denote by $t_s$ the infimum of times greater than $t_g$ when a node in
$W$ switches to state \srw\ and by $t_\join$ the infimum of times
greater than $t_g-T_1-d$ when a node in $W$ switches to state
\join.
Define
$t^+:=t_g+\Delta_s-\Delta_g+\tilde{\delta}_s+T_2+T_4+T_5+d$.
Then, starting from time $t_g+4d$, $tr(\rec,\join)$ is not satisfied
at any node in $W$ until time
\begin{equation*}
\min\left\{t_s+\frac{T_6}{\vartheta},t_g+\frac{T_7}{\vartheta}+4d\right\}\geq
\min\{t_s+\Delta_s,t^+\}
\end{equation*}
and $t_\join$ is larger than this time.
\end{lemma}
\begin{proof}
By Statements~(ii) and~(iii) of \lemmaref{lemma:clean} and
\inequalityref{eq:T_1}, we have that $t_s\geq t_g+T_1+4d\geq
t_g+(4\vartheta+4)d$ and $t_\join\geq t_g+4d$. Consider a node $i\in W$ not
observing itself in state \dorm\ at some time $t\in [t_g+4d,t_\join]$. According
to Statements~(i) and~(ii) of \lemmaref{lemma:clean}, the threshold condition of
$f+1$ nodes memorized in state \join\ cannot be satisfied at such a node. By
statements~(i) and~(iii) of the lemma, the threshold condition of $f+1$ nodes
memorized in state \srw\ cannot be satisfied unless $t>t_s$. Hence, if at time
$t$ a node from $W$ satisfies that it observes itself in state \act\ and $T_6$
expired, we have that $t>t_s+T_6/\vartheta$. Moreover, by Statement~(i) of
\lemmaref{lemma:clean}, we have that if $T_7$ is expired at any node in $W$ at
time $t$, it holds that $t>t_g+T_7/\vartheta+4d$. Altogether, we conclude that
$tr(\rec,\join)$ is not satisfied at any node in $W$ during
\begin{equation*}
\left[t_g+4d,\min\left\{t_s+\frac{T_6}{\vartheta},
t_g+\frac{T_7}{\vartheta}+4d\right\}\right]
\stackrel{(\ref{eq:T_6},\ref{eq:T_7})}{\supseteq}
\left[t_g+4d,\min\{t_s+\Delta_s,t^+\}\right].
\end{equation*}
In particular, $t_\join$ must be larger than the upper boundary of this
interval, concluding the proof.
\end{proof}

Before we can move on to proving eventual stabilization, we need one last key
lemma. Essentially, it states that after a good $W$-resynchronization point, any
node in $W$ switches to \rec\ or to \srw\ within bounded time, and all
nodes in $W$ doing the
latter will do so in rough synchrony, i.e., within a time window of
$\tilde{\delta}_s$. Using the previous lemma, we can show that this happens
before the transition to \join\ is enabled for any node.

\begin{lemma}\label{lemma:rec}
Suppose $t_g$ is a good $W$-resynchronization point and use the
notation of \lemmaref{lemma:switch}.
Then either 
\begin{itemize}
  \item [(i)] $t_s< t^+-\Delta_s$ and any node in $W$ switches to state \srw\
at some time in $[t_s,t_s+\tilde{\delta}_s]$ or is observed in state \rec\
during $[t_s+T_1+T_5,t_\join]$ or
\item [(ii)] all nodes in $W$ are observed in state \rec during $[t^+,t_\join]$.
\end{itemize}
\end{lemma}
\begin{proof}
By \lemmaref{lemma:switch}, it holds that
\begin{equation}
t_\join>\min\{t_s+\Delta_s,t^+\}.\label{eq:tjoin}
\end{equation}
For any node in $W$, consider the supremum $t$ of all times smaller or equal to
$t_g-\Delta_g$ when it switched to \slp. After that, it observed itself in state
\wake\ before time
\begin{equation}
t+(\vartheta+1)T_1+3d\leq t_g-T_1-d\label{eq:awake}
\end{equation}
(w.l.o.g.\ assuming that the node has ever been in state \slp\ since it became
non-faulty). By definition of a good $W$-resynchronization point, nodes in $W$
are not in state \join\ during $(t_g-T_1-d,t_\join)$ and do not switch to state
\slp during $(t_g-\Delta_g,t_s)$. Continuing to execute the basic cycle
after time $t_g-d > t_g-T_1-d$ thus necessitates that the node is in
one of the states \wake, \rdy, \prop\ or \acc at time $t_g-d$.

Assume that it is in state \wake\ (we just showed that if not, it
already was in \wake\ by time $t_g-d$).
As timeout $T_2$ cannot have been reset later than time
$t-T_1/\vartheta+d\leq t_g-\Delta_g-T_1/\vartheta+d$ at the
respective node, it observes itself in state \rdy\ by time
$t_g-\Delta_g-T_1/\vartheta+T_2+2d$, in state \prop\ by time
$t_g-\Delta_g-T_1/\vartheta+T_2+T_4+3d$, in state \acc\ by time
$t_g-\Delta_g-T_1/\vartheta+T_2+T_4+T_5+4d$, in state \slp\ by
time $t_g-\Delta_g+(1-1/\vartheta)T_1+T_2+T_4+T_5+5d$, and must
switch to \srw\ before time $t^+-\Delta_s$.

We next distinguish between two cases:

\medskip

{\bf Case 1:} Assume that $t_s<t^+-\Delta_s$. We already established
that no node in $W$ observes itself in states \slp\ or \srw\ at time $t_s$, and
by \inequalityref{eq:awake}, any node in $W$ observing itself in states \wake\
or \rdy\ reset its \acc\ flags after time $t_g-T_1-d$. Denote by $t_s'\in
(t_s-(\vartheta+1)T_1-d,t_s-(1+1/\vartheta)T_1)$ the minimal time greater or
equal to $t_g$ when a node from $W$ switches to state \slp; by the timeout
condition for switching from \slp\ to \srw\ and the definitions of $t_s$ and
good $W$-resynchronization points, such a time exists. According to
\lemmaref{lemma:sleep_one}, at least $f+1$ nodes have been in state \acc\ at
times in $(t_s'-T_1-d,t_s')$. By Statements~(i) and~(iii) of
\lemmaref{lemma:clean}, all nodes are in state \pass until at least time $t_s$.
Hence, any nodes from $W$ observing themselves in state \wake\ or \rdy\ at time
$t_s'+d$ satisfy $tr(\wake,\rec)$ or $tr(\emph{unsuspect},\sus)$, respectively.
Consequently, they will leave these states no later than time
\begin{equation*}
t_s'+(2\vartheta+2)d\leq
t_s-\left(1+\frac{1}{\vartheta}\right)T_1+(2\vartheta+2)d
\sr{eq:T_1}{\leq}t_s-4d.
\end{equation*}

It follows that any nodes from $W$ that are in state \prop\ at time $t_s$
observe themselves in this state since at least time $t_s-3d$, implying that
they switch to states \acc\ or \rec\ by time $t_s+T_5-3d$. After switching to
\acc, a node from $W$ switches to \slp\ and subsequently to \srw\ within another
$(2\vartheta+1)T_1+2d$ or is observed in state \rec\ after less than $T_1+2d$
time. Thus, as
\begin{equation*}
t_\join>t_s+\Delta_s-(\vartheta-1/\vartheta)T_1-d)\sr{eq:T_2}>t_s+T_1+T_5-d,
\end{equation*}
all nodes in $W$ that do not switch to state \srw\ during
\begin{equation*}
[t_s,t_s+(2\vartheta+1)T_1+T_5-d]\sr{eq:T_2}{\subseteq}
\left[t_s,t_s+\Delta_s-\left(\vartheta-\frac{1}{\vartheta}\right)T_1-d\right]
\subseteq\left[t_s,t_s'+\Delta_s+\left(1+\frac{1}{\vartheta}\right)T_1\right]
\end{equation*}
are observed in state \rec\ at time $t_s+T_1+T_5$. Because
$t_\join>t_s+\Delta_s-(\vartheta-1/\vartheta)T_1-d$ and no nodes from $W$ switch
to state \slp\ during $(t_g-\Delta_g,t_s)$, we can apply Statement~(ii) of
\corollaryref{coro:window} to conclude that no nodes from $W$ switch to state
\srw\ during
\begin{equation*}
\left(t_s+\tilde{\delta}_s,
t_s'+\Delta_s+\left(1+\frac{1}{\vartheta}\right)T_1\right],
\end{equation*}
i.e., any node from $W$ that does not switch to state \srw\ during
$[t_s,t_s+\tilde{\delta}_s]$ is observed in state \rec\ during
$[t_s+T_1+T_5,t_\join]$. Statement~(i) follows.

\medskip

{\bf Case 2:} Assume $t_s\geq t^+-\Delta_s$. Then by \inequalityref{eq:tjoin},
$t_\join \geq t^+$ holds. By definition of $t_s$, the first node in $W$
switching to \srw\ after $t_g$ does so at time $t_s$, and by the arguments given
above, no node from $W$ executing the basic cycle does so later than
$t^+-\Delta_s < t^+-d$. Hence, it is observed in state \rec\ during
$[t^+,t_\join]$, as it cannot leave \rec through \join\ before time $t_\join$.
Hence Statement~(ii) holds and the proof concludes.
\end{proof}

We have everything in place for proving that a good resynchronization point
leads to stabilization within $R_1/\vartheta-3d$ time.
\begin{theorem}\label{theorem:stabilization}
Suppose $t_g$ is a good $W$-resynchronization point. Then there is a
quasi-stabilization point during $(t_g,t_g+R_1/\vartheta-3d]$.
\end{theorem}
\begin{proof}
For simplicity, assume during this proof that $R_1=\infty$, i.e., by
Statement~(i) of \lemmaref{lemma:clean} all nodes in $W$ observe themselves in
states \pass\ or \act\ at times greater or equal to $t_g+(4\vartheta+4)d$. We
will establish the existence of a quasi-stabilization point at a time larger
than $t_g$ and show that it is upper bounded by $t_g+R_1/\vartheta-3d$. Hence
this assumption can be made w.l.o.g., as the existence of the
quasi-stabilization point depends on the execution up to time
$t_g+R_1/\vartheta$ only, and $R_1$ cannot expire before this time at any node
in $W$. We use the notation of \lemmaref{lemma:switch}. By Statements~(ii) of
\lemmaref{lemma:clean} and \inequalityref{eq:T_1}, we have that $t_s\geq
t_g+T_1+4d\geq t_g+(4\vartheta+4)d$. By \lemmaref{lemma:switch}, it holds that
$t_\join>\min\{t_s+\Delta_s,t^+\}$. We differentiate several cases.

\medskip

{\bf Case 1:} Assume $t_s\geq t^+-\Delta_s$. According to
\lemmaref{lemma:clean}, all nodes in $W$ switched to state \pass\ during
$(t_g+4d,t_g+(3+4\vartheta)d)$, implying that at any node in $W$, $T_7$ will
expire at some time from $(t_g+T_7/\vartheta+4d,t_g+T_7+(4\vartheta+4)d$. By
\lemmaref{lemma:rec} we have that all non-faulty nodes are observed in state
\rec\ during $[t^+,t_\join]$. By Statement~(v) of \lemmaref{lemma:clean}, no
node in $W$ resets its \join\ flags after time $t^+$ before it switches to state
\prop, returning to the basic cycle. Thus, any node from $W$ will switch to
state \join\ before time $t_g+T_7+(4\vartheta+4)d$ and switch to \prop as soon
as it memorizes all non-faulty nodes in state \join. Denote by $t_p\in
(t_g+T_7/\vartheta+4d,t_g+T_7+(4\vartheta+5)d)$ the minimal time when a node
from $W$ switches from \join\ to \prop. Certainly, nodes in $W$ do not switch
from \wake\ to \rdy\ during $(t_p,t_p+2d)$ and therefore also not reset their
\join\ flags before time $t_p+3d$. As nodes in $W$ reset their \prop\ and \acc\
flags upon switching to state \join, some node in $W$ must memorize $n-2f\geq
f+1$ non-faulty nodes in state \join\ at time $t_p$. According to Statement~(ii)
of \lemmaref{lemma:clean}, these nodes must have switched to state \join\ at or
after time $t_\join$. Hence, all nodes in $W$ will memorize them in state \join\
by time $t_p+d$ and thus have switched to state \join. Hence, all nodes in $W$
will switch to state \prop\ before time $t_p+2d$ and subsequently to state \acc\
before time $t_p+3d$, i.e. $t_p\leq t_g+T_7+(4\vartheta+5)d$ is a
quasi-stabilization point.

\textbf{Case 2a:} Assume $t_s< t^+-\Delta_s$ and $< f+1$ nodes in $W$ switch
to \srw\ during $[t_s,t_s+\tilde{\delta}_s]$. We then have that
$t^+\sr{eq:T_2}{>}t_s+T_1+T_5$. According to \lemmaref{lemma:rec}, any
node in $W$ that does not switch to state \srw\ is observed in state \rec\
during $[t_s+T_1+T_5,t_\join]$. Thus, any node in $W$ will observe at least
$n-2f\geq f+1$ nodes from $W$ in state \rec\ during
$[t_s+T_1+T_5,t_\join]$. As nodes in $W$ reset their
\prop\ flags when switching to state \rdy\ and
\begin{equation*}
t_s+T_1+T_5\stackrel{(\ref{eq:T_2},\ref{eq:T_3})}{\leq}
t_s+\frac{T_2+T_3}{\vartheta}-(\vartheta+2)T_1-(2\vartheta+4)d,
\end{equation*}
a node from $W$ switching to state \srw\ at or after time $t_s$ cannot switch to
\prop\ via states \wake\ and \rdy\ before time $t_s+T_1+T_5+(2\vartheta+1)d$.
Any node in $W$ switching to state \srw\ during $[t_s,t_s+\tilde{\delta}_s]$
will observe itself in state \wake\ before time
\begin{equation*}
t_s+\tilde{\delta}_s+2d\stackrel{(\ref{eq:T_1},\ref{eq:T_2})}{\leq}
t_s+\Delta_s-(2\vartheta+2)d\leq t^+-(2\vartheta+2)d.
\end{equation*}
By \lemmaref{lemma:switch}, $tr(\rec,\join)$ cannot be satisfied at any node in
$W$ until time $\min\{t_s+\Delta_s,t^+\}$. Thus, we have that no node from $W$
switches from \rdy\ to \join\ during $[t_s,t_\join)$ by definition of $t_\join$
and any node in $W$ that observes itself in states \rdy\ and \sus\ will switch
to state \rec\ once $(2\vartheta d,\sus)$ expires. In summary, any node in $W$
switching to state \srw\ at some time in $[t_s,t_s+\tilde{\delta}_s]$ will
switch from \wake\ to \rec\ or from \emph{unsuspect} to \sus\ by time
$t_s+\Delta_s-(2\vartheta+2)d$, and in the latter case it cannot leave state
\rdy\ before switching to state \rec\ due to $tr(\rdy,\rec)$ being satisfied. As
the latter happens before time $t_s+\Delta_s-d<t_\join-d$, all nodes in $W$ are
observed in state \rec\ during $[t_s+\Delta_s,t_\join]$. From here we can argue
analogously to the first case, i.e., there exists a quasi-stabilization point
$t_p\leq t_g+T_7+(4\vartheta+5)d$.

\textbf{Case 2b:} Assume $t_s< t^+-\Delta_s$ and $\geq f+1$ nodes
in $W$ switch to \srw\ during $[t_s,t_s+\tilde{\delta}_s]$. By
Statements~(ii) and~(iv) of \lemmaref{lemma:clean}, no node from $W$ resets its
\srw\ flags at or after time $t_s\geq t_g+(1+1/\vartheta)T_1$. Hence, by
Statement~(i) of the lemma, all nodes in $W$ switch to \act\ during
$(t_s,t_s+\tilde{\delta}_s+d)$. Between $T_6/\vartheta$ and $T_6+d$ time later
$T_6$ will expire. We have that
\begin{equation*}
t_s+\frac{T_6}{\vartheta}<t^+-\Delta_s+\frac{T_6}{\vartheta}
\sr{eq:T_7}{\leq}t_g+\frac{T_7}{\vartheta}+4d.
\end{equation*}
Thus, according to \lemmaref{lemma:switch}, $t_\join>t_s+T_6/\vartheta$. On the
other hand, at the latest once $T_6$ expires, $tr(\rec,\join)$ holds at every
node.

By time
\begin{equation*}
t_s+\frac{T_6}{\vartheta}\sr{eq:T_6}{\geq}
t_s+\tilde{\delta}_s-\left(1-\frac{1}{\vartheta}\right)T_1+T_2+2d,
\end{equation*}
the nodes in $W$ that switched to state \srw\ observe themselves in state \rdy\
because of timeouts or are in state \rec. By Statement~(v) of
\lemmaref{lemma:clean}, after this time no node in $W$ resets its \join\ flags
again before it runs through the basic cycle again and switches to state \rdy.

Hence, all nodes in $W$ will switch to states \join\ or \prop\ until time
\begin{eqnarray*}
&&\max\left\{t_s+\tilde{\delta}_s-\left(1+\frac{1}{\vartheta}\right)T_1
+T_2+T_4+2d,t_s+\tilde{\delta}_s+T_6+3d\right\}+d\\
&\stackrel{(\ref{eq:T_3},\ref{eq:T_4})}{=}&
t_s+\left(\vartheta+1-\frac{2}{\vartheta}\right)T_1+T_2+T_4+7d,
\end{eqnarray*}
where we accounted for an additional delay of $d$ due to a possible transition
from \rdy\ to \rec\ just before time $t_s+\tilde{\delta}_s+T_6+2d$ and, if no
node from $W$ switches from \rdy\ to \join, all nodes in $W$ needing to be
observed in state \join\ for a node in $W$ to switch to state \prop.
It follows that a minimal time
$t_p\in(t_s+T_6/\vartheta, t_s+(\vartheta+1-2/\vartheta)T_1+T_2+T_4+7d)$ exists
when a node from $W$ switches to state \prop. Again, we distinguish two cases.

\medskip

\textbf{Case 2b-I:} Assume that some node in $W$ switches from state
\join\ to state \prop\ at time $t_p$. Thus, there must be at
least $n-2f\geq f+1$ non-faulty nodes in state \join\ at time
$t_p-\varepsilon$ (for some arbitrarily small $\varepsilon>0$),
as any \prop\ or \acc\ flag corresponding to a non-faulty node
has been reset at a time $t$ satisfying that the respective node
has not been observed in one of these states during $[t,t_p]$.
Thus, all nodes in $W$ will switch to states \join\ or \prop\ before time
$t_p+d$.
At time $t_p+2d$, they will observe all non-faulty nodes in one of the
states \join, \prop, or \acc, i.e., they switch to state \prop\
before time $t_p+2d$.
Finally, they will observe all non-faulty nodes in states \prop\ or
\acc\ before time $t_p+3d< t_p+T_1/\vartheta$ and switch to state
\acc.
As $t_p$ is minimal, we conclude that all nodes in $W$ switched to state
\acc\ during $(t_p,t_p+3d)$, i.e., $t_p$ is a quasi-stabilization
point.

\medskip

\textbf{Case 2b-II:} Otherwise, some node in $W$ switched from state \rdy\ to
state \prop\ at time $t_p$.
As we have that
\begin{equation*}
t_s+\tilde{\delta}_s+T_6+4d\sr{eq:T_3}{\leq}
t_s-(\vartheta+1)T_1+\frac{T_2+T_3}{\vartheta}-d,
\end{equation*}
$T_6$ is expired at all nodes in $W$ since time $t_p-2d$, i.e., $tr(\rec,\join)$ is
satisfied at all nodes in $W$ since time $t_p-2d$. Hence, all nodes in
$W$ are observed in states \rdy\
or \join\ at time $t_p$, and no node from $W$ may switch to state \rec\ again
or reset its \prop\ flags before switching to \res\ or \acc\ first after time
$t_p$.

Denote by $t_a$ the infimum of times greater than $t_p$ when a node
from $W$ switches to
\acc\ and assume for the moment that no node from $W$ may switch from \prop\ to \rec\
before switching to \acc\ first after time $t_p$. As nodes in $W$ reset their \prop\ flags
upon switching to states \rdy\ or \join, there must be $n-2f\geq f+1$ non-faulty
nodes that switched to state \prop\ during $[t_p,t_a)$ (unless $t_a=\infty$,
which will be ruled out shortly). Thus, all nodes in $W$ leave state \rdy\ before time
$t_a+d$, and are observed in states \prop\ or \join\ before time $t_a+2d$.
Recalling that all nodes in $W$ switch to states \join\ or \prop\ until time
\begin{equation*}
t_s+\left(\vartheta+1-\frac{2}{\vartheta}\right)T_1+T_2+T_4+7d,
\end{equation*}
we get that indeed all nodes in $W$ are observed in one of these states after time
$t_p$ and before time
\begin{equation*}
\min\left\{t_a+2d,
t_s+\left(\vartheta+1-\frac{2}{\vartheta}\right)T_1+T_2+T_4+8d\right\}.
\end{equation*}
Thus, at any node from $W$, $tr(\join,\prop)$ will be satisfied before this time, and it
will be observed in state \prop\ less than $d$ time later. It follows that all
nodes in $W$ switch to state \acc\ before time
\begin{equation*}
t_q+3d:=\min\left\{t_a+3d,
t_s+\left(\vartheta+1-\frac{2}{\vartheta}\right)T_1+T_2+T_4+9d\right\},
\end{equation*}
i.e., $t_q$ is a quasi-stabilization point. As we made the assumption that no
node from $W$ switches from \prop\ to \rec\ before switching to \acc, we need to show
that $T_5$ does no expire at any node from $W$ in state \prop\ until time $t_q+3d$. This
holds true because
\begin{equation*}
t_p+\frac{T_5}{\vartheta}> t_s+\frac{T_5+T_6}{\vartheta}\sr{eq:T_5}{\geq}
t_s+\left(\vartheta+1-\frac{2}{\vartheta}\right)T_1+T_2+T_4+9d
\geq t_q+3d.
\end{equation*}

It remains to check that in all cases, the obtained quasi-synchronisation point
$t_q$ occurs no later than time $t_g+R_1/\vartheta-3d$. In Cases~1 and~2a, we
have that
\begin{equation*}
t_q\leq t_g+T_7+(4\vartheta+5)d\sr{eq:R_1}{\leq}t_g+\frac{R_1}{\vartheta}-3d.
\end{equation*}
In Case~2b, it holds that
\begin{eqnarray*}
t_q&\leq &t_s+\left(\vartheta+1-\frac{2}{\vartheta}\right)T_1+T_2+T_4+9d\\
&\leq &
t^+-\Delta_s+\left(\vartheta+1-\frac{2}{\vartheta}\right)T_1+T_2+T_4+9d\\
&\sr{eq:R_1}{\leq} & t_g+\frac{R_1}{\vartheta}-3d.
\end{eqnarray*}
We conclude that indeed all nodes in $W$ switch to \acc\ within a window of less
than $3d$ time before at any node in $W$, $R_1$ expires and it leaves state \res,
concluding the proof.
\end{proof}

Finally, putting together our main theorems and
\lemmaref{lemma:constraints}, we deduce that the system will
stabilize from an arbitrary initial state provided that a subset of $n-f$
nodes remains coherent for a sufficiently large period of time.
\begin{corollary}\label{coro:stabilization} 
Suppose that $\vartheta<\vartheta_{\max}\approx 1.247$ as given in
\lemmaref{lemma:constraints}. Let $W\subseteq V$, where $|W|\geq n-f$, and
define for any $k\in \N$
\begin{equation*}
T(k):=(k+2)(\vartheta(R_2+3d)+8(1-\lambda)R_2+d)+R_1/\vartheta.
\end{equation*}
Then, for any $k\in \N$, the proposed algorithm is a $(W,W^2)$-stabilizing
pulse synchronization protocol with skew $2d$ and accuracy bounds
$(T_2+T_3)/\vartheta-2d$ and $T_2+T_4+7d$ stabilizing within time
$T(k)$ with probability at least $1-1/2^{k(n-f)}$.
It is feasible to pick timeouts such that $T(k)\in \BO(kn)$ and
$T_2+T_4+7d\in \BO(1)$.
\end{corollary}
\begin{proof}
The satisfiability of \conditionref{cond:timeout_bounds} with $T(k)\in \BO(kn)$
and $T_2+T_4+7d\in \BO(1)$ follows from \lemmaref{lemma:constraints}. Assume
that $t^+$ is sufficiently large for $[t^- +T(k)+2d,t^+]$ to be non-empty, as
otherwise nothing is to show. By definition, $W$ will be coherent during
$[t_c^-,t^+]$, with $t_c^- = t^-+\vartheta(R_2+3d)+8(1-\lambda)R_2+d$. According
to \theoremref{theorem:resync}, there will be some good $W$-resynchronization
point $t_g\in [t_c^-,t_c^-+(k+1)(\vartheta(R_2+3d)+8(1-\lambda)R_2+d)]$ with
probability at least $1-1/2^{k(n-f)}$. If this is the case,
\theoremref{theorem:stabilization} shows that there is a $W$-stabilization point
$t\in [t_g,t^- +T(k)]$. Applying \theoremref{theorem:stability} inductively, we
derive that the algorithm is a $(W,E)$-stabilizing pulse synchronization
protocol with the bounds as stated in the corollary that stabilizes within time
$T(k)$ with probability at least $1-1/2^{k(n-f)}$.
\end{proof}

\section{Generalizations and Extensions}\label{sec:generalizations}

\subsection{Synchronization Despite Faulty Channels}
\theoremref{theorem:stabilization} and our notion of coherency require that all
involved nodes are connected by correct channels only. However, it is desirable
that non-faulty nodes synchronize even if they are not connected by correct
channels. To capture this, the notions of coherency and stability can be
generalized as follows.
\begin{definition}[Weak Coherency]
We call the set $C\subseteq V$ \emph{weakly coherent} during $[t^-,t^+]$, iff for
any node $i\in C$ there is a subset $C'\subseteq C$ that contains $i$, has size
$n-f$, and is coherent during $[t^-,t^+]$. 
\end{definition}
In particular, if there are in total at most $f$ nodes that are faulty
or have faulty outgoing channels, then the set of non-faulty nodes is
(after some amount of time) weakly coherent. 
\begin{corollary}\label{coro:weak_stabilization} 
For each $k\in\N$ let $T'(k) :=
     T(k)-((\vartheta(R_2+3d)+8(1-\lambda)R_2+d))$, where $T(k)$ is
     defined as in \corollaryref{coro:stabilization}.
Suppose the subset of nodes $C\subseteq V$ is weakly coherent during
     the time interval $[t^-,t^+] \supseteq [t^-
     +T'(k)+T_2+T_4+8d,t^+] \neq \emptyset$.
Then, with probability at least $1-(f+1)/2^{k(n-f)}$, there is a
     $C$-quasi-stabilization point $t\leq t^- + T'(k) + T_2+T_4+5d$
     such that the system is weakly $C$-coherent during $[t,t^+]$.
\end{corollary}
\begin{proof}
By the definition of weak coherency, every node in $C$ is in some
     coherent set $C'\subseteq C$ of size $n-f$.
Hence, for any such $C'$ it holds that we can cover all nodes in $C$
     by at most $1+|V \setminus C'| \leq f+1$ coherent sets
     $C_1,\ldots,C_{f+1}\subseteq C$.
By \corollaryref{coro:stabilization} and the union bound, with
     probability at least $1-(f+1)/2^{k(n-f)}$, for each of these sets
     there will be at least one stabilization point during
     $[t^-,t^-+T'(k)-(T_2+T_4+5d)]$.
Assuming that this is indeed true, denote by $t_{i_0}\in
     [t^-,t^-+T'(k)-(T_2+T_4+5d)]$ the time    
\begin{equation*}
\max_{i\in \{1,\ldots,f+1\}}\{\max \{t\leq t^-+T'(k)-(T_2+T_4+5d)\,|\,t\mbox{ is a
$C_i$-stabilization point}\}\},
\end{equation*}
where $i_0\in \{1,\ldots,f+1\}$ is an index for which the first maximum is
attained and $t_{i_0}$ is the respective maximal time, i.e., $t_{i_0}$ is a
$C_{i_0}$-stabilization point.

Define $t_{i_0}'\in (t_{i_0},t^-+T'(k)]$ to be minimal such that it is another
$C_{i_0}$-stabilization point. Such a time must exist by
\theoremref{theorem:stability}. Since the theorem also states that no node from
$C_{i_0}$ switches to state \acc\ during $[t_{i_0}+2d,t_{i_0}')$ and $C_i\cap
C_{i_0}\neq \emptyset$, there can be no $C_i$-stabilization point during
$(t_{i_0}+2d,t_{i_0}'-2d)$ for any $i\in \{1,\ldots,f+1\}$. Applying the theorem
once more, we see that there are also no $C_i$-stabilization points during
$(t_{i_0}'+2d,t_{i_0}'+(T_2+T_3)/\vartheta)-2d$ for any $i\in
\{1,\ldots,f+1\}$. On the other hand, the maximality of $t_{i_0}$ implies that
every $C_i$ had a stabilization point by time $t_{i_0}$. Applying
\theoremref{theorem:stability} to the latest stabilization point until time
$t_{i_0}$ for each $C_i$, we see that it must have another stabilization
point before time $t_{i_0}+T_2+T_4+5d$. We have that
\begin{equation*}
\frac{2(T_2+T_3)}{\vartheta}-2d
\sr{eq:T_2}{>}\frac{T_2+T_3+T_5}{\vartheta}\sr{eq:T_5}{>}T_2+T_4+5d,
\end{equation*}
i.e., all $C_i$ have stabilization points within a short time interval of
$(t_{i_0}'-2d,t_{i_0}'+2d)$. Arguing analogously about the previous
stabilization points of the sets $C_i$ (which exist because $t_{i_0}$ is
maximal), we infer that all $C_i$ had their previous stabilization point
during $(t_{i_0}-2d,t_{i_0}+2d)$.

Now suppose $t_a$ is the minimal time in $(t_{i_0}'-2d,t_{i_0}'+2d)$ when a node
from $C$ switches to \acc\ and this node is in set $C_i$ for some $i\in
\{1,\ldots,f+1\}$. As usual, there must be at least $f+1$ non-faulty nodes from
$C_i$ in state \prop\ at time $t_a$ and by time $t_a+d$ all nodes from $C_i$
will be in either of the states \prop\ or \acc. As $|C_i\cap C_j|\geq f+1$ for
any $j\in \{1,\ldots,f+1\}$ all nodes in $C_j$ will observe $f+1$ nodes in state
\prop\ at times in $(t_a,t_a+2d)$. We have that $t_a\geq
t_{i_0}+(T_2+T_3)/\vartheta-2d$ according to \theoremref{theorem:stability}. As
no nodes switched to state \acc\ during $(t_{i_0}+2d,t_a)$ and none of them
switch to state \rec\ (cf.~\theoremref{theorem:stability}), it follows from the
Inequality
\begin{equation*}
(T_2+T_3)/\vartheta-4d\sr{eq:T_3}{>}T_2+2T_1\sr{eq:T_1}{>}T_2+5d
\end{equation*}
that all nodes from $C_j$ observe themselves in one of the states \rdy\ or
\prop\ at time $t_a$. Hence, they will switch from \rdy\ to \prop\ if they
still are in \rdy\ before time $t_a+2d$. Less than $d$ time later, all nodes in
$C_j$ will memorize $C_j$ in state \prop\ and therefore switch to \acc\
if not done so yet. Since $j$ was arbitrary, it follows that $t_a$ is a
$C$-quasi-stabilization point.
\end{proof}


\begin{corollary}\label{coro:weak_stability}
Suppose $C$ is weakly coherent during $[t^-,t^+]$ and $t\in
[t^-,t^+-(T_2+T_4+8d)]$ is a $C$-quasi-stabilization point. Then
\begin{itemize}
  \item [(i)] all nodes from $C$ switch to \acc\ exactly once within $[t,t+3d)$
  and
  \item [(ii)] there will be a $C$-quasi-stabilization point
  $t'\in[t+(T_2+T_3)/\vartheta,t+T_2+T_4+5d)$ satisfying that no
  nodes switch to \acc\ in the time interval $[t+3d,t')$
  \item [(iii)] and each node $i$'s, $i\in W$, state of the basic cycle
  (\figureref{fig:main_simple}) is metastability-free during $[t+4d,t'+4d)$
\end{itemize}
\end{corollary}
\begin{proof}
Analogously to the proofs of \theoremref{theorem:stability} and
\corollaryref{coro:weak_stabilization}.
\end{proof}
We point out that one cannot get stronger results by the proposed technique.
Even if there are merely $f+1$ failing channels, this can e.g.\ effectively
render a node faulty (as it may never see $n-f$ nodes in states \prop\ or \acc)
or exclude the existence of a coherent set of size $n-f$ (if the channels
connect $f+1$ disjoint pairs of nodes, there can be no subset of $n-f$ nodes
whose induced subgraph contains correct channels only). Stronger resilience to
channel faults would necessitate to propagate information over several hops in a
fault-tolerant manner, imposing larger bounds on timeouts and weaker
synchronization guarantees.

Combination of \corollaryref{coro:weak_stabilization} and
     \corollaryref{coro:weak_stability} finally yields:

\begin{corollary}\label{cor:final_weak}
Suppose that $\vartheta < \vartheta_{\max} \approx 1.247$ as given in
     \lemmaref{lemma:constraints}.
Let $C \subseteq V$ be such that, for each $i\in C$, there is a set
     $C_i \subseteq C$ with $|C_i| = n-f$, and let $E =\bigcup_{i\in
     C} C_i^2$.
Then the proposed algorithm is a $(C,E)$-stabilizing pulse
     synchronization protocol with skew $3d$ and accuracy bounds
     $(T_2+T_3)/\vartheta-3d$ and $T_2+T_4+8d$ stabilizing within time
     $T(k)+T_2+T_4+5d$ with probability at
     least $1-(f+1)/2^{k(n-f)}$, for any $k\in\N$.
\end{corollary}
\begin{proof}
Analogously to the proof of \corollaryref{coro:stabilization}
\end{proof}

\subsection{Late Joining and Fast Recovery}
An important aspect of combining self-stabilization with Byzantine
fault-tolerance is that the system can remain operational when facing a limited
number of transient faults. If the affected components stabilize quickly enough,
this can prevent future faults from causing system failure. In an environment
where transient faults occur according to a random distribution that is not too
far from being uniform (i.e., one deals not primarily with bursts), the mean
time until failure is therefore determined by the time it takes to recover from
transient faults. Thus, it is of significant interest that a node that starts
functioning according to the specifications again synchronizes as fast as possible to an
existing subset of correct nodes making a quasi-stabilization point. Moreover, it is of interest that a node
that has been shut down temporarily, e.g.\ for maintenance, can join the
operational system again quickly.

In the presented form, the algorithm suffers from the drawback that a node in
state \rec\ may be caught there until the next good resynchronization point.
Since Byzantine faults of a certain pattern may deterministically delay this for
$\Omega(n)$ time, we would like to modify the algorithm in a way ensuring that a
non-faulty node can synchronize to others more quickly if a
quasi-stabilization point is reached.

This can be done in a simple manner. Whenever a node switches to state \none, it
stays until a new timeout $(R_1,\none)$ expires. When switching to \none, it
switches also to \pass, resets its \join\ and \srw\ flags, and repeats to reset
its \srw\ flags whenever a timeout of $(\vartheta-1)(\vartheta+2)T_1+\vartheta
5d$ expires. Thus, it will not switch to state \act\ because of outdated
information. On the other hand, it will not miss the next occurrence of a
set $C$, that are weakly coherent since a $C$-quasi-stabilization
point, switching to state \srw\ within a time window of
$(1-1/\vartheta)(\vartheta+2)T_1+\vartheta 5d$, as it will reset its \srw\ flags
at most once in this window, whereas $|C|\geq n-f\geq 2f$. Subsequently, it
will switch to state \join\ at an appropriate time to enter the basic cycle
again at the occurrence of the next $C$-stabilization point. Since nodes refrain
from leaving state \none\ for a constant period of time only, this way
stabilization time in face of severe failures can still be kept linear, while in
a stable system, nodes recovering from faults or joining late stabilize in
constant time.
\begin{corollary}
The pulse synchronization routine can be modified such that it retains
     all shown properties, $\hat{E}_3$ increases by a constant factor,
     and it holds that, for any node $i$ in $V$, if there is a $C$-quasi-stabilization point
     at some time $t < t^-$, so that $C$ is weakly coherent during
     $[t,t^+]$, and $(C\cup\{i\})$-coherent during $[t^-,t^+]$, then
     there exists a $(C\cup\{i\})$-quasi-stabilization point at some time $t'
     \leq t^-+\BO(1)$, so that $(C\cup\{i\})$ is weakly coherent
     during $[t',t^+]$.
\end{corollary}
\begin{proof}[Proof Sketch.]
Essentially, the fact that $n-f$ nodes continue to execute the basic
     cycle narrows down the possibilities in the proof of
     \theoremref{theorem:stabilization} to Case 2b-II, where the
     threshold for leaving state \join\ will be achieved close to the
     next $C$-stabilization point due to the involved threshold
     conditions.
Since the nodes in $C$ execute the basic cycle, they are not affected
     by the re-synchronisation subroutine at all.
Thus, $v$ stabilizes independently of this subroutine provided that it
     resets its \join\ and \srw\ flags in an appropriate fashion.
We explained above how this is done and why a consistent reset of the
     \srw\ flags is achieved.
The \join\ flags are not an issue since at most $n-|C|\leq f$ channels
     can attain state \join.
As a node switches to state \none\ again in constant time whenever it
     leaves, the node will stabilize in constant time provided that
     there is a $C$-quasi-stabilization point from where on $C$ is
     weakly coherent until time $t^+$.
On the other hand, we can easily adapt the re-synchronisation
     subroutine, \lemmaref{lemma:good}, \theoremref{theorem:resync},
     and \conditionref{cond:timeout_bounds} to allow for the
     additional time nodes are non-responsive with respect to the
     re-synchronisation subroutine, increasing $\hat{E}_3$ by a
     constant factor only.
\end{proof}

\subsection{Stronger Adversary}
So far, our analysis considered a fixed set $C$ of coherent (or weakly coherent)
nodes. But what happens if whether a node becomes faulty or not is not
determined upfront, but depends on the execution? Phrased differently, does the
algorithm still stabilize quickly with a large probability if an adversary may
``corrupt'' up to $f$ nodes, but may decide on its choices as time progresses,
fully aware of what happened so far? Since we operate in a system where all
operations take positive time, it might even be the case that a node might fail
just when it is about to perform a certain state transition, and would not have
done so if the execution had proceeded differently. Due to the way we use
randomization, this however makes little difference for the stabilization
properties of the algorithm.
\begin{corollary}
Suppose at every time $t$, an adversary has full knowledge of the state of the
system up to and including time $t$, and it might decide on in total up to $f$
nodes (or all channels originating from a node) becoming faulty at arbitrary
times. If it picks a node at time $t$, it fully controls its actions after and
including time $t$. Furthermore, it controls delays and clock drifts of
non-faulty components within the system specifications, and it initializes the
system in an arbitrary state at time $0$. For any $k\in \N$, define
\begin{equation*}
t_k:=2(k+2)(\vartheta(R_2+3d)
+8(1-\lambda)R_2+d)+R_1/\vartheta+T_2+T_4+5d.
\end{equation*}
Then the set of all non-faulty nodes have reached a
quasi-stabilization point by time $T(k)$ from where on they are weakly
coherent, with probability at least
\begin{equation*}
1-(f+1)e^{-k(n-f)/2}.
\end{equation*}
\end{corollary}
\begin{proof}
We need to show that \theoremref{theorem:resync} holds for the modified time
interval $[t,t+(k+2)\hat{E}_3]$ with the modified probability of at least $1-e^{-k(n-f)/2}$.
If this is the case, we can proceed as in
Corollaries~\ref{coro:weak_stabilization} and~\ref{coro:weak_stability}.

We start to track the execution from time $0$. Whenever a node switches to state
\init\ at a good time, the adversary must corrupt it in order to prevent
subsequent deterministic stabilization. In the proof of
\theoremref{theorem:resync}, we showed that for any non-faulty node, there are
at least $k+1$ different times until $2\vartheta (k+2)\hat{E}_3$ when it
switches to \init\ that have an independently by $1/2$ lower bounded probability
to be good. Since \lemmaref{lemma:good} holds for \emph{any} execution where we
have at most $f$ faults, the adversary corrupting some node at time $t$ affects
the current and future trials of that node only, while the statement still holds
true for the non-corrupted nodes. Thus, the probability that the adversary may
prevent the system from stabilizing until time $t_k$ is upper bounded by the
probability that $(k+1)(n-f)$ independent and unbiased coin flips show $f$ or
less times tail. Chernoff's bound states for the random variable $X$ counting
the number of tails in this random experiment that for any $\delta\in (0,1)$,
\begin{equation*}
P[X<(1-\delta)\E[X]]
<\left(\frac{e^{-\delta}}{(1-\delta)^{1-\delta}}\right)^{\E[X]}
<e^{-\delta \E[X]}.
\end{equation*}
Inserting $\delta=k/(k+1)$ and $\E[X]=(k+1)(n-f)/2$, we see that the probability
that
\begin{equation*}
P[X\leq f]\leq P[X<(n-f)/2]<e^{-k(n-f)/2},
\end{equation*}
as claimed.
\end{proof}

\section{Implementation Issues}
\label{sec:implementation}

In this section, we briefly survey some core aspects of the VLSI
     implementation of the pulse synchronization algorithm, which is
     currently being developed.
Thereby we focus on the three major building blocks: (1) asynchronous
     state machines, (2) memory flags with thresholds and (3) watchdog
     timers.

The pulse synchronization algorithm at every node consists of several
simple state machines that execute asynchronously and concurrently.
There are several types of conditions that can trigger state transitions:
\begin{itemize}
\item[(i)] The state machines of a certain number ($1$, $\geq f+1$, or $\geq
n-f$) of remote nodes reached some particular state, indicated by memory flags.
\item[(ii)] Some local state machine reached a particular state.
\item[(iii)] A watchdog timer expires.
\end{itemize}
These conditions may also be combined (using AND or OR). 

We will employ standard Huffman-type asynchronous state machines~\cite{Myers01}
for implementing our state machines, as they fit nicely to the $\Theta$-Model
already used in \darts.\footnote{The $\Theta$-Model assumes that we can enforce a
certain ratio between slowest and fastest end-to-end delay along critical
signaling paths.} Analyzing the transition conditions of all the five state
machines (Figures~\ref{fig:main}, \ref{fig:extended} and \ref{fig:resync}) of a
single node reveals that we need to communicate six different states (\rec,
\acc, \join, \prop, \srw\ and ``other'') of the core state machine
(\figureref{fig:main}) and two states each (\supp, \none\ and \init,
\emph{wait}) for the two state machines making up the resynchronization
algorithm from every node to every node. There are several possibilities for
implementing this communication. For example, both a simple high-speed serial
protocol and a parallel five bit bundled data bus with a strobe signal are
viable alternatives, each offering different trade-offs between implementation
complexity, speed, area consumption, etc.

We note, however, that any method for communicating states is complicated by the
fact that state occupancy times may be very short in an asynchronous state
machine: Reaching a state must always be faithfully conveyed to all remote nodes
even if it is almost immediately left again. In addition, the core state machine
may undergo various sequences of state transitions, implying that we cannot use
a state encoding where only a single bit changes between successive states. Care
must hence be taken in order not to trigger hazardous intermediate state
occupancies at the receiver when communicating some multi-bit state change. Both
problems can be handled using suitable bounded delay conditions.

\subsubsection*{Remote Memory Flags and Thresholds}

\figureref{fig:memflags} shows the principle of implementing remote memory
flags, which are the basic mechanism required for type (i) state transition
conditions at node $i$. For every remote node $j$, it consists of a hazard-free
demultiplexer that decodes the communicated state of node~$j$'s state machines,
a resettable memory flag per state that remembers whether node~$j$ has ever
reached the respective state since the most recent flag reset, and optionally a
threshold module that combines the corresponding flag outputs for all remote
nodes. Note that every memory flag is implemented as a (resettable) Muller C
Gate\footnote{A Muller C Gate retains its current output value when its inputs
are different, and sets its output to the common input otherwise.} here, but
could also be built by using a flip-flop.

\begin{figure}[ht!]
 \centering
 \includegraphics[width=.6\textwidth]{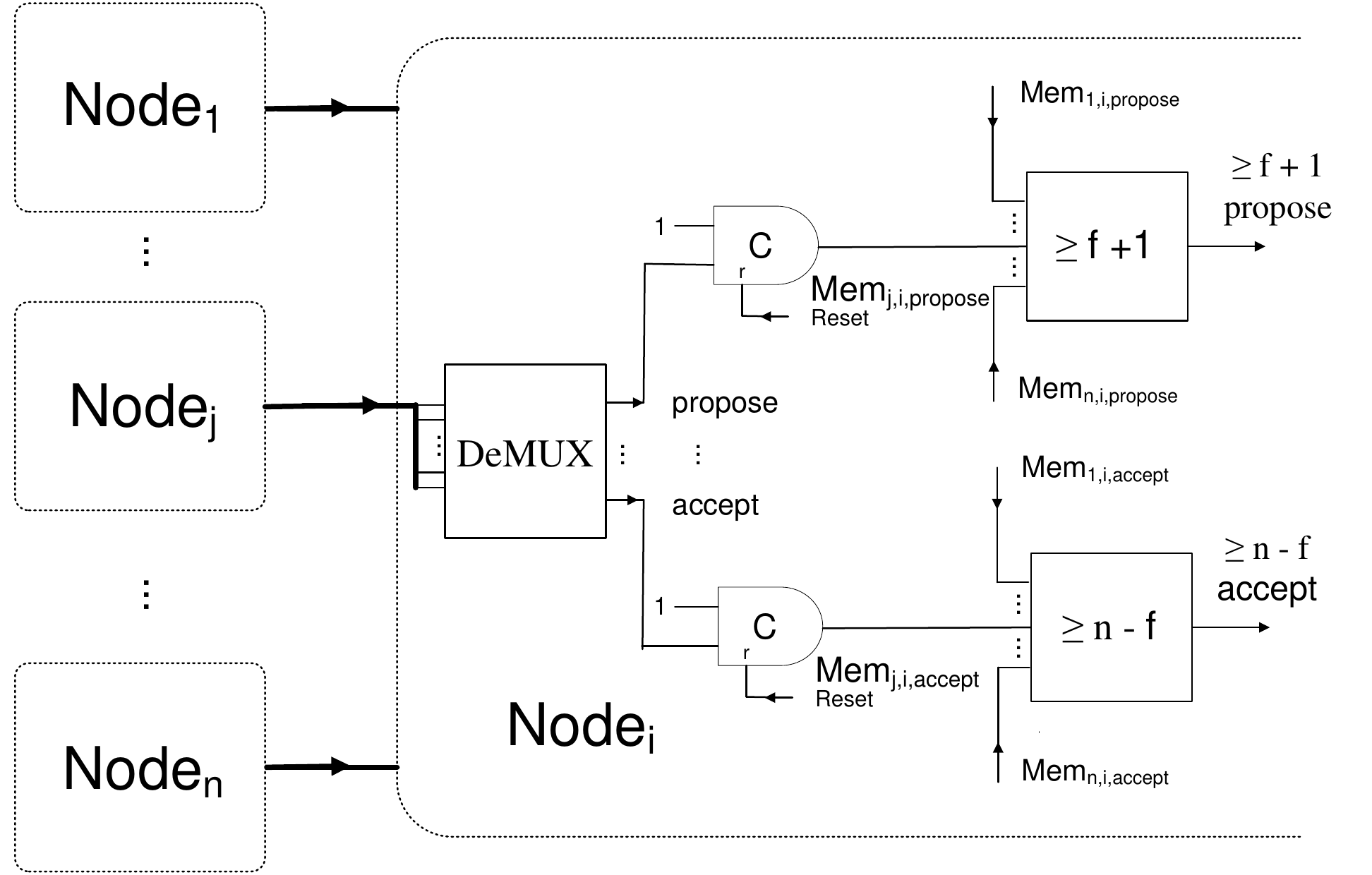}
 \caption{Implementation principle of remote memory flags and thresholds.}
 \label{fig:memflags}
\end{figure}

Implementing local state input transition conditions (ii) is pretty
much straightforward, as one simply needs to incorporate (single) 
state signals from local state machines here. Note that every transition
condition comprises the node observing itself in a particular state, which also
falls into this category. To avoid metastable upsets in the asynchronous state
machine (see below), it may be necessary to add memory flags for local signals
as well.

\subsubsection*{Watchdog Timers}

Our implementation of the watchdog timers, which are required for
     realizing state transition conditions (iii), will rest upon a
     single local clock generator (we will use a simple ring
     oscillator, i.e., a single inverter with feedback and a
     prescaler) per node that drives all watchdog timers, instead of a
     crystal oscillator, because of the possibility to integrate it
     on-chip.
However, the oscillator frequency of such ring oscillators vary
     heavily with operating conditions, in  particular with supply
     voltage and temperature, as well as with process conditions.
The resulting (two-sided) clock drift $\xi$ (with respect to supply
     voltage, temperature and process variation) is typically in the
     range of $7\%$ to $9\%$ for uncompensated ring oscillators and
     can be lowered down to $1\%$ to $2\%$ by proper compensation
     techniques \cite{SAA06}.
The two-sided clock drifts map to $\vartheta = (1+\xi)/(1-\xi)$ bounds
     of $1.15$ to $1.19$ and $1.02$ to $1.04$, respectively.
Recalling from \lemmaref{lemma:constraints} that $\vartheta_{\max}
     \approx 1.247$, one sees that both uncompensated and compensated
     ring oscillators are suitable for implementation of the pulse
     synchronization protocol's watchdog timers.
However, care must be taken when the protocol is used to stabilize
     \darts: to compensate a typical drift of $15\%$ of \darts\
     clocks, one must ensure that $\vartheta$ is
     smaller than roughly $1.064$ (cf.~\sectionref{sec:coupling}).
Thus, here, only compensated ring oscillators are sufficiently
     accurate.
Note, however, that these are conservative bounds, assuming that the
     synchronization protocol and \darts\ drift into different
     directions. Considering that a large share of the drift in both
     systems is due to variations in temperature, it seems
     reasonable to assume that, in the long term, both drift into 
     the same direction.

As shown in \figureref{fig:watchdog}, every watchdog timer consists of
     a resettable up-counter and a timeout register, which holds the
     timeout value.
A comparator compares the counter value and the timeout register after
     every clock tick, and raises a stable expiration output signal if
     the counter value is greater or equal to the register value.
The asynchronous reset of the counter, which also resets the timeout
     output signal, is used to re-trigger the watchdog.

\begin{figure}[ht!]
 \centering
 \includegraphics[width=.6\textwidth]{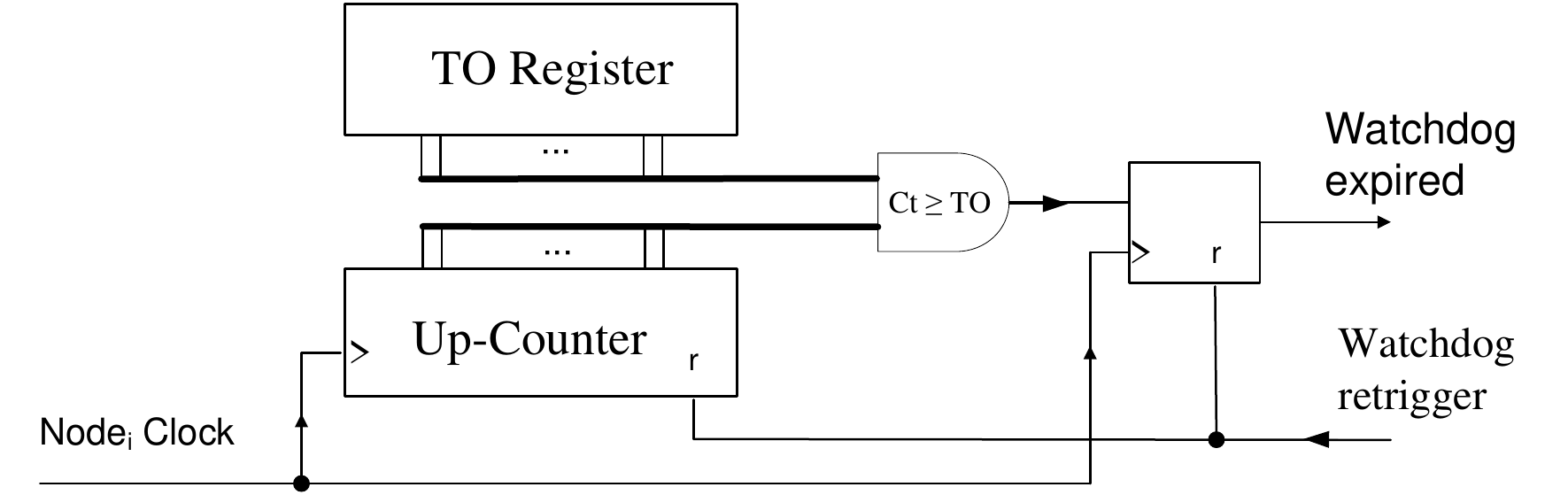}
 \caption{Implementation principle of watchdog timers.}
 \label{fig:watchdog}
\end{figure}

As for the watchdog timer with random timeout $R_3$ in the resynchronization 
algorithm, the simplest implementation would load a uniformly distributed
random value into the timeout register whenever the watchdog is re-triggered. 
Depending on the implementation technology, such random values can be 
generated either via true random sources (thermal noise) or pseudo-random 
sources (LFSRs) clocked by another ring oscillator. If we 
could guarantee that the content of the timeout and the random source can, 
by no means, read or probed somehow by anybody, such an implementation satisfies
the model requirements.\footnote{Note that in practice this is a reasonable
assumption, as even the node itself does not access this value except for
checking whether the timer expired and the computational power of the system is
very limited.} Alternatively, one could use random sampling per clock tick,
which avoids storing the future timeout value and also converges to
uniformly distributed timeouts for sufficiently large values of $R_3$.

\subsubsection*{Combined State Transition Conditions}

Combining different state transition conditions (i)--(iii) via AND/OR requires
some  care, since an asynchronous state machine requires stable input signals in
order not to become metastable during its state transition. Combining several
conditions (i) does not cause any problems, since the memory flags  ensure that
all outputs are stable. Non-stable signals, like ``$T_1$ AND $< n-f $ \acc''
require sampling via a flip-flop clocked by a stable signal. For example, the
status of $< n-f = \neg \geq n-f$ is sampled when the signal reporting
expiration of $T_1$ is issued. Similarly, it might happen that conditions
requiring conflicting state transitions are satisfied at the same time, e.g.,
$(T_2,\acc)$ might expire simultaneously with the threshold of ``$\geq f+1$
\rec\ or \acc'' being reached. 

Obviously, both of the above situations could create a metastable upset, either
of the sampling flip-flop, or directly of the register(s) holding the node's
state. Fortunately, \theoremref{theorem:stability} revealed that this can happen
during stabilization only. In regular operation, e.g.\ the critical threshold of
$\geq n-f$ \acc\ is always reached before $T_1$ expires. Thus, the former is
acceptable, as metastable upsets occur rarely and increase convergence time
only. Moreover, to further decrease the probability of a metastable upset that
might affect stabilization time, it is perfectly feasible to insert a
synchronizer or an elastic pipeline after the sampling flip-flop for capturing
metastability~\cite{FFS09:ASYNC09}. This additional precaution merely increases
the latency by a constant delay, which due to being restricted to the pulse
synchronization component will not adversely affect the final precision and
accuracy of the stabilized \darts clocks.

\section{Coupling of DARTS and Pulse Synchronization
Algorithm}\label{sec:coupling}

In this section, we describe how the self-stabilizing pulse synchronization
protocol could be coupled with \darts clocks. As this requires certain
implementation details, we also sketch some ideas that might be used in a
prototype implementation. The joint system provides a high-precision
self-stabilizing Byzantine fault-tolerant clocking system for multi-synchronous
GALS.

The coupling between the pulse synchronization protocol and
\darts clocks involves two directions: 
\begin{enumerate}
\item The pulse synchronization protocol primarily monitors
the operation of the \darts clocks. As long as 
\darts ticks are generated correctly, it must
not interfere with the \darts tick generation rules at all. 

\item If \darts clocks become inconsistent w.r.t.\ the
behavior of the pulse synchronization protocol, the latter must
interfere with the regular \darts tick generation, possibly
up to resetting \darts clocks.
\end{enumerate}

To assist the reader, we first provide a very brief overview of the
original \darts and its implementation.

\subsection{DARTS Overview}
\label{sec:dartsintro}

\darts clocks (called TG-Algs in the sequel) are instances of a simple
synchronizer \cite{WS09:DC} for the \mbox{$\Theta$-Model} 
based on consistent broadcasting \cite{ST87}. They generate
ticks \tick{0}, \tick{1}, \tick{2}, \dots\ approximately
simultaneously at all correct nodes. Since actual \darts ticks
are just binary clock signal transitions, which cannot carry tick 
numbers, the original algorithm had to be modified significantly in
order to be implementable in asynchronous digital logic.
\figureref{fig:TGAlg} shows a schematic of a single TG-Alg
for a 5-node system.

\begin{figure}[bth]
\centering
  \includegraphics[width=0.9\columnwidth]{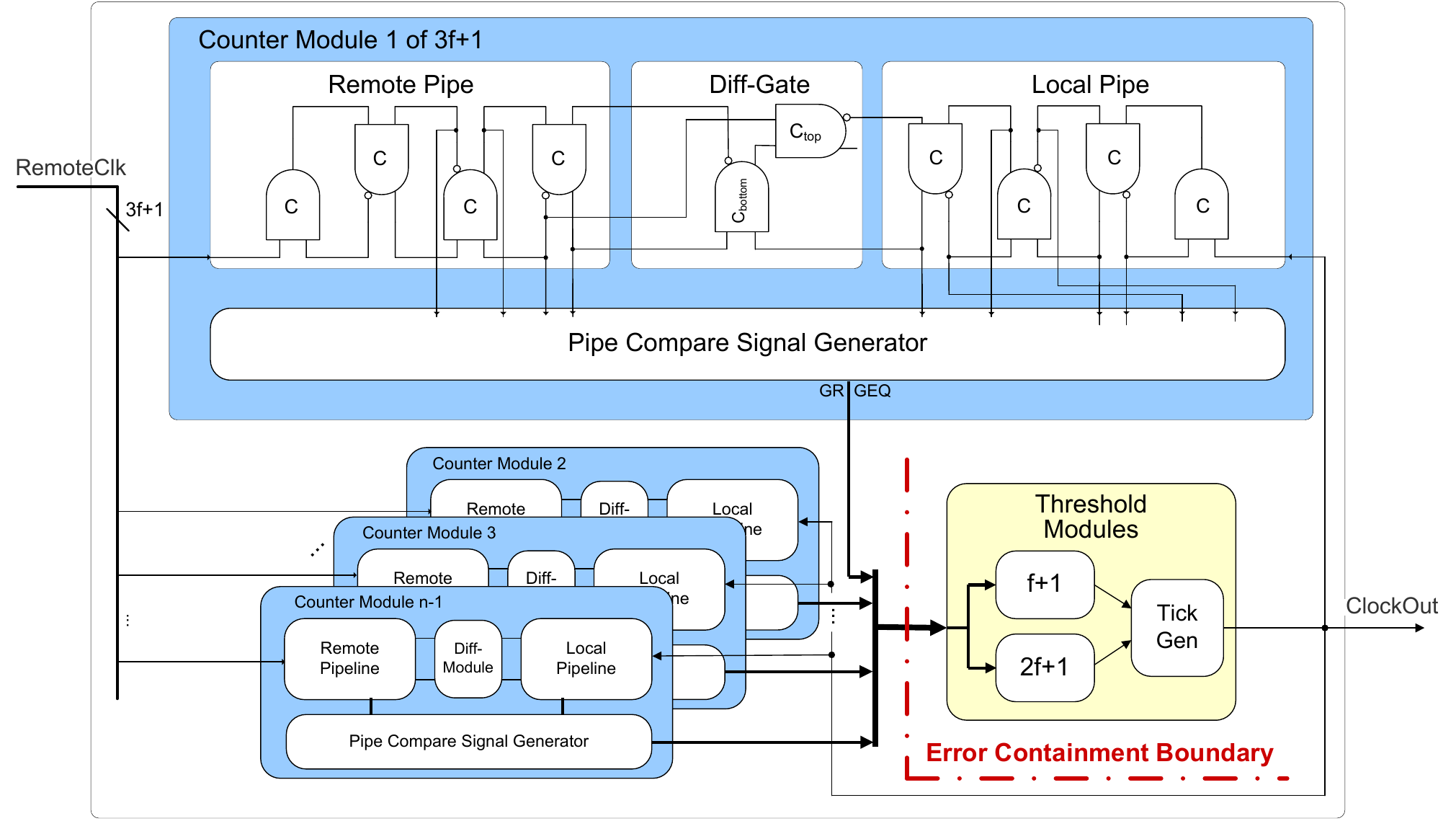}
  \caption{Schematic of \darts TG-Alg Implementation}\label{fig:TGAlg}
\end{figure}

Key components of a TG-Alg are counter modules, one per remote TG-Alg, which
just count the difference between the number of ticks generated locally and
remotely. They are implemented using a pair of elastic pipelines
\cite{Sutherland89}, which implement FIFO buffers for signal transitions.
Matching ticks in both pipelines, which are obviously irrelevant for the
difference, are removed by the connecting Diff-Gate. The status ($>0$, $\geq 0$)
of all counter modules is fed into two threshold modules, whose outputs trigger
the generation of the next local tick. A detailed discussion of the
implementation can be found in \cite{FFSK06:DFT}.

The correctness proof and performance analysis in
\cite{FS10:TR,FSFK06:edcc,FDS10:edcc} revealed that correct TG-Algs indeed
generate synchronized clock ticks, in the presence of up to $f$ Byzantine faulty
TG-Algs in a system with $n\geq 3f+2$ nodes: For any two correct nodes $p$, $q$,
the number of clock ticks generated by $p$ and $q$ by time $t$ do not differ by
more than a (very small) constant $\pi$, and the frequency of any correct clock
(and thus the maximum drift $\rho$) is within a certain range. In addition,
expressions (in the order of $\pi$) for the maximum size of the elastic
pipelines in the counter modules were established, which guarantee overflow-free
operation. Experiments with both FPGA and ASIC prototype implementations
demonstrated that \darts clocks indeed offer close to perfect synchronization and
very reliable operation.

Nevertheless, as already mentioned, (almost) simultaneous start-up of all
TG-Algs and at most $f$ failures during the whole life-time of the system are
mandatory preconditions for these results to hold. \darts neither supports late
joining or recovery of TG-Algs, nor recovery from more than $f$ failures.

\subsection{Required Extensions for Coupling DARTS 
and Pulse Synchronization}
\label{sec:dartsextensions}

The major obstacle for supporting late joining of TG-Algs, removing
spuriously generated ticks in the pipelines etc.\
are the anonymous clock ticks used in \darts: Since they are just signal
transitions on a single wire, they cannot encode any information except their
occurrence time. The most important extension of \darts is hence to add an
additional bundled data wire to the clock signal, which carries 1 bit of data.
This way, single ticks can be \emph{marked} with a 1, distinguishing them from
ordinary non-marked ticks that carry a 0. 

We will actually mark every $T$-th \darts tick, for some suitably
chosen $T$. Such a marked tick $kT$, $k\geq 0$, is to be 
understood as the start of the $(k+1)$-st \emph{\darts round}, which
consists of the marked \darts tick $kT$ and $T-1$ subsequent unmarked ticks
$kT+1,kT+2,\ldots,(k+1)T-1$; the marked tick $(k+1)T$ starts the
next \darts round. Note that the resulting \darts ticks can be interpreted 
as a discrete, bounded clock operating modulo $T$. As \darts
rounds at any two correct TG-Algs are synchronous, marked ticks
must always match in the pipelines of every counter, i.e., the Diff-Gate
must always remove pairs of matching marked (or non-marked) ticks and can hence
detect and remove any inconsistency.

The actual coupling between the instance of the pulse synchronization
     protocol and the \darts clock running at node $i$ is accomplished
     by means of two signals, namely, $\darts_i$ and $\pulse_i$:
\begin{itemize}
\item $\darts_i$ reports \darts rounds to the pulse synchronization protocol.
The rising edge of the $\darts_i$ signal, which may trigger a switch 
from \rdy\ to \prop; is issued
when the \darts clock of node~$i$ generates tick $kT-X$, for some 
fixed $X<T$. The falling edge of $\darts_i$ reports the occurrence 
of the marked tick $kT$.
\item $\pulse_i$ reports the generation of a pulse to the \darts clock.
Its rising edge is issued on the transition to \acc, and its falling
edge signals the expiration of a fixed timeout $T_y$ that is reset at the time
the rising edge is transmitted.
\end{itemize}

The basic idea underlying the coupling of the pulse synchronization protocol and
\darts is to align marked ticks and pulses as follows: If the system operates
normally, every correct node $i$ first reaches some \darts tick $kT-X$ and issues
$\darts_i=1$. Next, a pulse is generated at node~$i$ by the the pulse
synchronization protocol, which thus sets $\pulse_i=1$. Subsequently, the \darts
marked tick $kT$ occurs, which is signaled by $\darts_i=0$. Finally, the pulse
timeout $T_y$ and hence $\pulse_i=0$ occurs. Normal operation thus expects that
the \darts marked tick (= the falling edge of $\darts_i$) occurs within the time
window where $\pulse_i$ is 1. Provided that the timeout used for generating this
window\footnote{We remark that it is vital not to rely on the \darts clock here.}
is chosen sufficiently large, namely, $\vartheta\rho(\pi+2d+1)$, this
interleaving can indeed be guaranteed in normal operation.

As long as this is the case, we just let \darts generate its ticks using its
standard rules. Should a \darts clock fail, however, such that $\pulse_i$ and
$\darts_i$ are not properly interleaved, then we will force marking the next
\darts tick (and possibly resetting the TG-Alg, if needed) upon the falling edge
of $\pulse_i$. \darts ticks and pulses (as well as marked \darts ticks at
different nodes) will hence only be re-aligned in case of errors or
desynchronization: As long as \darts clocks work correctly, any two correct
TG-Algs will mark tick $kT$ within the \darts synchronization precision.

Provided that $X$ and $T_y$ are suitably chosen, it is not difficult to prove
that the joint system consisting of pulse synchronization protocol and \darts
clocks will stabilize: After some unstable period, the pulse synchronization
algorithm will stabilize, which we have proved to happen independently of the
(non-)operation of \darts clocks. When the pulse synchronization protocol
eventually starts to generate synchronous pulses, the \darts clocks will start to
recover in a guided (synchronized) manner. When all correct \darts clocks are
eventually synchronized to within the intrinsic \darts precision, the system will
perpetually ensure the above interleaving at all correct nodes.

Some additional observations:
\begin{enumerate}
\item[(1)] Since the \darts precision is typically considerably smaller than
the worst case pulse synchronization precision, the underlying \darts clocks
may be viewed as a ``precision amplifier'' (as well as a clock multiplier,
see \sectionref{sec:dartsextensions}). 

\item[(2)] There is no need to specify properties possibly achieved by \darts clocks 
during their own recovery. We only require that they eventually reach full
synchronization in the presence of synchronous pulses at all correct nodes. 
In practice, \darts clocks will typically also gradually
improve their synchronization precision during this interval.

\item[(3)] Although the pulse synchronization algorithm stabilizes even
when the \darts clocks behave arbitrary, it is nevertheless the case that
it achieves better pulse synchronization precision when the \darts clocks are
fully synchronized.

\item[(4)] One might ask why we
did not just use the $k$-th rising edge of $\pulse_i$ to mark the very next \darts 
tick generated by the TG-Alg at node~$i$. This simple solution has several major 
drawbacks. First, the pulse synchronization precision is
typically worse than the synchronization precision
provided by \darts. Thus, every pulse
would result in a temporary deterioration of the \darts synchronization quality.
Second, marked ticks are not necessarily generated 
exactly every $T$ \darts ticks. And last but not least, since \darts clocks and pulse
synchronization execute completely asynchronously, marking \darts ticks at 
pulse occurrence times would create the potential of metastability every 
$kT$ \darts ticks, even if there is no failure at all.
\end{enumerate}

\subsection{DARTS \texorpdfstring{$\Rightarrow$}{=>} pulse synchronization}
\label{sec:dartstopulse}

To implement this part of the coupling, every \darts clock signals the upcoming
occurrence of marked tick $kT$ to its local instance of the pulse
synchronization protocol. This is accomplished by the rising edge of $\darts_i$,
the dedicated \emph{\darts signal}, which is generated upon \darts tick $KT-X$.
If all correct nodes happen to do this within some time window when they are
(w.r.t.\ the pulse algorithm) in state \rdy\ with $T_3 < T_4$ already expired,
all correct nodes will switch to state \prop\ within $\pi$ time.\footnote{In
contrast to the model we employed for our analysis, we neglect the local
signaling delay here, as it is smaller than the time to generate a single tick.}
Subsequently, they will all switch to state \acc\ within $d$ time. To make sure
that indeed all correct nodes are in state \rdy\ with $T_3$ already expired, up
to small additional terms of $\BO(d)$, we must choose the minimal duration of a
\darts round to be larger than $T_2+T_3+4d$, while $(T_2+T_4)/\vartheta$ is to
exceed its maximal duration. Assuming that $\rho<1.15$,
\lemmaref{lemma:constraints} shows that this is feasible up to $\vartheta\approx
1.064$, which is clearly within the reach of ring
oscillators~\cite{SAA06}.\footnote{This is true regardless of the additional
term of $\BO(d)$, as the bound is derived from an asymptotic statement.}

\subsection{Pulse synchronization \texorpdfstring{$\Rightarrow$}{=>} DARTS}
\label{sec:pulsetodarts}

This part of the coupling between \darts and the pulse synchronization
protocol requires two mechanisms:

\begin{enumerate}
\item[(1)] A way to force a marked \darts tick at node~$i$ upon the 
occurrence of the falling edge of $\pulse_i$, provided that
no marked tick (i.e., the falling edge of $\darts_i$)
has been generated while $\pulse_i$ was 1. This may also include 
recovering from a complete stall of the \darts tick generation.
\item[(2)] A way to recover accurate \darts synchronization after forcing marked 
ticks, which may also include the need to get rid of any information from
the preceding unstable period.
\end{enumerate}

To achieve (1), we use a simple asynchronous circuit that supervises
the interleaving of $\pulse_i$ and $\darts_i$, and generates a ``force 
marking'' signal if $\darts_i$ does not occur in time. 
Note that this device can be built in a way that entirely avoids metastability 
in case of normal operation. In an unstable period, however, it may happen that
force marking occurs exactly at the time when \darts generates its marked 
tick, so a metastable upset or two very close marked ticks (a forced and 
a regularly generated one) are possible.

There are several variants for implementing the forced marking itself, including
the simplest variant of just resetting the TG-Alg in order to generate marked
tick $0$. The need for possibly resetting a TG-Alg originates from the fact that
stateful TG-Alg components may deadlock due to earlier failures. For example, a
deadlocked pipe will never propagate ticks from its input to the Diff-Gate.
Unfortunately, resetting TG-Algs complicates \darts recovery considerably: If a
TG-Alg reset would also reset the remote pipes of its counters, it might lose
``fresh'' marked ticks generated by remote TG-Algs. Hence, remote pipes should
only be reset when the \emph{remote} node is reset. However, since a remote
node might never observe a discrepancy between \darts rounds and pulses, this
approach might end up in the pipe not being reset at all. This is problematic
as it might effectively render the node faulty despite all its components being
operational. Luckily, we may utilize the fact that solving (2) under the
assumption that correct pipes are not deadlocked yields a trivial means to
distinguish a locked pipe from an operational one: If the Diff-Gate cannot
remove any ticks within a certain time interval after a (correct) pulse, the
pipe must have deadlocked and can safely be reset. At the next pulse, all pipes
will have recovered from previous deadlocks and a solution to (2) assuming
deadlock-free pipes will succeed.

\medskip

To explain how we achieve (2), we start with the observation that our way of
marking every $T$-th tick implies that, for any two correct \darts TG-Algs, it
will always be a marked tick $kT$ from a remote node that is matched by the
local marked tick $kT$ in every counter of \figureref{fig:TGAlg} when the
Diff-Gate removes it. That is, if ever a marked tick is matching a non-marked
tick in a counter, ticks have been lost or spuriously generated somewhere, or
local and remote node are severely out of synchrony.

Assume for the moment that we could generate exactly one marked tick at every
correct node, we made sure that no such tick is in the system before this
happens, and that we have elastic pipelines of infinite size. The following
simple strategy would eventually establish matching \darts ticks: Whenever a
Diff-Gate encounters a marked tick in one pipe matched by an ordinary tick
in the other, it removes the ordinary tick only. At the pipe level, this rule
implies that whatever the state of the pipes was before the marked ticks were
generated, they will be cleared before the matching pair of marked ticks is
removed. Since all \darts tick generation rules ensure that no TG-Alg
generates any tick $kT+1$, $kT+2$, \dots\ based on information from 
the previous \darts round $k-1$ (consisting of \darts ticks up to $kT-1$) 
all counter states will be valid as soon as the matching marked ticks $kT$ 
have been removed. As \darts essentially generates ticks based on comparing the
number of locally and remotely generated ticks, this is enough to ensure
stabilization of the \darts system; full \darts precision will be achieved quickly
because nodes ``catch up'', i.e., generate tick numbers that at least $f+1$
correct \darts clocks already reached, faster than ``new ticks'', i.e., ones that
no correct node generated yet, may occur.


The issue of finite-size pipes is (largely) solved by the pulse synchronization
protocol: Pulses and hence marked ticks are generated close to each other, in a
time window of at most $2d+T_y\in \BO(d)$ (provided that $T_y$ is not
unnecessarily large). Hence, apart from implementation issues, pipes that can
accommodate all ticks that may be generated within this time window are
sufficient for not losing any valid \darts tick.

In reality, however, we cannot always expect the ``single marked tick'' setting
described above: Elastic pipelines may initially be populated with arbitrarily
many marked ticks from the unstable period. We must hence make sure that all
these marked ticks (and the white ticks in between) are eventually removed, and
that we do not generate new marked ticks close to each other. The pulse
synchronization protocol will ensure that forced ticks are separated by $T$
\darts ticks, and our implementations of (1) and (2) will ensure with a large
probability that a forced marked tick will not be generated close to a marked
tick generated regularly by \darts. Under these conditions, it is a relatively
easy task to clear all superfluous marked ticks between pulses.

For example, we could asynchronously reset the whole data flip-flop chain that
holds the markings of the ticks (not the ticks themselves!) currently in a pipe
shortly after the rising flank of $\pulse_i$. Enlarging $X$ and $T_y$ slightly,
we can be sure that all TG-Algs will remove spurious markings from their pipes
before any marked tick associated with the respective correct pulse is
generated. Although this could generate metastability in the Diff-Gate, namely,
when the tick at the head of the pipe is a marked tick and the Diff-Gate is
about to act when the pipe is reset upon arrival of a new marked tick arrives,
this cannot happen during normal operation.

\newpage

\pagestyle{empty}

\bibliographystyle{abbrv}
\bibliography{pulse}

\end{document}